\documentclass{amsart}

\usepackage{amsmath,amsthm}
\usepackage{color}
\usepackage{amsxtra}
\usepackage{amscd}
\usepackage{amsfonts}
\usepackage{amssymb}
\usepackage{eucal}
\usepackage{graphicx}

\usepackage{xr-hyper}
\usepackage[
pdftex,
bookmarks=false,
colorlinks=true,
debug=true,
pdfnewwindow=true]{hyperref}

\newtheorem{theorem}{Theorem}[section]
\newtheorem{cor}[theorem]{Corollary}
\newtheorem{lem}[theorem]{Lemma}
\newtheorem{prop}[theorem]{Proposition}

\newtheorem{thm}[theorem]{Theorem}

\theoremstyle{definition}

\newtheorem{nott}[theorem]{Notation}

\theoremstyle{remark}
\newtheorem{rem}{Remark}[section]

\newcommand{\nc}{\newcommand}
\nc\pa{{\partial}}
\nc\beq{\begin{equation}}
\nc\eeq{\end{equation}}
\nc\ii{{\rm i}}
\nc\ol{\overline} \nc\ul{\underline} \nc\wt{\widetilde}
\nc{\z}{\zeta}
\nc\ve{{\varepsilon}}
\nc\mbn{{\mathsf{m}}_{\bn}}
\nc\tmbn{{\tilde{\mathsf{m}}}_{\bn}}

\nc{\BZ}{{\mathbb Z}} 
\nc{\NN}{{\mathbb N}} 
\nc{\BC}{{\mathbb C}}
\nc{\BR}{{\mathbb R}}
\nc{\BP}{{\mathbb P}}
\nc{\TT}{{\mathbb T}} 
\nc{\CP}{{\mathbb {CP}}}
\nc{\bh}{{\mathfrak h}} 
\nc{\er}{{\mathfrak e}} 
\nc{\fr}{{\mathfrak f}} 
\nc{\rz}{{\rm z}}
\nc{\tz}{{\tilde{\rm z}}}
\nc{\A}{{\mathcal A}} 
\nc{\bA}{{\mathbb A}} 
\nc{\bL}{{\mathbb L}} 
\nc{\be}{\mathbf{e}}
\nc{\bP}{{\mathbb P}}
\nc\U{{\mathfrak U}}  
\nc\FF{{\mathfrak F}}
\nc\D{{\mathfrak D}}  
\nc\dd{{\mathfrak d}}

\nc{\F}{{\mathcal F}} \nc{\N}{{\mathcal N}} \nc{\Aa}{{\mathcal A}}
\nc{\E}{{\mathcal E}} \nc{\sS}{{\mathbb S}} \nc{\K}{{\mathcal K}}
\nc{\Oo}{{\mathcal O}} \nc{\M}{{\mathcal M}} \nc{\PP}{{\mathcal P}}
\nc{\X}{{\Xi}}

\newcommand{\ep}{\ve}

\newcommand{\bJ}{\bf J}

\newcommand{\mf}{\rm{m}}
\newcommand\n{\mathfrak n}

\newcommand{\gl}{\mathfrak{gl}}
\newcommand{\ssl}{\mathfrak{sl}}
\newcommand{\g}{\mathfrak{g}}

\newcommand{\Hom}{\mathrm{Hom}}
\newcommand{\Res}{\mathrm{Res}}

\newcommand{\Inst}{\mathrm{inst}}
\newcommand{\inst}{\mathrm{inst}}
\newcommand{\KZ}{\mathrm{KZ}}
\newcommand{\Cas}{\mathrm{Cas}}

\newcommand {\CalA} {\mathcal A}

\newcommand {\CalC} {\mathcal C}

\newcommand {\sF} {\tt F}
\newcommand {\se} {\tt e}

\newcommand {\CalD} {\mathcal D}

\newcommand {\CalG} {\mathcal G}
\newcommand {\CalH} {\mathcal H}
\newcommand {\CalI} {\mathcal I}

\newcommand {\CalL} {\mathcal L}
\newcommand {\CalM} {\mathcal M}

\newcommand {\CalN} {\mathcal N}
\newcommand {\CalO} {\mathcal O}

\newcommand {\CalS} {\mathcal S}

\newcommand {\CalU} {\mathcal U}
\newcommand {\CalV} {\mathcal V}
\newcommand {\CalX} {\mathcal X}

\newcommand {\qe} {\mathfrak q}

\nc{\bp}{{\mathbf{p}}} 
\nc{\bq}{{\mathbf{q}}} 
\nc{\ba}{{\mathbf{a}}}
\nc{\bm}{{\mathbf{m}}}
\nc{\bw}{{\mathbf{w}}}
\nc{\bx}{{\mathbf{x}}}
\nc{\bb}{{\mathbf{b}}} 
\nc{\bc}{{\mathbf{c}}} 
\nc{\bC}{{\mathbf{C}}}
\nc{\bR}{{\mathbf{R}}}
\nc{\hh}{{\mathbf{h}}} 
\nc{\xx}{{\mathbf{x}}} 
\nc{\yy}{{\mathbf{y}}}
\nc{\bj}{{\mathbf{j}}} 
\nc{\bU}{{\mathbf{U}}}
\nc{\sd}{\mathsf{d}}
\nc{\bs}{\mathbf{s}}
\nc{\bz}{\mathbf{z}}
\nc{\bn}{\mathbf{n}}

\nc{\iso}{{\stackrel{\sim}{\longrightarrow}}}

\nc{\bnu}{\boldsymbol{\nu}}

\begin{document}

\title[Surface defects in gauge theory and KZ]{Surface defects in gauge theory and KZ equation}

\author{Nikita Nekrasov}

\address{Simons Center for Geometry and Physics, 
Stony Brook University, Stony Brook NY 11794-3636, USA, 
also, on leave of absence from:
CAS Skoltech and IITP RAS, Moscow, Russia}
\email{nikitastring@gmail.com}

\author{Alexander Tsymbaliuk}

\address{Department of Mathematics, 
 Purdue University, West Lafayette, IN 47907-2067, USA}
\email{sashikts@gmail.com}


\begin{abstract}
We study the regular surface defect in the $\Omega$-deformed four-dimensional supersymmetric 
gauge theory with gauge group $SU(N)$ with $2N$ hypermultiplets in fundamental representation.  
We prove its vacuum expectation value obeys the Knizhnik-Zamolodchikov equation for 
the $4$-point conformal block of the $\widehat{{\ssl}}_{N}$-current algebra, originally introduced 
in the context of two-dimensional conformal field theory. The level and the vertex operators 
are determined by the parameters of the $\Omega$-background and the masses of the hypermultiplets; 
the cross-ratio of the $4$ points is determined by the complexified gauge coupling. 
We clarify that in a somewhat subtle way the branching rule is parametrized by the Coulomb moduli. 
This is an example of the BPS/CFT relation. 
\end{abstract}

\maketitle


\section{Introduction}

The rich mathematics of quantum field theory has a remarkable feature of admitting, to
some extent, an analytic continuation in various parameters, such as momenta, spins etc.
This feature is best studied in the examples of two-dimensional conformal field theories,
where one can observe almost with a naked eye that the building blocks of the correlation
functions are analytic in the parameters, such as the central charges, conformal
dimensions, weights, spins and so on, cf.~\cite{Zamolodchikov:1995aa}. 
Some formulae admit analytic continuation in the level $k$ of the current algebra, 
cf.~\cite{KZ}. The analytic continuation offers some glimpses of 
the Langlands duality~\cite{AG} $(k+h^{\vee})(k^{\vee}+h) = 1$, which suggests an
identification of the quantum group parameter $q$ with the modular parameter 
${\exp} (2\pi \ii \tau)$ of some elliptic curve~\cite{EN}. These observations
solidified as soon as the connection between the $S$-duality of four-dimensional 
supersymmetric theories and the modular invariance of two-dimensional conformal field theories
was observed~\cite{Vafa:1994tf}. Localization computations in supersymmetric gauge 
theories~\cite{Nekrasov:2002qd,Losev:2003py,Nekrasov:2003rj,NekBPSCFT} showed that
the correlation functions of selected observables coincide with conformal blocks of some
two-dimensional conformal field theories, or, more generally, are given by the matrix
elements of representations of some infinite-dimensional algebras, such as Kac-Moody,
Virasoro, or their $q$-deformations, albeit extended to the complex domain of parameters,
typically quantized in the two-dimensional setup. In~\cite{Nekrasov:2002qd}, this
phenomenon was attributed to the chiral nature of the tensor field propagating on 
the worldvolume of the fivebranes. The fivebranes ($M5$ branes in $M$-theory and 
$NS5$ branes in the $IIA$ string theory) were used in~\cite{Klemm:1996bj,Witten5} to
engineer, in string theory setup, the supersymmetric systems whose low energy 
is described by ${\CalN}=2$ supersymmetric gauge theories in four dimensions. 
This construction was extended and generalized in~\cite{Gaiotto:2009we}.  
This correspondence, named \emph{the BPS/CFT correspondence} in~\cite{NN2004},  
has been supported by a large class of very detailed examples 
in~\cite{NekBPSCFT,Alday:2009aq, Alday:2010vg}, and more recently 
in~\cite{KT, Kimura:2015rgi, Haouzi:2019jzk, Haouzi:2020yxy}.

Finally, in~\cite{Witten6, Witten:2011zz}, the relation of the quantum group parameter $q$ 
with the elliptic curves has been brought into the familiar context of the relation 
of the ${\CalN}=4$ super-Yang-Mills theory to elliptic curves. Hopefully, 
with this understanding of the analytically continued Chern-Simons theory, 
the (quasi)-modularity conjectures of~\cite{Zth} could be tested. 

In this paper, we shall be studying a particular corner of that theoretical landscape: 
the $SU(N)$ gauge theory with $2N$ fundamental hypermultiplets. 
In the BPS/CFT correspondence, it is associated with a zoo of two-dimensional 
conformal theories living on a $4$-punctured sphere, all related to the
$\widehat{\ssl}_{N}$ current algebra, either directly, or through the
Drinfeld-Sokolov reduction, producing the $W_{N}$-algebra~\cite{ZW}, depending on 
the supersymmetric observables one uses to probe the four-dimensional theory. 
Two observables are of interest for us. First, the supersymmetric partition function 
${\bf Z} = {\bf Z}( {\ba}, {\bm}, {\ep}_{1}, {\ep}_{2}; {\qe})$ on ${\BR}^4$, which is 
a function of the vacuum expectation value 
${\ba}= {\rm diag} \left( a_{1}, \ldots , a_{N} \right)$ of the scalar in the vector multiplet, 
the masses ${\bm} = \left\{ m_{1}, m_{2}, \ldots , m_{2N-1}, m_{2N} \right\}$,
the exponentiated complexified gauge coupling ${\qe} = {\exp} (2\pi \ii \tau)$, 
\begin{equation*}
  {\tau} = \, \frac{\vartheta}{2\pi} +  \frac{4{\pi}{\ii}}{e^{2}}
\end{equation*}
and the parameters ${\bf\ep}=\left( {\ep}_{1}, {\ep}_{2} \right)$ of 
the $\Omega$-deformation. The latter are the equivariant parameters of 
the maximal torus $U(1) \times U(1)$ of the Euclidean rotation group $Spin(4)$. 
In the complex coordinates $(z_{1}, z_{2})$ on ${\BC}^2 \approx {\BR}^4$, 
the rotational symmetry acts by 
  $(z_{1}, z_{2}) \mapsto 
   ( e^{{\ii}t{\ep}_{1}} \, z_{1}\, , \,  e^{{\ii}t{\ep}_{2}} \, z_{2} )$. 
Exchanging ${\ep}_{1} \leftrightarrow {\ep}_{2}$ is part of the $Spin(4)$ Weyl group,
hence it is a symmetry of $\bf Z$. The second observable is the partition function
${\bf\Psi}$ of the regular surface defect which breaks the gauge group down to its 
maximal torus $U(1)^{N-1}$ along the surface, which we shall take to be the $z_2=0$ plane. 
This partition function depends on all the parameters ${\ba}, {\bm}, {\bf\ep}, {\qe}$ 
that the bulk partition function $\bf Z$ depends on and, in addition, it depends on the parameters 
\begin{equation*}
  {\bw} =  (w_{0}:w_{1}: \dots : w_{N-1}) \in {\mathbb{CP}}^{N-1}
\end{equation*}
of a two-dimensional theory the defect supports. 
The physics and mathematics setup of the problem are explained in the Parts IV, V 
of~\cite{NekBPSCFT}, which the reader may consult for motivations and orientation.
However, our exposition is self-contained as a well-posed mathematical problem, 
which we introduce presently. 

Our main result is the proof of a particular case of the BPS/CFT conjecture~\cite{NN2004}: 
the vacuum expectation value $\langle\, {\CalS} \, \rangle$ of the surface defect obeys 
the Knizhnik-Zamolodchikov equation~\cite{KZ}, specifically the equation obeyed by the 
$\left( \widehat{\ssl}_{N} \right)_{k}$ current algebra conformal block
\beq
  {\Phi} =  \Big\langle \, 
    {\bf V}_{1} (0) {\bf V}_{2} ({\qe})
    {\bf V}_{3} (1) {\bf V}_{4} ({\infty}) \, 
  \Big\rangle^{\bf a}
\label{eq:wzwconfblock}
\eeq
with the vertex operators corresponding to irreducible infinite-dimensional representations of $\ssl_{N}$. 
More specifically, the vertex operators at $0$ and $\infty$ correspond to the generic lowest weight 
${\CalV}_{\vec\nu}$ and highest weight ${\tilde{\CalV}}_{\vec{\tilde\nu}}$ Verma modules, 
while the vertex operators at $\qe$ and~$1$ correspond to the so-called \emph{twisted HW-modules} 
${\CalH}_{\mf}^{\vec \mu}, {\tilde{\CalH}}_{\tilde{\mf}}^{\vec{\tilde\mu}}$. The subscripts 
${\vec\nu},{\vec{\tilde\nu}} \in {\BC}^{N-1}$ and ${\mf}, {\tilde \mf} \in {\BC}$ determine 
the values of the Casimir operators, in correspondence with the $2N$ masses $\bm$ and one of the $\Omega$-background 
parameters $\ve_1$.  The superscripts ${\vec\mu}, \vec{\tilde\mu} \in {\BC}^{N-1}$ determine the so-called twists of 
the HW-modules, all defined below, which we express via $\bm, \ve_1$, and the Coulomb parameters~$\ba$. In other words, 
the Coulomb parameters determine the analogue of the ``intermediate spin'', which we indicate by placing a superscript 
$\ba$ in \eqref{eq:wzwconfblock} to label the specific fusion channel. 
We define these representations and the Knizhnik-Zamolodchikov equation~\cite{KZ} below. 

The appearance of the twisted representations is a curious fact 
not visible in the rational conformal field theories. 

   
\subsection*{Acknowledgments}

N.N.\ is grateful to A.~Okounkov, A.~Rosly, and A.~Zamolodchikov for discussions. 
A.T.\ gratefully acknowledges support from the Simons Center for Geometry and Physics 
and is extremely grateful to IHES (Bures-sur-Yvette, France) for the hospitality and wonderful 
working conditions, where some parts of the research for this paper were performed. 
The work of A.T.\ was partially supported by the NSF Grants DMS-$1502497$ and DMS-$2037602$.


\section{Basic setup in four dimensions}\label{sec Setup}

First we introduce the setup of the four-dimensional gauge theory calculation. 


\subsection{Notations}\label{ssec Notations}

We start by reviewing our notations. The reader is invited to consult~\cite{NekBPSCFT} 
for the general orientation. 

\begin{enumerate}

\item[${\bullet}$] 
The parameters of the $\Omega$-deformation: $\ve_1, \ve_2$ -- two complex parameters, 
generating the equivariant cohomology 
  $\mathrm{H}^{\bullet}_{\BC^\times\times \BC^\times}(\mathrm{pt})$.
The \emph{twist} part of the $\Omega$-deformation is $\ve=\ve_1+\ve_2$. 
The torus ${\BC^\times\times \BC^\times}$ is the complexification of 
the maximal torus of the spin cover of the rotation group $Spin(4)$.
We also define
\beq
  {\kappa} = \frac{{\ve}_{2}}{{\ve}_{1}} \, .
\label{eq:kappar}
\eeq

\item[${\bullet}$] 
The Coulomb moduli:
\beq
  \ba =  ( a_{b} )_{b =1}^{N} \equiv \left( a_1,\ldots,a_N \right) \in {\BC}^{N}
\label{eq:coulmod}
\eeq
--  the equivariant parameters of the \emph{color group}, 
in other words these are the generators of $\mathrm{H}^{\bullet}_{(\BC^\times)^N}(\mathrm{pt})$, 
on which the symmetric group $S(N)$ acts by permutations. 

\item[${\bullet}$] 
The masses: 
\beq
  \bm = (m_f)_{f=1}^{2N} \equiv \left( m_1,\ldots,m_{2N} \right) \in {\BC}^{2N}
\label{eq:masses}
\eeq
-- the equivariant parameters of the \emph{flavor group}. 
The symmetric group $S(2N)$ acts on them by permutations. 
The $S(2N)$-invariants are encoded via the polynomial
\begin{equation}\label{polynomial P} 
  P(x)=\prod\limits_{f=1}^{2N} (x-m_{f})\, .
\end{equation}

\item[${\bullet}$]
The splitting of the set of masses into the $N$ ``fundamental'' and $N$ ``anti-fundamental'' ones:
\begin{equation}
  P(x) = P^{+}(x) P^{-}(x) \, , \qquad P^{\pm}(x) = \prod_{f=1}^{N} (x-m_{f}^{\pm})\, .
\end{equation}

\item[${\bullet}$]
\emph{The lattice of equivariant weights} ${\Lambda} \subset {\BC}$ is defined by:
\begin{equation}\label{lattice}
  {\Lambda} := \BZ\cdot \ve_1\oplus \BZ\cdot \ve_2\oplus \bigoplus_{{b}=1}^N \BZ\cdot a_{b} 
               \oplus \bigoplus_{f=1}^{2N} {\BZ}\cdot m_{f}\, . 
\end{equation}
We assume that all the parameters ${\ve}_{1,2}, {\ba}, {\bm}$ are generic, 
up to the overall translation $a_{b} \mapsto a_{b} + s,\ m_{f} \mapsto m_{f} + s$, for $s \in {\BC}$. 
Thus, the rank of $\Lambda$ is at least $3N +1$. 

\medskip
Recall that the bulk theory (subject to noncommutative deformation, 
leading to instanton moduli space being partially compactified to the moduli space
${\mathfrak{M}}_{k,N}$ of charge $k$ rank $N$ framed torsion-free sheaves on ${\mathbb{CP}}^2$)
is invariant under the nonabelian symmetry group $U(2)$ of rotations, preserving the
complex structure of ${\BC}^2 \approx {\BR}^4$. The group $U(N)$ of constant
gauge transformations acts on ${\mathfrak{M}}_{k,N}$ by changing the asymptotics of
instantons at infinity. The Coulomb parameters $\ba$ represent the maximal torus of
$U(N)$; they can be viewed as local coordinates on the 
  ${\rm Spec}\, \mathrm{H}^{\bullet}_{U(N)} (\mathrm{pt}) = {\BC} [{\ba}]^{S(N)}$, 
with $S(N)$ being the Weyl group. 
Likewise, the parameters $({\ve}_{1}, {\ve}_{2})$ are acted by the Weyl group ${\BZ}_{2}$
which acts by permuting ${\ve}_{1} \leftrightarrow {\ve}_{2}$. 
The physical theory has a larger rotation symmetry group $Spin(4)$, whose Weyl group is 
${\BZ}_{2} \times {\BZ}_{2}$ but we don't see the full symmetry in the ${\bf Z}$-function.
The full symmetry is present once ${\bf Z}$ is divided by the so-called $U(1)$-factor,
having to do with decoupling of the $U(1)$-part of gauge group~\cite{Alday:2009aq}. 

Finally, the masses represent the equivariant parameters of the flavor group $SU(2N)$ 
(the physical theory has a larger flavor symmetry group, which we don't see either), hence, 
the Weyl group  $S(2N)$ symmetry making the polynomial $P(x)$ of~(\ref{polynomial P}) the good parameter. 

The surface defect we are going to study in this paper breaks both the gauge group 
$U(N)$ to its maximal torus $U(1)^N$ and the flavor group to its maximal torus $U(1)^{2N}$. 
The group $S(N) \times S(2N)$ acts, therefore, on the space of surface defects. 
In describing the specific bases in the vector space of surface defects, we keep track 
of the ordering of the Coulomb and mass parameters. 

\medskip
\item[${\bullet}$] 
The set of vertices of the Young graph $\PP$ -- the set of all Young diagrams 
($=$ partitions of nonnegative integers) $\{\lambda\}$. Then  
\begin{equation*}
  {\PP}^N = 
  \left\{\, \overline{\lambda} = (\lambda^{(1)},\ldots,\lambda^{(N)})\, \Big| \,
         {\lambda}^{(b)} \in \PP \, \ \mathrm{for}\, \ 1\leq b \leq N \, \right\}\, .
\end{equation*}

\item[${\bullet}$]
For a box $\square=(i,j)$, define its \emph{content} $c(\square)$ by:
\beq
  c(\square):=(i-1)\ve_1+(j-1)\ve_2\, .
\label{eq:cont}
\eeq

\noindent
\item[${\bullet}$]
For $\lambda\in \PP$, define:
\begin{equation}
  \chi_\lambda:=\sum_{\square\in \lambda}  e^{c(\square)}
  \qquad \mathrm{and} \qquad
  \chi^*_\lambda:=\sum_{\square\in \lambda} e^{-c(\square)}\, .
\end{equation}

\noindent
\item[${\bullet}$]
For $\overline{\lambda}\in \PP^N$, define the multiset, i.e., its elements may have
multiplicities, \emph{of tangent weights}, 
  $\{ {\rm w}_t\}_{t\in T_{\overline{\lambda}}} \subset \Lambda$ 
by the character
\beq
  \sum_{t\in T_{\overline{\lambda}}}e^{{\rm w}_{t}} :=
  \sum_{b,c=1}^N \ e^{a_{b} - a_{c}}\, 
  \Big( \, \chi^{*}_{{\lambda}^{({c})}}+e^{\ve}\cdot \chi_{\lambda^{({b})}}-
        (1-e^{\ve_1})(1-e^{\ve_2})\cdot\chi_{\lambda^{({b})}}
        {\chi}^{*}_{{\lambda}^{({c})}}\, \Big)\, .
\label{eq:tangweights}
\eeq
\begin{rem}
The \emph{duality}: 
  $\{ {\rm w}_{t} \}_{t \in T_{\overline{\lambda}}} = \{ \ve - {\rm w}_{t} \}_{t\in T_{\overline{\lambda}}}$ 
is related to the symplectic structure on the instanton moduli space and its completion ${\mathfrak{M}}_{k,N}$. 
\end{rem}

\item[${\bullet}$]
The \emph{pseudo-measure}
  $\mu={\mu}({\ba}, {\bm}, {\ve}_{1}, {\ve}_{2}; {\qe})\colon {\PP}^{N} \to {\BC}$ 
on $\PP^N$ is defined via:
\begin{equation}\label{eq:bulkmu}
\begin{split}
  & {\mu}({\ba}, {\bm}, {\ve}_{1}, {\ve}_{2}; {\qe})\, \vert_{\overline{\lambda}} \, := \, 
  \frac{1}{{\bf Z}^{\inst}}\cdot \left((-1)^N {\qe}\right)^{|{\overline\lambda}|}\cdot
  \frac{\prod\limits_{f=1}^{2N}\prod\limits_{b=1}^N\prod\limits_{\square\in \lambda^{({b})}}\,
    \left( a_{b} + c(\square) -m_f \right)}
   {\prod\limits_{t\in T_{\overline{\lambda}}} {\rm w}_{t}}
     \, =\\
  & \ \ \ \ \ \ \ \ \ \ \ \ \ \ \ \ \ \ \ \ \ \ \ \ \ \ \ \ \ \ \
  \frac{1}{{\bf Z}^{\inst}}\cdot {\qe}^{|{\overline\lambda}|}\cdot
  \frac{\prod\limits_{b=1}^N\prod\limits_{\square\in \lambda^{({b})}}\,
     \, 
    \Big( - P^{-}\left( a_{b} + c({\square}) \right) P^{+}\left( a_{b}+c(\square) \right) \Big)}
   {\prod\limits_{t\in T_{\overline{\lambda}}} {\rm w}_{t}}\, ,
\end{split}   
\end{equation}
where
  $|{\overline\lambda}| = \sum\limits_{b=1}^N |\lambda^{({b})}|$
with
  \[ |\lambda^{({b})}| = \sum_{i} {\lambda}^{({b})}_{i} \] 
denoting the total number of boxes in $\lambda^{({b})}$, and 
  ${\bf Z}^{\inst}={\bf Z}^{\inst}( {\ba},{\bm}, {\ve}_{1},{\ve}_{2}; {\qe})$ 
is the Taylor series in $\qe$ uniquely determined by the normalization
\begin{equation*}
  \sum_{\overline{\lambda}\in \PP^N} \mu\, \vert_{\overline{\lambda}} \, = \, 1\, .
\end{equation*}

\begin{rem}
The restriction ${\rm deg}\, P(x) = 2N$ comes from the convergence of ${\bf Z}^{\inst}$ for generic 
${\ba}, {\bm}, {\ve}_{1,2}$, cf.~\cite{Felder}. When working over the ring ${\BC}[[{\qe}]]$ of formal 
power series in~$\qe$, the restriction on the degree of $P(x)$, i.e., the number of masses, can be dropped.  
\end{rem}

\item[${\bullet}$]
For $\lambda\in \PP$, we call $\blacksquare\in \lambda$ a \emph{corner box}
if $\lambda\backslash \blacksquare \in \PP$ and we call $\square\notin \lambda$
a \emph{growth box} if $\lambda\sqcup \square \in \PP$. We denote by
${\partial}_{+}{\lambda}$ the set of all growth boxes of $\lambda$, and 
by ${\partial}_{-}{\lambda}$ the set of all corner boxes of ${\lambda}$. 
It is easy to check that:
\begin{equation*}
  {\#} {\partial}_{+}{\lambda} \, -\,  {\#}{\partial}_{-}{\lambda} \, = \, 1\, .
\end{equation*}

\item[${\bullet}$]
For $x\in {\BC}$, we define the function $Y(x)$ on ${\PP}^{N}$ as follows: 
its value $Y(x)\, \vert_{\overline{\lambda}}$ on ${\bar\lambda} \in {\PP}^{N}$ is equal to
\begin{equation} \label{eq:yobs}
\begin{split}
  & Y(x)\, \vert_{\overline{\lambda}}\, :=\, \prod_{b=1}^N
  \left((x-a_{b})\prod_{\square\in \lambda^{(b)}}
        \frac{(x-a_{b}-c(\square)-\ve_1)(x-a_{b}-c(\square)-\ve_2)}
             {(x-a_{b}-c(\square))(x-a_{b}-c(\square)-\ve)}\right) = \\
  & \ \ \ \ \ \ \ \ \ \ \ \ \ \ \ \ \ \ \ \ \ \ \ \ \ \ \ \ \ \ \ \ \ \ \ \ \ \
   \ \ \ \ \ \ \ \ \ \ \ \ \ \ \ \ \ \ \ \ \  \, \prod_{b=1}^N
   \frac{\prod\limits_{{\square} \in {\partial}_{+}{\lambda}^{({b})}} 
      \left( x - a_{b} - c ({\square}) \right)}
        {\prod\limits_{{\blacksquare} \in {\partial}_{-}{\lambda}^{({b})}} 
        \left( x - a_{b} - {\ve} - c ({\blacksquare}) \right)}\, ,
\end{split}   
\end{equation}
the second line being obtained from the first one by the simple inspection of the cancelling common factors. 

\item[${\bullet}$]
For $x \in {\BC}$, we define the function ${\CalX}(x)$ on ${\PP}^N$, 
called the \emph{fundamental $qq$-character}, by specifying its value 
${\CalX}(x)\, \vert_{\overline{\lambda}}$ on ${\bar\lambda} \in {\PP}^{N}$ as follows: 
\begin{equation}\label{fundamental qq-char}
  {\CalX} (x)\, \vert_{\overline{\lambda}}\, :=\ 
  Y (x+{\ve})\, \vert_{\overline{\lambda}}\, +\, {\qe}\, \frac{P(x)}{Y(x)\, \vert_{\overline{\lambda}}}\, .
\end{equation}

\item[${\bullet}$]
For a pseudo-measure $\widetilde{\mu}\colon \PP^N\to \BC$ and a function $g\colon \PP^N\to \BC(x)$, 
the \emph{average} $\langle\, g(x)\, \rangle_{\widetilde{\mu}}$ is defined via:
\begin{equation}\label{average}
  \Big\langle g(x) \Big\rangle_{\widetilde{\mu}}\, :=
  \sum_{\overline{\lambda}\in \PP^N} 
  \widetilde{\mu}\,\vert_{\overline{\lambda}}\cdot g(x)\,\vert_{\overline{\lambda}}\, .
\end{equation}

\end{enumerate}


\subsection{Dyson-Schwinger equation}\label{ssec Regularity}

{}The following is the key property of ${\CalX} \colon {\PP}^{N}\to {\BC}(x)$ 
of~(\ref{fundamental qq-char}):

\begin{prop}\label{regularity 1}
The average $\langle\, {\CalX}(x)\, \rangle_{\mu}$ is a regular function of $x$.
\end{prop}

This is the simplest case of the general result on the $qq$-characters as 
established in~\cite{NekBPSCFT}. For completeness of our exposition, an elementary proof 
is presented in Appendix~\ref{Appendix A0}.

   
\subsection{An orbifold version}\label{ssec Orbifold version}

As explained in Part~III of~\cite{NekBPSCFT}, there is a very important 
$\BZ_N$-equivariant counterpart of the above story. It is defined in several steps. 

First, we change the notations: 
\begin{equation*}
  a_{b} \mapsto {\tilde a}_{b}\, ,\quad 
  m_{f}^{\pm} \mapsto {\tilde m}_{f}^{\pm}\, ,\quad
  ({\ve}_{1}, {\ve}_{2}) \mapsto ({\ve}_{1}, {\tilde\ve}_{2})\, ,\quad
  \mathrm{so\ that}\quad 
  \ve\mapsto {\tilde \ve}:=\ve_1+{\tilde \ve}_2\, .
\end{equation*}
Next, we introduce the $\BZ_N$-grading ${\lambda}\mapsto {\mathfrak S}_{\lambda}\in {\BZ}_{N}$ 
of the lattice $\Lambda$ via: 
\begin{equation}
\begin{split}
  & {\lambda} = k_{1} {\ve}_{1} + k_{2} {\tilde\ve}_{2} + 
    \sum_{b} k^{a}_{b}\, {\tilde a}_{b} + 
    \sum_{f} k^{m^{+}}_{f} {\tilde m}_{f}^{+} + \sum_{f} k^{m^{-}}_{f} {\tilde m}_{f}^{-} \, \mapsto \\
  & \ \ \ \ \ \ \ \ \ \  
    {\mathfrak S}_{\lambda} :=  k_{2} + 
                     \sum_{\omega \in {\BZ}_{N}} {\omega} \left( \sum_{b \in A_{\omega}} k^{a}_{b} \, +
                     \sum_{f \in F_{\omega}^{+}} k^{m^{+}}_{f}  +
                     \sum_{f \in F_{\omega}^{-}} k^{m^{-}}_{f}  \right) \, {\rm mod}\, N\, , 
\end{split}
\end{equation}
for some partitions 
\begin{equation*}
  \Big\{ 1, \ldots , N \Big\} \, = 
  \bigsqcup\limits_{\omega\in \BZ_N}  A_{\omega} \, = 
  \bigsqcup\limits_{\omega\in \BZ_N}  F_{\omega}^{\pm}
\end{equation*}
of the sets of the Coulomb moduli and the fundamental/anti-fundamental masses. 
Such $\BZ_N$-grading is also often called an \emph{$N$-coloring}. 
We define: 
\begin{equation*}
  P_{\omega}^{\pm}(x) \, = \prod_{f\in F_{\omega}^{\pm}} (x - {\tilde m}_{f}^{\pm})\, .
\end{equation*}
The following depends on a choice of a section ${\BZ}_{N}\to {\BZ}$. We send 
\beq
  {\BZ}_{N}\ni \omega \, \mapsto\, 0 \leq {\omega} < N\, ,
\label{eq:lift}
\eeq
thus, identifying $\BZ_N$ with $\{0,\ldots,N-1\}$, as a set.{}
An $N$-coloring is called \emph{regular} iff 
\begin{equation*}
  \# A_{\omega} \, =\, \# F_{\omega}^{+}\, =\, \# F_{\omega}^{-} \, =\, 1
  \, \ \mathrm{for\ all}\, \ \omega\in \BZ_N\, .
\end{equation*}
For a regular $N$-coloring, the $\omega$-colored masses are packaged into a degree two polynomial
\beq
  P_{\omega}(x) = P_{\omega}^{+}(x) P_{\omega}^{-} (x) =: 
  (x - {\ve}_{1}{\mu}_{\omega}-{\omega} {\tilde\ve}_{2})^{2} - {\ve}_{1}^{2} {\delta\mu}_{\omega}^{2}\, . 
\label{eq:massparam}
\eeq

\noindent
Also, for a regular $N$-coloring, assuming \eqref{eq:lift}, we set: 
\beq
  {\alpha}_{\omega} \, := - {\omega}{\tilde\kappa} + \frac{1}{\ve_1} \sum_{b \in A_{\omega}} {\tilde a}_{b} \, ,
\label{eq:coulpar}
\eeq
where
\beq
  {\tilde\kappa} = \frac{\kappa}{N} \, .
\label{eq:tkap}
\eeq
We shall also need a few more new notations. 
\begin{enumerate}
\item[${\bullet}$] 
For every ${\omega}\in {\BZ}_N$, define the observable $k_{\omega}\colon {\PP}^{N} \to \BZ_{\geq 0}$ by:
\begin{equation}
  k_{\omega}\,\vert_{\overline{\lambda}} \, :=
  \sum_{b=1}^N \sum_{\square\in \lambda^{(b)}} 
  \delta_{{\mathfrak S}_{{\tilde a}_{b}+ {\tilde c}(\square)}}^{\omega}\, ,
\end{equation}
where ${\tilde c}(i,j) := (i-1){\ve}_{1} + (j-1){\tilde\ve}_{2}$, cf.~\eqref{eq:cont}, 
and $\delta_i^j \equiv {\delta}_{i,j}$ is the Kronecker delta. 

\item[${\bullet}$] 
The \emph{fractional couplings}:
\begin{equation}\label{fractional couplings}
  \overline{\qe} = (\qe_{\omega})_{\omega\in \BZ_N} \equiv (\qe_0,\qe_1,\ldots,\qe_{N-1})\in {\BC}^{N}\, .
\end{equation}

\item[${\bullet}$] 
Given $\overline{\qe}$ of~(\ref{fractional couplings}), define the  observable 
${\underline{\qe}}\colon {\PP}^{N} \to {\BC}$, called the \emph{fractional instanton factor}, 
as follows:
\beq
  {\underline{\qe}}\,\vert_{\overline{\lambda}} \, := 
  \prod\limits_{b=1}^N \prod\limits_{\square\in \lambda^{(b)}} 
  {\qe}_{{\mathfrak S}_{{\tilde a}_{b}+ {\tilde c}(\square)}} = 
  \prod\limits_{\omega\in \BZ_N} {\qe}_{\omega}^{k_{\omega}\vert_{\overline{\lambda}}}\, . 
\label{eq:frcoup}
\eeq

\item[${\bullet}$] 
The pseudo-measure 
  $\mu^{\rm orb} = 
   \mu^{\rm orb}({\tilde\ba},{\tilde\bm},{\ve}_{1},{\tilde\ve}_{2};{\overline{\qe}}) 
   \colon {\PP}^{N} \to {\BC}$ 
on ${\PP}^{N}$ is defined via:
\begin{equation}
\begin{split}
\label{orbifold measure}
  & \mu^{\rm orb} ( {\tilde\ba}, {\tilde\bm} , {\ve}_{1}, {\tilde\ve}_{2}; {\overline{\qe}}) 
  \,\vert_{\overline{\lambda}}\, := \\
  & \frac{{\underline{\qe}}\,\vert_{\overline{\lambda}}}
       {{\bf\Psi}^{\inst}} \cdot
  \frac{\prod\limits_{f=1}^{N} \prod\limits_{b=1}^N \prod\limits_{\square\in \lambda^{({b})}}
        \left( {\tilde a}_{b}+{\tilde c}({\square})-{\tilde m}_{f}^{+}\right)^
             {{\delta}_{{\mathfrak S}_{{\tilde a}_{b}+{\tilde c}({\square})-{\tilde m}_{f}^{+}}}^{0}}
             \left( -{\tilde a}_{b}-{\tilde c}({\square})+{\tilde m}_{f}^{-}\right)^
             {{\delta}_{{\mathfrak S}_{{\tilde a}_{b}+{\tilde c}({\square})-{\tilde m}_{f}^{-}}}^{0}}}
       {\prod\limits_{t\in T_{\overline{\lambda}}} {\rm w}_{t}^{\delta_{{\mathfrak S}_{{\rm w}_{t}}}^{0}}}\, ,
\end{split}
\end{equation}
where the tangent weights $\{{\rm w}_{t}\}_{t\in T_{\overline{\lambda}}}$ 
are defined via~\eqref{eq:tangweights} with the substitution 
  $a_{b} \mapsto {\tilde a}_{b}$, $\ve_2 \mapsto {\tilde \ve}_{2}$, 
and the partition function
  ${\bf\Psi}^{\inst} = 
   {\bf\Psi}^{\inst} ( {\tilde\ba}, {\tilde\bm}, {\ep}_{1}, {\tilde\ep}_{2}; \overline{\qe})$ 
is the formal power series\footnote{One can show that this power series converges 
when all $|{\qe}_{\omega}|< 1$, uniformly on compact sets in the complex domain 
${\tilde a}_{b} - {\tilde a}_{c} + i{\ve}_{1}  + j{\tilde \ve}_{2}  \neq 0$ for all $i,j \geq 1$.} 
in $\left( {\qe}_{0}, \ldots , {\qe}_{N-1} \right)$ uniquely determined by the normalization
\begin{equation}\label{orbifold normalization}
  \sum_{\overline{\lambda}\in \PP^N}\, \mu^{\rm orb}\, \vert_{\overline{\lambda}} \, = \, 1\, .
\end{equation}

\item[${\bullet}$] 
For every ${\omega}\in \BZ_N$, define the ${\BC}(x)$-valued observable 
$Y_{\omega}\colon {\PP}^{N} \to {\BC}(x)$ via:
\begin{equation}\label{colored Y-observable}
\begin{split}
  & Y_{\omega}(x)\, \vert_{\overline{\lambda}}\, := 
  \prod\limits_{b=1}^N \, 
  \left( \left( x - {\tilde a}_{b} \right)^{{\delta}_{{\mathfrak S}_{{\tilde a}_{b}}}^{\omega}} 
  \times \right. \\
  &  \ \ \ \ \ \ \ \ \ \ \ \ \ 
  \left. \prod\limits_{\square\in \lambda^{(b)}} 
    \left(\frac{x-{\tilde a}_{b}-{\tilde c}({\square})-{\ep}_{1}}
               {x-{\tilde a}_b-{\tilde c}({\square})}\right)^
        {\delta_{{\mathfrak S}_{{\tilde a}_{b}+{\tilde c}({\square})}}^{{\omega}}}
   \left(\frac{x-{\tilde a}_{b}-{\tilde c}({\square})-{\tilde\ve}_{2}}
              {x-{\tilde a}_b-{\tilde c}({\square})-{\tilde\ve}}\right)^
              {\delta_{{\mathfrak S}_{{\tilde a}_{b}+{\tilde c}({\square})}}^{{\omega}-1}}\right).
\end{split}
\end{equation}

\item[${\bullet}$] 
For every  ${\omega} \in {\BZ}_{N}$, define the ${\BC}(x)$-valued observable
${\CalX}_{\omega}\colon {\PP}^N \to {\BC}(x)$ via:
\begin{equation}\label{orbifold qq-character}
  {\CalX}_{\omega}(x)\,\vert_{\overline{\lambda}}\, := \,
  Y_{{\omega}+1}(x+{\tilde\ve})\,\vert_{\overline{\lambda}}\, +\, 
  {\qe}_{\omega}\, \frac{P_{\omega}(x)}{Y_{{\omega}}(x)\,\vert_{\overline{\lambda}}}\, .
\end{equation}

\end{enumerate}

   
\subsection{Surface defects}\label{ssec Surface defects}

Consider a map 
\begin{equation}\label{map pi}
  {\pi}_{N}\colon {\PP}^{N} \longrightarrow {\PP}^{N}
\end{equation} 
defined via 
\beq
  \overline{\lambda} = \left( {\lambda}^{(1)}, \ldots , {\lambda}^{(N)} \right) \, \mapsto \, 
  \overline{\Lambda}=\left( {\Lambda}^{(1)}, \ldots , {\Lambda}^{(N)} \right)
\label{eq:pilalam}
\eeq
with
\beq
  {\Lambda}^{({b})}_{i} \ =\ \left[ \frac{{\lambda}_{i}^{({b})} + {b}-1}{N} \right] 
  \, , \qquad b = 1, \ldots , N \, . 
\label{eq:pilalam2}
\eeq
The geometric origin of $\pi_N$ is explained in~\cite{NekBPSCFT}. 
Note that ${\pi}_{1} = \mathrm{Id}_{\PP}$.

Following~\cite{NekBPSCFT}, let us now pass from $\overline{\qe} = ({\qe}_{0}, \dots , {\qe}_{N-1})$ 
of~\eqref{fractional couplings} to another set of variables, namely ${\bw}=(w_{0} : w_{1} : \dots : w_{N-1})$ 
and ${\qe}$ via:
\begin{equation}
  {\qe}_{0} = w_{1}/w_{0}\, , \, {\qe}_{1} = w_{2}/w_{1}\, , \, \ldots\, , \, 
  {\qe}_{N-2} = w_{N-1}/w_{N-2}\, , \, {\qe}_{N-1} = {\qe} w_{0}/w_{N-1} \, , 
\label{eq:qez}
\end{equation}
where the \emph{bulk} coupling $\qe$ is recovered by:
\beq
  {\qe} = {\qe}_{0} {\qe}_{1} \ldots {\qe}_{N-1}\, .
\label{eq:frco}
\eeq
The variables $\bw$ are redundant, in the sense that correlation functions are invariant 
under the simultaneous rescaling of all $w$'s. However, just as the bulk coupling $\qe$ is 
identified below with the cross-ratio of four points on a sphere, thus revealing a connection to 
the $4$-point function in conformal field theory, the variables $w$'s are identified with 
the coordinates of $N$ particles, whose dynamics is described by the partition function ${\bf\Psi}^{\inst}$. 

{}In terms of the $({\bw}, {\qe})$-variables, the instanton factor 
looks as follows (recall that $k_{-1} = k_{N-1}$):
\beq
  \prod_{\omega \in {\BZ}_{N}} {\qe}_{\omega}^{k_{\omega}}\, = \, 
  {\qe}^{k_{N-1}}\, \prod_{\omega=0}^{N-1}\, w_{\omega}^{k_{\omega-1} - k_{\omega}}\, . 
  \label{eq:qfromw}
\eeq
Evoking~\eqref{eq:lift}, we also have an obvious equality
\beq
  \sum_{\omega\in \BZ_N} k_{\omega} \, = \, 
  N k_{N-1} + \sum_{i=1}^{N-1} i (k_{i -1} - k_{i})\, .
\label{eq:kbulk}
\eeq

Using the aforementioned map $\pi_N$, we define the \emph{Surface defect observable}
${\CalS}({\ba},{\bm},{\ve}_{1},{\ve}_{2};{\bw},{\qe})$ in the statistical model 
defined by the pseudo-measure $\mu$ of~\eqref{eq:bulkmu} via:
\beq
  {\CalS}({\ba}, {\bm} , {\ve}_{1}, {\ve}_{2}; {\bw},{\qe})\, \vert_{\overline{\Lambda}} \ := 
  \sum_{{\overline{\lambda}} \in {\pi}^{-1}_{N}({\overline{\Lambda}})} \, 
  \prod_{\omega=0}^{N-1}\, w_{\omega}^{\left(k_{\omega-1}-k_{\omega}\right)\,\vert_{\overline{\lambda}}}\ 
  \frac{{\mu}^{\rm orb} \left({\tilde\ba},{\tilde\bm},{\ve}_{1},{\tilde\ve}_{2};\overline{\qe}\right)
  \vert_{\overline{\lambda}}}{{\mu}({\ba},{\bm},{\ve}_{1},{\ve}_{2};\qe)\,\vert_{\overline{\Lambda}}}\, ,
  \label{eq:surfdef}
\eeq
where, again with \eqref{eq:lift} understood,
\begin{equation}
   \ve_2=N{\tilde\ve}_{2}\, ,\qquad
   a_{b}={\tilde a}_{b}-{\mathfrak S}_{{\tilde a}_{b}}\cdot {\tilde\ve}_{2}\, ,\qquad
   m_{f}^{\pm}={\tilde m}_{f}^{\pm} - {\mathfrak S}_{{\tilde m}_{f}^{\pm}}\cdot {\tilde\ve}_{2} 
    \, \qquad \mathrm{for}\  1\leq b,f\leq N \, .
\label{eq:bulkamve}
\end{equation}
Note that 
\beq
  m_{f}^{\pm} = {\ve}_{1} \left( {\mu}_{f-1} \pm {\delta\mu}_{f-1} \right) \, 
\label{eq:massdeltamass}
\eeq
evoking the notations of~\eqref{eq:massparam}. 
The shifts \eqref{eq:bulkamve} are motivated by the relation between the sheaves on the orbifold 
${\BC} \times {\BC}/{\BZ}_{N}$ and the covering space ${\BC} \times {\BC}$, see~\cite{NikBlowup, Lee:2020hfu}.  
In what follows, we shall not be using the observable~\eqref{eq:surfdef}. 
Instead, we shall work directly with the pseudo-measure ${\mu}^{\rm orb}$.

   
\subsection{The key property of $\CalX_{{\omega}}$}\label{ssec Regularity of colored}

{}The following result \cite{NekBPSCFT} (whose proof is presented in Appendix~\ref{Appendix A0} for
completeness of our exposition) is a simple consequence of Proposition~\ref{regularity 1}:

\begin{prop}\label{regularity 3}
The average 
  $\left\langle\, \CalX_{\omega}(x)\, \right\rangle_{\mu^{\rm orb}}$ 
is a regular function of $x$ for every ${\omega}\in \BZ_N$.
\end{prop}

For a power series $F(x)=\sum_{\ell=-\infty}^\infty F_{\ell} x^{-\ell}$ and $k\in \BZ$, 
let $\left[x^{-k}\right]F(x)$ denote the coefficient~$F_k$. The regularity property 
of Proposition~\ref{regularity 3} implies the following result:
\begin{equation}\label{eq:residue_k}
   \Big\langle\, \left[x^{-k}\right]\CalX_{\omega}(x)\, \Big\rangle_{\mu^{\rm orb}} \, = \, 
   \left[x^{-k}\right] \Big\langle\, \CalX_{\omega}(x)\, \Big\rangle_{\mu^{\rm orb}} \, = \, 0
   \qquad \mathrm{for\ any}\ k>0\ \mathrm{and\ every}\ \omega\in \BZ_N \, .
\end{equation}
The \underline{main point to take home} is that the $k=1$ case of the equation~\eqref{eq:residue_k} 
implies a second-order differential equation on the partition function 
$\bf\Psi^{\inst}$, viewed as a function of ${\qe}_0,\ldots,{\qe}_{N-1}$. 
This differential equation is the subject of the following subsection.


\subsection{The differential operator ${\CalD}^{\rm BPS}$}\label{ssec qq-operator}

To apply~\eqref{eq:residue_k} for $k=1$, we shall first explicitly compute
  $[x^{-1}]\CalX_{\omega}(x)\,\vert_{\overline{\lambda}}$.
For every ${\omega}\in {\BZ}_{N}$, define the observable $c_{\omega, \ba}\colon {\PP}^{N} \to {\BC}$ via:
\begin{equation}\label{c-observable}
  c_{\omega, \ba} \, \vert_{\overline{\lambda}} \, :=\, 
  \frac{ {\ve}_{1} }{2} k_{\omega} \vert_{\overline{\lambda}} \ +\,
  \sum_{b=1}^N \sum_{\square\in \lambda^{(b)}}
  {\delta}_{{\mathfrak S}_{{\tilde a}_{b}+{\tilde c}(\square)}}^{\omega}\cdot 
  \left( {\tilde a}_{b}+{\tilde c}(\square) \right)\, .
\end{equation}
Recalling~(\ref{eq:coulpar},~\ref{eq:bulkamve}), so that in particular 
$a_b=\ve_1 \alpha_{{\mathfrak S}_{{\tilde a}_{b}}}$ and ${\tilde \kappa}={\tilde \ve}_2 / {\ve_1}$, 
we get:  
\begin{multline*}
  Y_{\omega}(x)\,\vert_{\overline{\lambda}} \, = \,
  \left( x - {\ve}_{1} {\alpha}_{\omega} - {\omega}{\tilde\ve}_{2}\right) \, \times\, 
  \prod_{b=1}^N \prod_{\square\in \lambda^{(b)}}
  \left\{\left(1 - \frac{\ve_1}{x} - 
                \frac{\ve_1({\tilde a}_b + {\tilde c} (\square))}{x^2} + O(x^{-3})\right)^
         {\delta_{{\mathfrak S}_{{\tilde a}_{b}+{\tilde c}(\square)}}^{\omega}}\right.\\
  \ \ \ \ \ \ \times \left.\left(1+\frac{\ve_1}{x}+
       \frac{\ve_1({\tilde a}_b+{\tilde c}(\square)+{\tilde \ve})}{x^2}+O(x^{-3})\right)^
       {\delta_{{\mathfrak S}_{{\tilde a}_{b}+{\tilde c}(\square)}}^{\omega-1}}\right\}\, ,
\end{multline*}
which implies:

\begin{lem}\label{laurent 1}
The large $x$ expansion of the observable $Y_{\omega}(x)$ has $x$ as a leading term,
while the next two coefficients are the observables ${\PP}^{N}\to {\BC}$ given explicitly by:
\begin{equation*}
\begin{split}
  & {\ve}_{1}^{-1} \left[x^{0}\right] Y_{\omega}(x) \,  =\, {\rm d}_{\omega} := 
      k_{{\omega}-1}-k_{\omega}  - {\alpha}_{\omega} -{\omega}{\tilde\kappa}  \, , \\
  & {\ve}_{1}^{-2} \left[x^{-1}\right] Y_{\omega}(x) \, =\, 
    \frac{{\rm d}_{\omega}^{2} - ({\alpha}_{\omega} + \omega{\tilde\kappa})^{2}}{2} + 
    {\tilde \kappa} k_{\omega-1} + \frac{c_{{\omega}-1, \ba} - c_{\omega, \ba}}{\ve_1}\, .
\end{split}
\end{equation*}
\end{lem}

\noindent
As an immediate corollary, using the notations of~(\ref{eq:kappar},~\ref{eq:massparam},~\ref{eq:coulpar}), we obtain:

\begin{prop}\label{laurent 2} 
The observable $\left[x^{-1}\right] \CalX_{\omega}(x) \colon {\PP}^{N} \to {\BC}$ is explicitly given by: 
\begin{equation}\label{key coefficient}
\begin{split}
  {\ve}_{1}^{-2} \left[x^{-1}\right]\CalX_{\omega}(x) \, = & \,
  \frac{\left(c_{\omega,\ba}-c_{{\omega}+1,\ba} \right) - {\qe}_{\omega} \left(c_{{\omega}-1,\ba}-c_{\omega,\ba}\right)}
       {\ve_{1}} \, + \\
  &  {\tilde\kappa} \left(k_{\omega} - {\qe}_{\omega} k_{{\omega}-1}\right)
     + {\qe}_{\omega} \left( \left( {\rm d}_{\omega}  +  {\mu}_{\omega} + {\omega} {\tilde\kappa}\right)^2 - 
                            {\delta\mu}_{\omega}^{2} - {\rm d}_{\omega}^2 \right)\, + \\
  & \frac{1}{2} \left({\rm d}_{\omega+1}^{2}+{\qe}_{\omega}{\rm d}_{\omega}^{2} +
     {\qe}_{\omega} \left( {\alpha}_{\omega} + {\omega}{\tilde\kappa} \right)^{2} - 
     \left( {\alpha}_{\omega+1} + ({\omega+1}){\tilde\kappa} \right)^{2} \right) \, .
\end{split}
\end{equation}
\end{prop}

To get rid of the observables $c_{\omega, \ba}$'s \eqref{c-observable} in the 
right-hand side of~\eqref{key coefficient}, we introduce, following~\cite{NekBPSCFT}, 
the functions $\{ U_{\omega}\}_{\omega \in \BZ_N}$ via:
\beq
  U_{\omega} = 1+\qe_{\omega+1}+\qe_{\omega+1}\qe_{\omega+2}+\ldots+\qe_{\omega+1}\cdots\qe_{\omega-1} \, ,
\label{eq:u-parameters}
\eeq
with the conventions $U_{\omega+ N} = U_{\omega}$ being used. 
They provide a (unique up to a common factor) solution of the following linear system:
\begin{equation}\label{nu-equation}
  (1+{\qe}_{\omega})\cdot U_{\omega} - U_{\omega-1} - {\qe}_{\omega+1}\cdot U_{\omega+1} \, =\, 0 
   \qquad \mathrm{for\ any}\, \ \omega\in \BZ_N\, .
\end{equation}
We also note that 
\begin{equation*}
  U_{\omega} - {\qe}_{\omega +1}\cdot U_{\omega +1}\, =\, 1 - {\qe}
   \qquad \mathrm{for\ any}\, \ \omega\in \BZ_N\, .
\end{equation*}

Due to the key property~(\ref{nu-equation}) of $U_\omega$'s, the coefficient of $x^{-1}$ 
in the observable
  $\sum_{\omega\in \BZ_N} U_{\omega}\, \CalX_{\omega}(x)$ 
is a degree two polynomial in the \emph{instanton charges} $\{k_{\omega}\}_{\omega\in \BZ_N}$.
Therefore,
\begin{equation*}
  \left\langle 
    \left[x^{-1}\right] \left(\sum_{\omega\in \BZ_N} U_{\omega}\, \CalX_{\omega}(x)\right)
  \right\rangle_{{\mu}^{\rm orb}}\, = \ D^{\Inst} \left({\bf\Psi}^{\Inst}\right)
\end{equation*}
with $D^{\Inst}$, a second-order differential operator in ${\qe}_{\omega}$'s, 
naturally arising from the equality
\begin{equation}\label{deriving diff. eq-n}
  \left\langle \prod_{{\omega} \in {\BZ}_{N}} k_{\omega}^{r_{\omega}} \right\rangle_{{\mu}^{\rm orb}}  =\ 
  \frac{\prod_{{\omega} \in {\BZ}_{N}} 
    \left( {\qe}_{\omega}\frac{{\partial}}{{\partial}{\qe}_{\omega}}\right)^{r_{\omega}}
    {\bf\Psi}^{\inst} ({\tilde \ba}, {\tilde \bm}, {\ep}_{1}, {\tilde \ep}_{2}; \overline{\qe})}
       {{\bf\Psi}^{\inst} ({\tilde \ba}, {\tilde \bm}, {\ep}_{1}, {\tilde \ep}_{2}; \overline{\qe})}\, , 
\end{equation}
due to~(\ref{orbifold measure},~\ref{orbifold normalization}).
We can further express $D^\Inst$ as a differential operator in $\qe$ and $w_\omega$'s by using
\begin{equation}\label{diff q-to-w}
  {\qe}\frac{{\partial}}{{\partial}{\qe}}\, =\, {\qe}_{N-1}\frac{{\partial}}{{\partial}{\qe}_{N-1}}\, , \quad 
  {w}_{\omega}\frac{{\partial}}{{\partial}w_{\omega}}\, =\, 
  {\qe}_{\omega-1}\frac{{\partial}}{{\partial}{\qe}_{\omega-1}} - {\qe}_{\omega}\frac{{\partial}}{{\partial}{\qe}_{\omega}}
  \quad \mathrm{for\ any}\ \ \omega\in \BZ_N\, .
\end{equation}
It is convenient to introduce the normalized partition function ${\bf \Psi}$ via:
\beq
\label{renormalized partition function} 
  {\bf\Psi} \, = \,   {\bf\Psi}^{\rm tree} \,  \cdot\, {\bf\Psi}^{\inst}\, , 
\eeq
where
\beq
  {\bf\Psi}^{\rm tree} \, := \, 
  {\qe}^{-\frac{1}{2{\kappa}} \sum_{\omega = 0}^{N-1} {\alpha}_{\omega}^{2}} \cdot 
  \prod_{\omega = 0}^{N-1} \, w_{\omega}^{{\mu}_{\omega}-{\alpha}_{\omega}} \, .
\label{eq:psitree}
\eeq
Combining Propositions~\ref{regularity 3},~\ref{laurent 2} with formulae~\eqref{nu-equation} 
and~\eqref{deriving diff. eq-n}, we get (cf.~Parts I, V of~\cite{NekBPSCFT}):

\begin{thm}\label{Differential eq-n 1}
The normalized partition function 
  ${\bf\Psi} = {\bf\Psi}({\tilde \ba}, {\tilde \bm}, \ve_1, {\tilde \ve}_2; \bw,\qe)$ 
of~(\ref{renormalized partition function}) satisfies the equation 
$$
  {\CalD}^{\rm BPS} ({\bf\Psi}) \, =\, 0
$$  
with the second-order differential operator ${\CalD}^{\rm BPS}$ explicitly given by (cf.~\eqref{eq:kappar})
\beq\label{eq:dbps}
 {\CalD}^{\rm BPS} \, = \, 
    \kappa \frac{{\partial}}{{\partial}{\qe}}  + \frac{{\hat H}_{0}}{\qe} + \frac{{\hat H}_{1}}{{\qe}-1} \, , 
\eeq
where ${\hat H}_{0}, {\hat H}_{1}$ are the second-order differential operators in $w_\omega$'s,
$\underline{independent \ of}$ $\qe$ and $\alpha_\omega$'s:
\begin{equation}\label{inst-d-op}
\begin{split}
  & {\hat H}_{0} = \sum_{\omega = 0}^{N-1} \, \Biggl\{ 
    \sum_{\omega' = \omega + 1}^{N-1} \frac{w_{\omega'}}{w_{\omega}} 
    \left( D_{\omega}^{2} - {\delta\mu}_{\omega}^{2} \right) + 
    \frac {1}{2} \left( D_{\omega} - {\mu}_{\omega} \right)^{2} \Biggr\} \, , \\
  & {\hat H}_{1} \, = \, - \sum_{\omega', \omega=0}^{N-1} \, 
    \frac{w_{\omega'}}{w_{\omega}} \, \left( D_{\omega}^{2} - {\delta\mu}_{\omega}^{2} \right)\, , 
\end{split}     
\end{equation}
with
\beq
  D_{\omega} = w_{\omega} \frac{{\partial}}{{\partial} w_{\omega}}\, .
\eeq
\end{thm}

\begin{rem}
Note that ${\bf\Psi}^{\rm inst}$ is a single-valued homogeneous function of $w_\omega$'s. If we wrote the 
differential equation obeyed by ${\bf\Psi}^{\rm inst}$ in the original variables ${\qe}_{0}, \ldots, {\qe}_{N-2}, {\qe}_{N-1}$, it would not contain any ambiguity due to the redundant nature of the variables $w_0, \ldots, w_{N-1}$.
However, 
the equations written in the invariant variables, such as the variables $v_i$ introduced below, look more complicated.
Conversely, by introducing more degrees of freedom with additional symmetries, modifying accordingly the prefactor
${\bf\Psi}^{\rm tree}$, one arrives at a very simple form of the operators ${\hat H}_{0}, {\hat H}_{1}$, 
cf.~Theorem \ref{BPSCFT} below. This is known as the \emph{projection method} in the theory of many-body systems~\cite{OP}. 
\end{rem}

\begin{rem}
The normalized partition function $\bf\Psi$ obeys:
\beq
  \sum_{\omega=0}^{N-1} D_{\omega}\, ({\bf\Psi}) = 
  \sum_{\omega=0}^{N-1} \left( {\mu}_{\omega} - {\alpha}_{\omega} \right) \cdot \, {\bf\Psi} \, .
\eeq
The operators ${\hat H}_{0}, {\hat H}_{1}$ in \eqref{inst-d-op} are therefore defined up to addition of 
the second-order differential operators of the form
\beq
  {\mathfrak{D}}_{1} \sum_{\omega=0}^{N-1} \left( D_{\omega} + {\alpha}_{\omega} - {\mu}_{\omega} \right)
\eeq
with a first-order differential operator $\mathfrak{D}_{1}$. The choice \eqref{inst-d-op} is uniquely characterized
by its $\alpha_\omega$-independence, for any $\omega$. 
\end{rem}


\subsection{One more coordinate change}\label{ssec v-ccordinates}

For the purpose of the next section, it will be convenient to use the coordinates
\beq
  v_{i} = \frac{w_{i-1}}{w_{0}+w_{1} + \ldots + w_{N-1}}\, , \qquad i = 1, \ldots , N-1 \, ,
\eeq
and the associated quantities
\beq
  u_{i} \, = \sum_{j=i+1}^{N} v_{j} \, , \qquad i = 0, \ldots , N-1 \, ,
\eeq
with 
\beq
   v_{N} \equiv 1 - \sum_{i=1}^{N-1}v_{i} \qquad \mathrm{and} \qquad u_{N} \equiv 0\, . 
\label{eq:uvN}
\eeq
Define the ${\BC}[[v_{1}^{\pm 1}, v_{2}^{\pm 1}, \ldots , v_{N-1}^{\pm 1} ]]$-valued power series in $\qe$ by:
\beq
  {\psi}(v_{1}, v_{2} , \ldots , v_{N-1} ; {\qe} )\,  = \, 
  {\bf\Psi}^{\rm inst} \left( v_{2}/v_{1}, v_{3}/v_{2} , \ldots , v_{N}/v_{N-1} , {\qe} v_{1}/v_{N} \right) ,
\label{eq:chif}
\eeq
where we intentionally omit the parameters ${\tilde \ba}, {\tilde \bm}, \ve_1, {\tilde \ve}_2$ in 
the right-hand side and note that 
  $$v_2/v_1={\qe}_0,\, v_3/v_2={\qe}_1,\,  \ldots\, ,\, {\qe} v_{1}/v_{N}={\qe}_{N-1}\, .$$
The following is a straightforward reformulation of Theorem~\ref{Differential eq-n 1} in the present setting:

\begin{thm}\label{Differential eq-n 1 v-ccordinates}
The function $\psi={\psi}(v_{1}, v_{2} , \ldots , v_{N-1} ; {\qe} )$ satisfies the equation 
\beq
  {\nabla}^{\rm bps} ({\psi}) = 0 
\eeq
with
\beq
\label{eqn:bps-connection}
  {\nabla}^{\rm bps} \, =  \,
  {\kappa} \frac{{\partial}}{{\partial}\qe} + \frac{{\hat h}^{\rm bps}_0}{\qe} + \frac{{\hat h}^{\rm bps}_{1}}{{\qe}-1}
\eeq
with the residues of the meromorphic connection $\nabla^{\rm bps}$ at $\qe = 0$ and $\qe = 1$ having the decomposition:
\beq
\label{eqn:bps-hamiltonians}
    {\hat h}_{0}^{\rm bps} = 
    {\hat h}_{0, \rm kin}^{\rm bps} + {\hat h}_{0, \rm mag}^{\rm bps} + {\hat h}_{0, \rm pot}^{\rm bps}\, , \qquad
    {\hat h}_{1}^{\rm bps} = 
    {\hat h}_{1, \rm kin}^{\rm bps} + {\hat h}_{1, \rm mag}^{\rm bps} + {\hat h}_{1, \rm pot}^{\rm bps}\, ,
\eeq
    with the kinetic, magnetic, and potential terms given by:
\beq
\begin{aligned}
  & {\hat h}_{0, \rm kin}^{\rm bps} = \frac{1}{2}{\rm D}^2 + \sum_{i=1}^{N-1} \,\left( u_{i}  + \frac{v_{i}}{2} \right)
    \left( v_{i}^{-1} {\rm D}_{i}  - 2{\rm D}  \right) {\rm D}_i \, ,  \qquad
    {\hat h}_{1, \rm kin}^{\rm bps} = {\rm D}^{2} - \sum_{i=1}^{N-1} v_{i}^{-1} {\rm D}_{i}^{2} \, , \\
  & {\hat h}_{0, \rm mag}^{\rm bps} =  
    \left(  {\alpha}_{N-1} + 1 - N  + \sum_{i=1}^{N-1} ( N - i - {\alpha}_{N-1}) v_{i} \right) {\rm D} \ + \\
  & \qquad \qquad \qquad \qquad 
    2 \sum_{i=1}^{N-1} \,\left(  {\mu}_{i-1} u_{i} -{\alpha}_{i-1} \left(  u_{i} + \frac{v_{i}}{2} \right) \right)
    \left(v_{i}^{-1} {\rm D}_i - {\rm D} \right) ,  \\
  & {\hat h}_{1, \rm mag}^{\rm bps} = 
    \left(  N - 1  + 2{\mu}_{N-1} - 2{\alpha}_{N-1} \right)  {\rm D} -  2 \sum_{i=1}^{N-1}    \left( {\mu}_{i-1} - {\alpha}_{i-1}
    \right)\left( v_{i}^{-1} {\rm D}_{i}  - {\rm D} \right) \, , \\
  & {\hat h}_{0, \rm pot}^{\rm bps} = 
    \sum_{i=1}^{N-1} u_{i}\frac{\left( {\mu}_{i-1} - {\alpha}_{i-1} \right)^2 -
    {\delta\mu}_{i-1}^{2}}{v_{i}} \, , \qquad 
    {\hat h}_{1, \rm pot}^{\rm bps} = 
    - \sum_{a=1}^{N} \frac{\left( {\mu}_{a-1} - {\alpha}_{a-1} \right)^{2} - {\delta\mu}_{a-1}^{2}}{v_{a}} \, ,
\end{aligned}
\label{eq:bpsvv}
\eeq
where we defined 
\beq
  {\rm D}_{i} = v_{i} \frac{{\partial}}{{\partial} v_{i}} \, , \qquad i = 1, \ldots , N-1 \, ,
\eeq
and
\beq
  {\rm D} \, = \sum_{i=1}^{N-1} {\rm D}_{i}\, .
\eeq
\end{thm}

\begin{rem}
The operator ${\nabla}^{\rm bps}$ of~(\ref{eqn:bps-connection}) depends, explicitly, on ${\vec\mu}, {\delta}{\vec\mu}, {\vec\alpha}$.
However, Theorem~\ref{Differential eq-n 1} shows that the ${\vec\alpha}$ dependence is a pure gauge: 
\beq
  Y^{-1}\, {\nabla}^{\rm bps}\, Y \qquad {\rm is} \qquad {\vec\alpha}{\rm -independent} \, , 
\eeq
where (cf.~\eqref{eq:psitree})
\beq
  Y = 
  {\qe}^{\frac{1}{2{\kappa}} \sum_{\omega = 0}^{N-1} {\alpha}_{\omega}^{2}} \cdot 
  \prod\limits_{i=1}^{N} v_{i}^{{\alpha}_{i-1} - {\mu}_{i-1}} \, .
\eeq
\end{rem}


\section{The CFT side, or the projection method}\label{ssec CFT side}

The operator ${\hat h}^{\rm bps}_{0}/{\qe} + {\hat h}^{\rm bps}_{1}/({\qe}-1)$ of~(\ref{eqn:bps-connection}) can be 
viewed as a time-dependent Hamiltonian of a quantum mechanical system with $N-1$ degrees of freedom $v_1, \ldots , v_{N-1}$. 
The parameters $\vec{\mu}=(\mu_0,\dots,\mu_{N-1}), \delta\vec{\mu}=(\delta\mu_0,\dots,\delta\mu_{N-1})$ play the r{\^o}le 
of the coupling constants, while the parameters $\vec{\alpha}=(\alpha_0,\dots,\alpha_{N-1})$ play the r{\^o}le of the
spectral parameters, such as the asymptotic momenta of $N$ particles, in the center-of-mass frame, where the interactions
between the particles can be neglected. 

The BPS/CFT correspondence \cite{NN2004} suggests to look for 
the representation-theoretic realization of the operators ${\hat h}^{\rm bps}_{0}$ and ${\hat h}^{\rm bps}_{1}$. 

We present such a realization below. 


\subsection{Flags, co-flags, lines, and co-lines}

Let $W \approx {\BC}^{N}$ be the complex vector space of dimension $N$, and let $W^*$ denote its dual. 
Let $F(W),\, F(W^{*}),\, {\BP}(W),\, {\BP}(W^*)$  denote the space of complete flags in $W$, the space of complete flags 
in $W^{*}$,  the projective space of lines in $W$, and the projective space of lines in $W^*$, respectively.  
The natural action of the general linear group $GL(W)$ on $W$ and $W^*$ gives rise to canonical actions of $GL(W)$ on 
those four projective varieties. Let $J^{a}_{b},\, {\tilde J}^{a}_{b},\, V^{a}_{b},\, {\tilde V}^{a}_{b}$, with 
$a, b = 1, \ldots , N$, denote the vector fields on $F(W),\, F(W^{*}),\, {\BP}(W),\, {\BP}(W^{*})$, respectively, 
representing those actions. Here, to define those vector fields, we need to choose some basis $\{e_{a}\}_{a = 1}^{N}$ 
in $W$, with the dual basis in $W^*$ denoted by $\{{\tilde e}^{b}\}_{b = 1}^{N}$, so that the operators 
\beq
  T^{a}_{b} = e_{b} \otimes {\tilde e}^{a} \in {\rm End}(W)
\eeq
represent the action of the Lie algebra of $GL(W)$ on $W$. They obey the $\gl_N$ commutation relations:
\beq
  \left[ T^{a}_{b}, T^{a'}_{b'} \right] = {\delta}^a_{b'} T^{a'}_{b} -\, {\delta}^{a'}_{b} T^{a}_{b'}
\label{eq:gln}
\eeq
to which we shall refer in what follows. 

We define the second-order differential operators ${\hat h}_{0}, {\hat h}_{1}$ on the product 
\beq
\label{eqn:X-variety}
  {\CalX} \, = \, F(W)  \times F(W^{*}) \times {\BP}(W) \times {\BP}(W^{*}) 
\eeq
by 
\beq
\label{eqn:hp}
  {\hat h}_{0} \, = \sum_{a,b=1}^{N} J^{a}_{b} V^{b}_{a}  \, , \qquad 
  {\hat h}_{1} \, = \sum_{a,b=1}^{N} V^{a}_{b} {\tilde V}^{b}_{a} \, . 
\eeq
These operators are independent of the choice of the basis in $W$ and are globally well-defined on ${\CalX}$.
Furthermore, they commute with the diagonal action of $GL(W)$ on $\CalX$:
\beq
  \left[ J^{a}_{b} + {\tilde J}^{a}_{b} + V^{a}_{b} + {\tilde V}^{a}_{b} , {\hat h}_{p} \right] = 0 \, , 
  \qquad a, b = 1, \ldots, N\, , \quad p  = 0, 1 \, .
\eeq
Note that the center of $GL(W)$ acts trivially on $\CalX$, hence, a natural action of $PGL(W)$ on $\CalX$.


\subsection{The $v$-coordinates}

Let us now endow $W$ with the volume form $\varpi \in {\Lambda}^{N} W^{*}$. Denote
\beq
  {\tilde\pi}^{N} = {\varpi}\, , \qquad {\pi}_{N} = {\varpi}^{-1} \in {\Lambda}^{N}W \, .
\label{eq:topf}
\eeq
Let $H = SL(W, {\varpi}) \approx SL(N, {\BC})$ denote the group of linear transformations of $W$ preserving~$\varpi$. 
The center $Z(H)\simeq \BZ_N$ of $H\subset GL(W)$ is finite and acts trivially on $\CalX$. There is an $H$-invariant 
open subset ${\CalX}^{\circ}$ (described in~(\ref{eqn:X-null})) of $\CalX$, on which the action of $H/Z(H)$ is free. 
The corresponding quotient ${\CalX}^{\circ}/H$ can be coordinatized by the values of $N-1$ functions
$v_{1}, \ldots, v_{N-1}$,  defined as follows:
\beq
\label{eq:vcor}
  v_{i} \left( w, {\tilde w}, {\rz}, {\tz}  \right) = 
  \frac{\left( {\tz} \wedge {\tilde \pi}^{i-1} \right) ({\pi}_{i}) \cdot {\tilde\pi}^{i} \left( {\rz} \wedge {\pi}_{i-1} \right)}
       {{\tz}({\rz}) \cdot {\tilde\pi}^{i-1}({\pi}_{i-1}) \cdot {\tilde\pi}^{i}({\pi}_{i})} 
  \, , \qquad i = 1, \ldots , N-1 \, ,
\eeq
where 
$$
  \left(  w =  \left( W_i \right)_{i=1}^{N-1} \, , \,  {\tilde w} = \left( {\tilde W}_{i} \right)_{i=1}^{N-1} \, , \,  
  {\rz} \, , \, {\tz}  \right) \in {\CalX}^{\circ}
$$ 
is the collection consisting of a pair 
\beq
\begin{aligned}
  w\colon  & \qquad 
    0 = W_{0} \subset W_{1} \subset W_{2} \subset \ldots  \subset W_{N-1} \subset W_{N} \equiv W \in F(W) \, , \\
  {\tilde w}\colon & \qquad 
    0 = {\tilde W}_{0} \subset {\tilde W}_{1} \subset {\tilde W}_{2} \subset \ldots  \subset {\tilde W}_{N-1} 
    \subset {\tilde W}_{N} \equiv W^{*} \in F(W^{*})\, 
\end{aligned}
\label{eq:ffpp}
\eeq
of flags in $W$ and $W^{*}$, respectively, and another pair
\beq
  {\BC}  {\rz} \subset W\, , \qquad {\BC} {\tz} \subset W^{*}  
\eeq
of lines in $W$ and $W^{*}$; and finally,  
\beq
\label{eq:pis}
  {\pi}_{i} = {\Lambda}^{i}W_{i} \subset {\Lambda}^{i}W\, , \qquad 
  {\tilde\pi}^{i} = {\Lambda}^{i}{\tilde W}_{i} \subset {\Lambda}^{i}W^{*}  
\eeq
are the corresponding $i$-polyvector and the $i$-form on $W$, both defined up to a scalar multiplier. 
Note that these scalar factor ambiguities cancel out in \eqref{eq:vcor}. 

We can also view $v_i$'s as meromorphic functions on ${\CalX}/H$. To this end, we promote $\pi_i, {\tilde\pi}^i, {\rz}, {\tz}$ 
to global objects, the canonical holomorphic sections of the corresponding vector bundles:
\beq
  {\Pi}_{i} \in H^0 \left( F(W), {\Lambda}^{i}W \otimes {\rm det}(W_{i})^{-1} \right)\, , \qquad
  {\tilde\Pi}^{i} \in H^{0} \left( F(W^{*}) , {\Lambda}^{i} W^{*} \otimes {\rm det}({\tilde W}_{i}) \right) ,
\label{eq:pisec}
\eeq
and
\beq
  Z \in H^{0} \left( {\BP}(W), W \otimes {\CalO}(1) \right) \approx W \otimes W^{*}\, , \qquad 
  {\tilde Z} \in H^{0} \left( {\BP}(W^{*}), W^{*} \otimes {\CalO}(1) \right) \approx W^{*} \otimes W\, , 
\eeq
and define
\beq
  v_{i} = 
  \frac{\left( {\tilde Z} \wedge {\tilde\Pi}^{i-1} \right) 
        \left( {\Pi}_{i} \right) \cdot {\tilde\Pi}^{i} \left( Z \wedge {\Pi}_{i-1} \right)}
       {{\tilde Z}(Z) \cdot {\tilde\Pi}^{i}({\Pi}_{i}) \cdot {\tilde\Pi}^{i-1} ({\Pi}_{i-1})}\, , \qquad i=1,\ldots,N-1 \, .
\label{eq:vcorf}
\eeq
We also note that while~(\ref{eq:vcor},~\ref{eq:vcorf}) can be extended to $i = N$, the corresponding quantity $v_N$ satisfies
\beq
\label{eq:sumv}
  \sum_{a=1}^{N} v_{a} = 1 \, ,
\eeq
due to the Desnanot-Jacobi-Dodgson-Sylvester theorem, which states that
\beq
  v_{a+1} = u_{a} - u_{a+1}\, , \qquad 
  u_{a} = \frac{\left( {\tilde Z} \wedge {\tilde \Pi}^{a} \right)  \left( Z \wedge {\Pi}_{a} \right)}
               {{\tilde Z}(Z) \cdot {\tilde\Pi}^{a}({\Pi}_{a})}
  \, , \qquad a=0, \ldots, N-1 \, .
\label{eq:ucor}
\eeq
The open set $\CalX^{\circ}\subset \CalX$ has the following description: there exists a basis $e_a$ in $W$ such that
\beq
\label{eqn:X-null}
\begin{split}
  & W_i = \mathrm{Span} (e_{1}, \ldots , e_{i})\, , \qquad 
    {\tilde W}_{i} = \mathrm{Span} ({\tilde e}^{1}, \ldots, {\tilde e}^{i})\, , \\ 
  & Z = \sum_{a=1}^{N} {\xi}_{a} e_{a}\, , \qquad {\tilde Z} = \sum_{a=1}^{N} {\xi}_{a} {\tilde e}^{a}\, , 
    \qquad {\xi}_{a} \neq 0 \, .
\end{split}
\eeq
We note that the aforementioned equality~\eqref{eq:sumv} is obvious in this basis, since 
\beq
  v_{a} = \frac{{\xi}_{a}^2}{{\xi}_1^{2}+\ldots+{\xi}_N^2} \, , \qquad a = 1, \ldots , N \, .
\eeq

\begin{rem}
\label{rem:flag duality}
The flag varieties $F(W)$ and $F(W^{*})$ are isomorphic. For example, the assignment $W_{i} = {\tilde W}^{\perp}_{N-i}$ 
gives rise to an isomorphism $F(W^{*})\iso F(W)$. Alternatively, fixing the volume form $\varpi \in {\Lambda}^{N}W^{*}$, 
we have an $SL(W)$-equivariant isomorphism  $F(W)\iso F(W^{*})$ given by:
\beq
  {\tilde\pi}^{i} = {\varpi} ( {\pi}_{N-i} )\, , \qquad i = 1, \ldots, N-1 \, .
\eeq
\end{rem}

\begin{rem} 
\label{rem:cross-ratio}
In the $N=2$ case, we have $F(W) \simeq F(W^*) \simeq {\BP}(W) \simeq {\BP}(W^*)$, and the only nontrivial coordinate 
$v_1$ of~\eqref{eq:vcor} is determined by the usual cross-ratio of four points on ${\BC\BP}^1$. More precisely, 
if $z_1, z_2, z_3, z_4 \in W$ are defined (each up to a scalar multiplier) by:
\beq
  z_{1} = {\pi}_{1}\, , \ {\varpi} ( z_2, {\cdot} ) = {\tilde z}\, , \  
  {\varpi} ( z_3, {\cdot} ) = {\tilde\pi}^{1}\, , \  z_4 = z\, , 
\eeq
then
\beq
  v_{1} = \frac{{\varpi}(z_2, z_1) {\varpi}(z_3, z_4)}{{\varpi} (z_3, z_1) {\varpi}( z_2, z_4 )}
\eeq
depends only on the four points ${\BC}z_{i} \in {\BP}(W)$. 
\end{rem}


\subsection{The ${\mathfrak L}$-twist}\label{ssec line bundle}

Let $L_{1}, \ldots , L_{N-1}$ denote the tautological line bundles over $F(W)$, the fiber of $L_i$ over the point 
$0 = W_{0} \subset W_{1} \subset W_{2} \subset \ldots \subset W_{N-1} \subset W_{N} \equiv W$ being
\beq
  L_i = W_{i}/W_{i-1}\, , \qquad i = 1, \ldots , N-1 \, .
\eeq
Similarly, let ${\tilde L}^1, \ldots, {\tilde L}^{N-1}$ denote the tautological line bundles over $F(W^{*})$, and 
\beq
  {\CalL} = {\CalO}_{{\BP}(W)} (-1)\, , \qquad {\tilde\CalL} = {\CalO}_{{\BP}(W^{*})}(-1)
\eeq
be the tautological line bundles over ${\BP}(W),\, {\BP}(W^{*})$, respectively. We note that
$$
  {\rm det}(W_{a}) = {\Lambda}^{a} W_{a}\simeq \bigotimes_{i=1}^a L_i\, , \qquad 
  {\rm det}({\tilde W}_{a}) = {\Lambda}^{a} {\tilde W}_{a}\simeq \bigotimes_{i=1}^a \tilde{L}^i\, , \qquad a=1,\ldots, N-1 \, .
$$
All these line bundles are $GL(W)$-equivariant. By abuse of notation, we shall use the same notations for the 
pull-backs of the aforementioned line bundles to $\CalX$ of~\eqref{eqn:X-variety} under the natural projections. 
The line bundles 
  ${\tilde\CalL}^{-1} \otimes {\Lambda}^{a-1}{\tilde W}_{a-1} \otimes ({\Lambda}^{a} W_{a})^{-1}$, 
  ${\Lambda}^{a}{\tilde W}_{a} \otimes {\CalL}^{-1} \otimes ({\Lambda}^{a-1} W_{a-1})^{-1}$, 
  ${\Lambda}^{i} {\tilde W}_{i} \otimes ({\Lambda}^{i} W_{i})^{-1}$, and ${\tilde\CalL}^{-1} \otimes {\CalL}^{-1}$
on $\CalX$ are $H$-invariant (and those with $a< N$ are actually $GL(W)$-invariant). 
Furthermore, each factor in formula~(\ref{eq:vcorf}) can be viewed as a holomorphic section of one 
of those line bundles. For example, 
\beq
  {\tilde\Pi}^{a} ( Z \wedge {\Pi}_{a-1} )
\eeq
is a holomorphic section of ${\rm det} ( {\tilde W}_{a}) \otimes {\CalL}^{-1} \otimes {\rm det}(W_{a-1})^{-1}$. 
Its zeroes determine the locus in $\CalX$ where the plane $W_{a-1}$, the line ${\BC}z$, and the plane 
${\tilde W}_{a}^{\perp} \subset W$ are not in general position, i.e., their linear span does not coincide 
with the entire $W$. Let $\Sigma \subset {\CalX}^{\circ}$ denote the union of vanishing loci of 
  ${\tilde\Pi}^{a} ( Z \wedge {\Pi}_{a-1} ),\, 
   ({\tilde Z} \wedge {\tilde\Pi}^{a-1}) ({\Pi}_{a}),\, 
   {\tilde \Pi}^{i}({\Pi}_{i})$
for $a = 1, \ldots, N$ and $i = 1, \ldots , N-1$. 

For ${\vec n}, {\vec{\tilde n }} \in {\BC}^{N},\ {\vec\gamma} \in {\BC}^{N-1}$, 
consider the tensor product of ``complex powers of line bundles''
\begin{equation}
\begin{split}
  & {\mathfrak L} = 
  \bigotimes_{i=1}^{N-1}  \Big( {\rm det}(W_{i}) \Big)^{-\nu_i} \otimes \, 
  \bigotimes_{i=1}^{N-1} \Big( {\rm det}({\tilde W}_{i}) \Big)^{{\tilde\nu}_{i}}  \otimes \, 
  {\CalL}^{-\mf} \, \otimes \, {\tilde\CalL}^{-\tilde{\mf}} \, = \\
  & \bigotimes_{a=1}^{N} \, 
    \left( {\tilde\CalL}^{-1} \otimes {\rm det}({\tilde W}_{a-1}) \otimes {\rm det} ( W_{a})^{-1} \right)^{{\tilde n }_{a}} 
    \otimes \,
    \left( {\rm det} ({\tilde W}_{a}) \otimes {\CalL}^{-1} \otimes {\rm det} ( W_{a-1})^{-1} \right)^{{ n }_{a}} \otimes \, \\
  &  \qquad \bigotimes_{i=1}^{N-1} 
      \left( {\rm det}( {\tilde W}_{i} ) \otimes {\rm det} ( W_{i})^{-1} \right)^{\gamma_i - n_i - {\tilde n}_i}
\label{eq:clb}
\end{split}
\end{equation}
defined on any simply-connected open domain ${\CalU} \subset \left({\CalX}^{\circ}\backslash \Sigma \right)/H$. 
Here, the complex numbers ${\mf}, {\tilde{\mf}} \in {\BC}$ and the vectors  ${\vec\nu}, {\vec{\tilde\nu}} \in {\BC}^{N-1}$ are defined via:
\begin{equation}
\label{eqn:mz-parameters}
\begin{split}
  & {\mf} = \sum_{a=1}^{N} n_{a}\, ,\quad  {\tilde{\mf}} = \sum_{a=1}^{N} {\tilde n }_{a}\, ,\\
  & {\nu}_{i} = n_{i+1} - n_{i} + {\gamma}_{i}  \, ,\qquad
  {\tilde\nu}_{i} =  {\tilde n }_{i+1} - {\tilde n}_{i} + {\gamma}_{i}\, , \qquad
  i = 1, \ldots , N-1 \, .
\end{split}
\end{equation}

\noindent
Our main result is:

\begin{thm}\label{BPSCFT}
The operators ${\hat h}^{\rm bps}_{0}, {\hat h}^{\rm bps}_{1}$ of~(\ref{eqn:bps-connection}) coincide with the operators 
${\hat h}^{\rm cft}_{0}, {\hat h}^{\rm cft}_{1}$, which are ${\hat h}_{0}, {\hat h}_{1}$ of~(\ref{eqn:hp}), viewed now as 
the differential operators on ${\CalX}^{\circ}/H$, twisted by the ``line bundle''~${\mathfrak L}$:
\beq
\label{eq:cft-operators}
  {\hat h}^{\rm cft}_{p} \, =\,  {\Upsilon}^{-1} \, {\hat h}_{p} \, {\Upsilon} \, , \qquad p=0,1 \, ,  
\eeq
where
\beq
  {\Upsilon} = 
  \prod\limits_{a=1}^{N}\, 
  \left( \frac{\left( {\tilde Z} \wedge {\tilde\Pi}^{a-1} \right)  ({\Pi}_{a}) }{{\tilde\Pi}^{a}({\Pi}_{a})} \right)^{{\tilde n }_{a}} 
  \cdot\, \left( \frac{{\tilde\Pi}^{a} \left( {Z} \wedge {\Pi}_{a-1} \right)}{{\tilde\Pi}^{a}({\Pi}_{a})} \right)^{n_{a}} \cdot \, 
  \prod\limits_{i=1}^{N-1} \left( {\tilde\Pi}^{i} \left( {\Pi}_{i}\right) \right)^{\gamma_i}
  \label{eq:ups}
\eeq
is the holomorphic section of $\mathfrak L$ on ${\CalU}$. The parameters $\vec{n}, \vec{\tilde n}, \vec{\gamma}$ 
are related to the parameters $\vec{\mu}, \delta{\vec\mu}$ and $\vec{\alpha}$ (which encode the mass parameters $\bm$ 
and the Coulomb parameters $\ba$ via~(\ref{eq:massparam},~\ref{eq:massdeltamass}) and~(\ref{eq:coulpar}), respectively) 
as follows:
\beq
\begin{aligned}
  &  n_{b} = {\mu}_{b-1} + {\delta\mu}_{b-1} - {\alpha}_{b-1}\, , \\
  &  {\tilde n }_{b} = {\mu}_{b-1} - {\delta\mu}_{b-1} - {\alpha}_{b-1}\, , \\ 
  &  {\gamma}_{i} = - 1 - {\alpha}_{i-1} + {\alpha}_{i}\, ,
\end{aligned}
\label{eq:betagammamual}
\eeq
for $b=1, \ldots, N$ and $i = 1, \ldots, N-1$. 
\end{thm}

For future use, let us record the relation between the parameters of the gauge theory and the parameters 
${\vec\nu}, {\vec{\tilde\nu}}, {\mf} , {\tilde{\mf}}$ of~\eqref{eqn:mz-parameters}:
\beq
\begin{aligned}
  & {\ve}_{1}{\nu}_{i} = m^{+}_{i+1} - m^{+}_{i} - {\ve}_{1}\, , \qquad 
   {\ve}_{1} {\tilde\nu}_{i} = m^{-}_{i+1} - m^{-}_{i}-{\ve}_{1} \, , \\
  & {\ve}_{1}{\mf} =  \sum_{f=1}^{N} m_{f}^{+} - \sum_{b=1}^{N} a_{b}  \, , \qquad
   {\ve}_{1}{\tilde{\mf}} = \sum_{f=1}^{N} m_{f}^{-} - \sum_{b=1}^{N} a_{b} \, ,
\end{aligned}
\label{eq:zetamass}
\eeq
where we used~(\ref{eq:coulpar}, \ref{eq:massdeltamass}) and the second formula of~(\ref{eq:bulkamve}).


\subsection{Proof of Theorem~\ref{BPSCFT}}\label{ssec bps=cft}

The vector fields $V^{a}_{b}, {\tilde V}^{a}_{b}$ can be explicitly written in the homogeneous coordinates 
$({\rz}^1:{\rz}^2: \cdots :{\rz}^N)$ on ${\BP}(W)$ and $({\tz}_{1}:{\tz}_{2} : \cdots : {\tz}_{N})$ 
on ${\BP}(W^{*})$:
\beq
  V^{b}_{a} = -{\rz}^{b} \frac{{\partial}}{{\partial} {\rz}^{a}}\, , \qquad 
  {\tilde V}^{b}_{a} = {\tz}_{a} \frac{{\partial}}{{\partial} {\tz}_{b}}\, ,
\label{eq:glnz}
\eeq
so that ${\hat h}_{1}$ of~(\ref{eqn:hp}) is explicitly given by:
\beq
  {\hat h}_{1} = - {\tz}({\rz})\cdot \sum_{a=1}^{N} \frac{{\partial}^{2}}{{\partial} {\rz}^{a} {\partial} {\tz}_{a}} \, ,
\label{eq:h1}
\eeq
where
\beq
  {\tz}({\rz}) = \sum_{a=1}^{N} {\tz}_{a} {\rz}^{a} \, .
\eeq
The minus sign in \eqref{eq:glnz} in the formula for $V_{a}^{b}$  does match the commutation relations~\eqref{eq:gln}. 
This minus sign is due to the fact that the vector space of polynomials in $z^a$'s is the symmetric algebra built on $W^*$, 
while that of polynomials in $\tilde z_a$'s is built on $W$. Thus, \eqref{eq:glnz} is the infinitesimal version of 
the group action, where $h \in GL(W)$ acts on $f = f({\rz}), {\tilde f} = {\tilde f}({\tz})$ via 
$f \mapsto f^{h}, {\tilde f} \mapsto {\tilde f}^{h}$:
\beq
  f^{h}({\rz}) = f (h^{-1} \cdot {\rz}) \, , \qquad 
  {\tilde f}^{h}({\tz}) = {\tilde f}({\tz}\cdot  h ) \, .
\eeq

As for $J^{a}_{b}, {\tilde J}^{a}_{b}$, let us first recall the quiver description of the flag varieties $F(W), F(W^{*})$. 
Let $F_{1}, {\tilde F}_{1}, \ldots, F_{N-1}, {\tilde F}_{N-1}$ be the sequence of complex vector spaces 
with ${\rm dim}\, F_{i} = {\rm dim}\, {\tilde F}_{i} = i$. Consider the vector spaces of linear maps: 
\beq
  {\CalA} \, =\, \bigoplus\limits_{i=1}^{N-1} \, {\rm Hom} (F_{i}, F_{i+1})\, , 
\label{eq:cavec}
\eeq
\beq
  {\tilde\CalA} \, =\, \bigoplus\limits_{i=1}^{N-1} {\rm Hom}({\tilde F}_{i+1}, {\tilde F}_{i})\, ,
\label{eq:tcavec}
\eeq
where we set $F_{N}=W$ and ${\tilde F}_{N}=W$. 
Consider the groups
\beq
  {\CalG} \, = \, \prod_{i=1}^{N-1} GL(F_{i})\, , \qquad 
  {\tilde\CalG} \, = \, \prod_{i=1}^{N-1} GL( {\tilde F}_{i})
\eeq
of linear transformations of the respective vector spaces. 
The groups $\CalG$, ${\tilde\CalG}$ act on ${\CalA}, {\tilde\CalA}$, respectively, in the natural way:
\begin{equation}\label{eq:utou}
\begin{split}
 & \left( g_{i} \right)_{i=1}^{N-1} \colon  
   \left( U_{i} \right)_{i=1}^{N-1}\, \in {\CalA} \, \mapsto  \left( g_{i+1} U_{i} g_{i}^{-1} \right)_{i=1}^{N-1}\, \in {\CalA}\, , \\
 & \left( {\tilde g}_{i} \right)_{i=1}^{N-1} \colon 
   \left( {\tilde U}_{i} \right)_{i=1}^{N-1}\, \in {\tilde\CalA} \, \mapsto  
   \left( {\tilde g}_{i} {\tilde U}_{i} {\tilde g}_{i+1}^{-1} \right)_{i=1}^{N-1}\, \in {\tilde\CalA}\, ,
\end{split}
\end{equation}
where 
  $g_{i} \in GL(F_{i}),\, U_{i}\colon F_{i} \to F_{i+1},\, 
   {\tilde g}_{i} \in GL({\tilde F}_{i}),\, {\tilde U}_{i}\colon {\tilde F}_{i+1} \to {\tilde F}_{i}$, 
and $g_N, {\tilde g}_{N}$ are vacuous. 
Then, the flag variety $F(W)$ is the quotient of the open subvariety $\CalA^{s}$ of $\CalA$, consisting of the collections 
$\left( U_{i} \right)_{i=1}^{N-1}$ for which the composition $U_{N-1} U_{N-2} \cdots U_{i}\colon F_{i} \to W$ has no kernel 
for any $i=1,\dots,N-1$, by the free action of ${\CalG}$: 
\beq
\label{eqn:flag-as-qtnt}
  F(W) = {\CalA}^{s}/{\CalG} \, .
\eeq
We can represent the $\pi_i$'s of~\eqref{eq:pis}, in coordinates, as:
\beq
\label{eq:pic}
  {\pi}_{i}  \ = 
  \sum_{1\leq a_1 < a_2 < \ldots < a_i \leq N}\, {\rm Det} 
    \Big{\Vert} \Big[ U_{N-1}U_{N-2}\cdots U_{i} \Big]^{a_k}_{\ell} \Big{\Vert}_{k,\ell=1}^{i} \ 
    e_{a_{1}} \wedge \cdots \wedge e_{a_{i}} \, .
\eeq
Here, $\Big[ U_{N-1}U_{N-2}\cdots U_{i} \Big]^{a_k}_{\ell}$ denote the matrix coefficients of the corresponding 
linear operator with respect to some bases $\{{\ve}_{\ell}^{(i)}\}_{\ell=1}^i$ in $F_i$ and the chosen basis $\{e_a\}_{a=1}^N$ in $W$. 
Note that the group $\CalG$ acts on ${\CalA}^{s}$ by the changes of bases $\{{\ve}_{\ell}^{(i)}\}_{\ell=1}^i$ 
in each $F_i$: ${\ve}_{\ell}^{(i)} \mapsto \sum_{m=1}^{i} g_{i | \ell}^{m} {\ve}_{m}^{(i)}$. This results 
in $U_{N-1}U_{N-2}\cdots U_{i}$ being multiplied on the right by $g_{i}^{-1}$; hence, according to~(\ref{eq:pic}), 
the ${\pi}_{i}$'s are transformed via:
\beq
  {\pi}_{i} \mapsto {\pi}_{i} \cdot {\rm det}(g_{i})^{-1}\, , 
\eeq
thus justifying the ${\rm det}(W_{i})^{-1}$ factor in \eqref{eq:pisec}. 
The group $GL(W)$ acts on $\CalA$ via: 
\beq
\label{eqn:GL-action-quiver}
  h \cdot \left( U_{N-1}, U_{N-2}, \ldots , U_{1} \right) \, =\, \left( h\, U_{N-1}, U_{N-2} , \ldots, U_{1} \right)\, .
\eeq
\noindent
This $GL(W)$-action preserves ${\CalA}^s\subset \CalA$ and also commutes with the $\CalG$-action. The resulting action 
of $GL(W)$ on ${\CalA}^s/{\CalG}$ clearly coincides with the natural action of $GL(W)$ on $F(W)={\CalA}^s/{\CalG}$. 
Accordingly, the $GL(W)$-action on functions on $F(W)$ is given by:
\beq
  h\colon f \mapsto f^{h} \, , \qquad f^{h}[ U_{N-1}, U_{N-2} , \ldots, U_{1} ] = f [ h^{-1} U_{N-1} , U_{N-2} , \ldots, U_{1} ] \, .
\eeq
This means that  the vector field $J_{a}^{b}\in Vect(F(W))$ representing the action of the element 
$T_a^b=e_{a} \otimes {\tilde e}^{b}\in \gl(W)$ on functions on $F(W)$ is given by (cf.\ the first formula of~(\ref{eq:glnz})):
\beq
  J_{a}^{b} \, = - \sum_{m=1}^{N-1} U_{N-1|m}^{b} \frac{{\partial}}{{\partial}U^{a}_{N-1|m}} \, , 
\label{eq:vabvf}
\eeq
where $U_{N-1|m}^{a}$ are the matrix coefficients of $U_{N-1}\colon F_{N-1} \to W$ defined via: 
\beq
  U_{N-1} {\ve}_{m}^{(N-1)} \, =\, \sum_{a=1}^{N} U_{N-1|m}^{a} \, e_{a} \, .
\eeq
Up to a compensating infinitesimal $g_i$-transformation, the vector field $J_{a}^{b}$ acts on ${\pi}_{i}$ 
(more precisely, on functions of $\pi_i$ viewed as functions on  $F(W)$) by:
\beq
\label{eq:action-on-Pi}
  J_{a}^{b} {\pi}_{i} \, = -  e_{a} \wedge {\tilde e}^{b}\, {\pi}_{i} \, .
\eeq
To clarify, the right-hand side of ~\eqref{eq:vabvf} should be viewed as a descent of the 
$\CalG$-equivariant vector field on ${\CalA}^s$, given by the same formula, to the quotient space ${\CalA}^s/{\CalG}=F(W)$.
The attentive reader will be content to see that the minus sign in \eqref{eq:vabvf} is needed to match 
the commutation relations~\eqref{eq:gln}. 

Likewise, the flag variety $F(W^*)$ admits the quotient realization:
\beq
\label{eqn:dualflag-as-qtnt}
  F(W^*) = {\tilde\CalA}^{s}/{\tilde\CalG} \, ,
\eeq
where the open subvariety ${\tilde\CalA}^{s}$ of ${\tilde\CalA}$ consists of the collections 
$\left( {\tilde U}_{i} \right)_{i=1}^{N-1}$ for which the composition
${\tilde U}_{i} {\tilde U}_{i+1} \cdots {\tilde U}_{N-1}\colon W \to {\tilde F}_{i}$ has no cokernel 
(i.e., has the maximal rank) for any $i=1,\dots,N-1$, and the action of ${\tilde\CalG}$ on ${\tilde\CalA}^{s}$ is free. 
We can represent the ${\tilde\pi}^i$'s of~\eqref{eq:pis}, in coordinates, as:
\beq
\label{eq:tpic}
  {\tilde\pi}^{i} \ = \sum_{1\leq a_1 < a_2 < \ldots < a_i \leq N} {\rm Det} 
  \Big{\Vert} \left[ {\tilde U}_{i}{\tilde U}_{i+1}\cdots {\tilde U}_{N-2} {\tilde U}_{N-1} \right]^{\ell}_{a_{k}}
  \Big{\Vert}_{k,\ell=1}^{i} \ {\tilde e}^{a_{1}} \wedge \cdots \wedge {\tilde e}^{a_{i}}\, .
\eeq
Here, $\left[ {\tilde U}_{i}{\tilde U}_{i+1}\cdots {\tilde U}_{N-2} {\tilde U}_{N-1} \right]^{\ell}_{a_{k}}$ denote 
the matrix coefficients of the corresponding linear operator with respect to some bases 
$\{{\tilde\ve}_{\ell}^{(i)}\}_{\ell=1}^i$ in ${\tilde F}_i$ and the bases $\{e_a\}_{a=1}^N$ in $W$ which is dual 
to the chosen basis $\{ {\tilde e}^a\}_{a=1}^N$ in $W^{*}$. 
Note  that the group ${\tilde\CalG}$ acts on ${\tilde\CalA}^{s}$ by the changes of bases
$\{{\tilde \ve}_{\ell}^{(i)}\}_{\ell=1}^i$ in each ${\tilde F}_i$: 
  ${\tilde \ve}_{\ell}^{(i)} \mapsto \sum_{m=1}^{i} {\tilde g}_{i | \ell}^{m} {\tilde\ve}_{m}^{(i)}$. 
This results in ${\tilde U}_{i}{\tilde U}_{i+1}\cdots {\tilde U}_{N-2} {\tilde U}_{N-1}$ being multiplied on the left by 
${\tilde g}_{i}$; hence, according to~(\ref{eq:tpic}), the ${\tilde\pi}^{i}$'s are transformed via:
\beq
  {\tilde\pi}^{i} \mapsto {\tilde\pi}^{i} \cdot {\rm det}({\tilde g}_{i})\, , 
\eeq
thus justifying the ${\rm det}({\tilde W}_{i})$ factor in \eqref{eq:pisec}. 
The group $GL(W)$ acts on ${\tilde\CalA}$ via: 
\beq
  h \cdot \left( {\tilde U}_{N-1}, {\tilde U}_{N-2}, \ldots, {\tilde U}_{1} \right) \, =\, 
  \left( {\tilde U}_{N-1} h^{-1}, {\tilde U}_{N-2} , \ldots, {\tilde U}_{1} \right)\, .
\eeq
This action preserves ${\tilde\CalA}^s\subset {\tilde\CalA}$ and also commutes with the ${\tilde\CalG}$-action. 
The resulting action of $GL(W)$ on ${\tilde\CalA}^s/{\tilde\CalG}$ clearly coincides with the natural action of 
$GL(W)$ on $F(W^*)={\tilde\CalA}^s/{\tilde\CalG}$, see~\eqref{eqn:dualflag-as-qtnt}. 
Therefore, the vector field ${\tilde J}_{a}^{b}\in Vect(F(W^*))$ representing the action of the element 
$T_a^b=e_{a} \otimes {\tilde e}^{b}\in \gl(W)$ on $F(W^*)$ is given by (cf.\ the second formula of~(\ref{eq:glnz})):
\beq
  {\tilde J}_{a}^{b} \, = \sum_{m=1}^{N-1} {\tilde U}_{N-1|a}^{m} \frac{{\partial}}{{\partial}{\tilde U}^{m}_{N-1|b}} \, ,
\label{eq:vtabvf}
\eeq
where ${\tilde U}_{N-1|a}^{m}$ are the matrix coefficients of ${\tilde U}_{N-1}\colon W\to {\tilde F}_{N-1}$ defined via: 
\beq
  {\tilde U}_{N-1} e_{a} = \sum_{m=1}^{N-1} {\tilde U}_{N-1|a}^{m} \, {\ve}_{m}^{(N-1)} \, .
\eeq
To clarify, the right-hand side of ~\eqref{eq:vtabvf} should be viewed as a descent of the ${\tilde\CalG}$-equivariant 
vector field on ${\tilde\CalA}^s$, given by the same formula, to the quotient space 
${\tilde\CalA}^s/{\tilde\CalG}=F(W^*)$. The attentive reader will be content to see that 
the commutation relations \eqref{eq:gln} are obeyed by ${\tilde J}_{a}^{b}$ of~(\ref{eq:vtabvf}). 


\subsection{End of proof of Theorem~\ref{BPSCFT}}\label{ssec end-proof}

It remains to compute the action of the operators ${\Upsilon}^{-1} {\hat h}_{p} {\Upsilon}$ in the coordinates $v_i$, 
and then to compare formulas~(\ref{eq:tc1},~\ref{eq:tc2}) in Appendix~\ref{sec:cftcalc} to  
formulas~(\ref{eqn:bps-hamiltonians},~\ref{eq:bpsvv}). We leave this straightforward computation to the interested reader.


\section{Representation theory}\label{sec repr theory}

Let us now explain the representation-theoretic meaning of the main Theorem~\ref{BPSCFT}. 
Namely, we identify the function $\Phi$, given by
\beq
  {\Phi} = {\Upsilon} \left( U, {\tilde U}, {\rz}, {\tz} \right) \cdot {\psi} (v_{1}, \ldots , v_{N-1};{\qe})\, , 
\label{eq:phiups}
\eeq
for any $\qe$, with the $\ssl_N$-invariant in the completed tensor product
\beq
  {\Phi} \in  \left(  V_{1} {\hat\otimes} V_{2} {\hat\otimes} V_{3} {\hat\otimes} V_{4}  \right)^{\ssl_N}
\eeq
of four irreducible infinite-dimensional representations $\{V_i\}_{i=1}^{4}$ of the Lie algebra $\ssl_N$. 

We shall actually define $V_i$'s as representations of $\gl_N$. Let us denote the generators of $\gl_N$ by ${\bJ}_{a}^{b}$, 
with $a, b = 1, \ldots, N$. These obey the commutation relations \eqref{eq:gln}:
\beq
  \left[ {\bJ}^{a}_{b}, {\bJ}^{a'}_{b'} \right] = {\delta}^{a}_{b'} {\bJ}^{a'}_{b} - {\delta}^{a'}_{b} {\bJ}^{a}_{b'} \, .
\label{eq:gln1}
\eeq

\begin{nott}
For a Lie algebra $\mathfrak{g}$, its element $\xi \in \mathfrak{g}$, and a representation $R$ of $\mathfrak{g}$,
we denote by $T_{R}({\xi}) \in {\rm End}(R)$ the linear operator in $R$, corresponding to $\xi$.
\end{nott}

It is well-known that \eqref{eq:gln1} implies that the \emph{Casimir operators}
\beq
  {\CalC}_{k} \ = \sum_{a_{1}, a_{2}, \ldots , a_{k} = 1}^{N} 
  {\bJ}_{a_{1}}^{a_{2}} {\bJ}_{a_{2}}^{a_{3}} \dots {\bJ}_{a_{k}}^{a_{1}} \, \in \, U({\gl}_{N})
\label{eq:casimirs}
\eeq
commute with all generators ${\bJ}^{a}_{b}$, so that in every irreducible $\gl_N$-representation $R$ the operator ${\CalC}_{k}$
acts via a multiplication by a scalar $c_{k}(R)$, also commonly known as the \emph{$k$-th Casimir of $R$}:
\beq
\label{eq:casimir-values}
  \sum_{a_{1}, a_{2}, \ldots , a_{k} = 1}^{N} 
  T_{R}\left( {\bJ}_{a_{1}}^{a_{2}} \right) T_{R} \left( {\bJ}_{a_{2}}^{a_{3}} \right) \dots 
  T_{R} \left( {\bJ}_{a_{k}}^{a_{1}}\right) = c_{k}(R) \cdot {\bf 1}_{R} \, .
\eeq

\begin{nott}
The Lie algebra $\ssl_N$ is a subalgebra of $\gl_N$ with a basis consisting of ${\bJ}_{a}^{b}$, with $a \neq b$, and
\beq
  {\bh}_{i} = {\bJ}_{i}^{i} - {\bJ}_{i+1}^{i+1}\, , \qquad i = 1, \ldots , N-1 \, .
\label{eq:h-gen}
\eeq
\end{nott}

\begin{nott}
The Chevalley generators of $\ssl_N$ are formed by ${\bh}_i$'s, and 
\beq
  {\fr}_{i} = {\bJ}_{i+1}^{i} \, , \qquad {\er}_{i} = {\bJ}_{i}^{i+1} \, , 
\label{eq:efh-gen}
\eeq
also for $i = 1, \ldots, N-1$. 
\end{nott}

The elements ${\er}_{i}$ generate, via commutators, the Lie subalgebra $\mathfrak{n}_{+}$ of $\ssl_N$.
As a vector space, $\mathfrak{n}_{+}$ has a basis consisting of ${\bJ}_{a}^{b}$ with $b > a$. Likewise, 
the elements ${\fr}_{i}$ generate the Lie subalgebra $\mathfrak{n}_{-}$ which, as a vector space, has a basis 
consisting of ${\bJ}_{a}^{b}$ with $b < a$.

\begin{rem}
With a slight abuse of notation, when this does not lead to a confusion, below we shall also denote by 
${\bh}_{i}, {\fr}_i, {\er}_i$ the corresponding operators 
\beq
  T_{R} ({\bJ}_{i}^{i}) - T_{R}({\bJ}_{i+1}^{i+1}) \, , \ T_{R} ({\bJ}_{i+1}^{i}) \, , \ T_{R}({\bJ}_{i}^{i+1})
\eeq
in a $\gl_N$-module $R$.
\end{rem}


\subsection{Verma modules}


\subsubsection{Lowest weight module}\label{ssec:hwverma}

For a generic ${\vec\nu} \in {\BC}^{N-1}$, the lowest weight Verma $\ssl_N$-module ${\CalV}_{\vec\nu}$ is defined, 
algebraically, as follows. There is a vector ${\Omega}_{\vec\nu} \in {\CalV}_{\vec\nu}$, which obeys:
\beq
  {\bJ}^{a}_{b} {\Omega}_{\vec\nu} = 0\, , \qquad a < b \, ,
\label{eq:hwc}
\eeq
and:
\beq
  {\bh}_{i}\, {\Omega}_{\vec\nu} \, = \,  -{\nu}_{i} \, {\Omega}_{\vec\nu}\, , \qquad i = 1, \ldots , N-1 \, ,
\label{eq:hwvac}
\eeq
and which generates ${\CalV}_{\vec\nu}$, i.e., ${\CalV}_{\vec\nu}$ is spanned by polynomials in ${\bJ}^{a}_{b}$, 
with $a > b$, acting on ${\Omega}_{\vec\nu}$. Geometrically, ${\CalV}_{\vec\nu}$ can be realized as the space of 
analytic functions ${\Psi}$ of $\left( U_{i} \right)_{i=1}^{N-1}$, obeying:
\beq
  {\Psi} \left[ g_{i+1} U_{i} g_{i}^{-1} \right]_{i=1}^{N-1} \prod_{i=1}^{N-1} {\rm det}(g_{i})^{{\nu}_{i}} = 
  {\Psi} \left[ U_{i} \right]_{i=1}^{N-1} \, , \qquad 
  \left( g_{i} \right)_{i=1}^{N-1} \in {\CalG}^{\rm formal} \, ,
\label{eq:psiu}
\eeq
where $g_N$ is vacuous and ${\CalG}^{\rm formal}$ denotes the group of formal exponents $g_{i} = {\exp} \, h {\xi}_{i}$ 
with ${\xi}_{i} \in {\rm End}(F_{i})$ and $h$ being a nilpotent parameter.

\begin{rem}
\label{rem:global quasi-invariants}
For ${\vec\nu} \in {\BZ}^{N-1}$, the equation~\eqref{eq:psiu} makes sense for $(g_{i})_{i=1}^{N-1} \in {\CalG}$.
For ${\vec\nu} \in {\BZ}^{N-1}_{\geq 0}$, the polynomial solutions to the equation~\eqref{eq:psiu} are in one-to-one 
correspondence with the holomorphic sections of the following line bundle on the complete flag variety $F(W)$: 
\beq
  {\bL}_{W, \vec\nu} = \bigotimes\limits_{i=1}^{N-1}\, {\rm det}(W_{i})^{-\nu_i} \, .
\label{eq:lwnu}
\eeq
\end{rem}

For our chosen basis $\{e_{a}\}_{a=1}^N$ of $W$, consider the $i$-form ${\tilde\pi}^{i}_{0}$ defined via:
\beq
\label{eq:fixed form}
  {\tilde\pi}^{i}_{0} = {\tilde e}^{1} \wedge {\tilde e}^{2} \wedge \cdots \wedge {\tilde e}^{i} \, .
\eeq
Then:
\beq
  {\Omega}_{\vec\nu} \, := \prod_{i=1}^{N-1} \left( {\tilde\pi}^{i}_{0} ( {\pi}_{i} ) \right)^{{\nu}_{i}} = \ 
  \prod_{i=1}^{N-1} \, \left( {\rm Det} \left\Vert \Big[ U_{N-1} U_{N-2} \cdots U_{i} \Big]_{a}^{b} \right 
  \Vert_{a,b =1}^{i} \right)^{{\nu}_{i}} 
\label{eq:vacn}
\eeq
(here, the index $b$ runs through the labels of the first $i$ basis vectors $e_b$ in $W$, while the index $a$ 
runs through the labels of a basis ${\ve}_{a}^{(i)}$ in $F_i$) clearly satisfies~\eqref{eq:psiu}. 
Furthermore, using ${\tilde\pi}^{i}_{0} \left( e_{a} \wedge {\tilde e}^{b} {\pi}_{i} \right) = 0$
unless $i\geq a$ and $b>i$ for $a\ne b$, we get \eqref{eq:hwc} and \eqref{eq:hwvac}, due to~\eqref{eq:action-on-Pi}. 

The Lie algebra $\gl_N$ acts on the space of analytic functions $\Psi=\Psi[U_i]$ by vector fields, 
viewed as the first-order differential operators, via~\eqref{eq:vabvf}:
\beq
  T_{{\CalV}_{\vec\nu}} \left( {\bJ}_{a}^{b} \right) {\Psi} = {\rm Lie}_{J_{a}^{b}} \left( {\Psi}\right) \, .
\eeq
We can easily compute the first two Casimirs of ${\CalV}_{\vec\nu}$:
\beq
\begin{aligned}
  & c_{1}({\CalV}_{\vec\nu}) =  -\sum_{i=1}^{N-1} i {\nu}_{i} \, , \\
  & c_{2}({\CalV}_{\vec\nu}) = \sum_{i=1}^{N-1} i {\nu}_{i} \left( N - i + \nu_i + 2\sum_{j = i+1}^{N-1} {\nu}_{j} \right) .
\end{aligned}
\eeq
Now, obviously $\Omega_{\vec\nu}$ is not well-defined for arbitrary $U_i$'s. We need first to impose:
\beq
  {\tilde\pi}^{i}_{0} ({\pi}_{i}) \neq 0\, , \qquad i = 1, \ldots, N-1 \, .
\label{eq:ineq}
\eeq
On the open set of $U_i$'s obeying \eqref{eq:ineq} $\Omega_{\vec\nu}$ is not single-valued. We can, however, view it 
as an analytic function in the neighborhood $F(W)^{\circ}$ of the point where, in some ${\CalG}$-gauge, 
${\pi}_{i} = {\pi}_{i}^{0}$ with the $i$-polyvector ${\pi}_{i}^{0}$ defined via:
\beq
  {\pi}_{i}^{0} = e_{1} \wedge \dots \wedge e_{i} \, .
\label{eq:pi0}
\eeq
To parametrize $F(W)^{\circ}$, we use:
\beq
\label{eq:key coordinate 1}
  u_{k}^{(i)} = 
  \frac{{\tilde\pi}^{i}_{0} \left( e_{k} \wedge {\tilde e}^{i+1} {\pi}_{i} \right)}{{\tilde\pi}^{i}_{0} ({\pi}_{i})} = 
  \frac{{\rm Det}\Vert \left( U_{N-1}\ldots U_{i} \right)^{a_{m}}_{\ell} \Vert_{m,\ell=1}^{i}}
       {{\rm Det}\Vert \left( U_{N-1}\ldots U_{i} \right)^{m}_{\ell} \Vert_{m,\ell=1}^{i}}\, , \qquad 
  1 \leq k \leq i \leq N-1 \, ,
\eeq
where $a_{m} = m$ for $m \neq k$ while $a_{k} = i+1$, so that the vectors 
\beq
  e^{(i)}_{\ell} \, , \qquad 1 \leq \ell \leq i \, ,
\eeq
form the unique basis in $W_i = {\rm Im}\left( U_{N-1}U_{N-2} \ldots U_{i} \right)$, $ i = 1, \ldots , N-1$,  obeying:
\beq
\begin{aligned}
  & {\pi}_{i} = e^{(i)}_{1} \wedge e^{(i)}_{2} \wedge \cdots \wedge e^{(i)}_{i} \, , \\
  & e^{(i)}_{\ell} = e^{(i+1)}_{\ell} + u_{\ell}^{(i)} e^{(i+1)}_{i+1}\, , \qquad 1 \leq \ell \leq i \leq N-1 \, ,
\label{eq:GZbasis}
\end{aligned}
\eeq
with $e^{(N)}_{a} : = e_{a}$. Therefore, we have:
\beq
\begin{aligned}
  & e^{(i)}_{\ell} = e_{\ell} + \sum_{j=1}^{N-i} {\bU}^{i | j}_{\ell} e_{i+j} \, , \\
  & {\bU}^{i | j}_{\ell}  = u_{\ell}^{(i)}  {\delta}_{j}^{1}  +  {\bU}^{i+1 | j-1}_{\ell} + 
    u_{\ell}^{(i)} {\bU}^{i+1 | j-1}_{i+1}  \, , 
\end{aligned}
\label{eq:eilc}
\eeq
with ${\bU}^{i | j}_{\ell}$ polynomial in $u^{(m)}_{k}$, $m \geq i$, nonzero only for $1\leq j \leq N-i, 1 \leq \ell \leq i$. 
Explicitly,  
\beq
\begin{aligned}
  & {\bU}^{i|1}_{\ell} = u^{(i)}_{\ell} \, , \qquad 
    {\bU}^{i|2}_{\ell} = u^{(i+1)}_{\ell} + u_{\ell}^{(i)} u^{(i+1)}_{i+1} \, , \\
  & {\bU}^{i|3}_{\ell} = u^{(i+2)}_{\ell} + u_{\ell}^{(i+1)} u^{(i+2)}_{i+2} + 
   u^{(i)}_{\ell} \left( u^{(i+2)}_{i+1} + u^{(i+1)}_{i+1} u^{(i+2)}_{i+2} \right) \, , \quad \dots
\end{aligned}
\eeq
Invoking~\eqref{eq:GZbasis} and the first equality of~\eqref{eq:eilc}, we obtain the following analogue of~\eqref{eq:key coordinate 1}:
\beq
\label{eq:key higher coordinate}
  {\bU}^{i|a-i}_{b} = 
  \frac{{\tilde\pi}^{i}_{0} \left( e_{b} \wedge {\tilde e}^{a} {\pi}_{i} \right)}{{\tilde\pi}^{i}_{0} ({\pi}_{i})}\, , 
  \qquad 1 \leq b \leq i <a \leq N \, .
\eeq
Since the local coordinates $u_{k}^{(i)}$ are $\CalG$-invariant, the general solution to \eqref{eq:psiu} can be written as:
\beq
  {\Psi} \left[ U_{i} \right] = {\psi}\left[ u_{k}^{(i)} \right] \cdot\, {\Omega}_{\vec\nu}
\label{eq:psiuver}
\eeq
with some analytic functions ${\psi}$. We amend the definition of ${\CalV}_{\vec\nu}$ given prior to 
Remark~\ref{rem:global quasi-invariants} by rather defining ${\CalV}_{\vec\nu}$ as the space of analytic functions $\Psi$,
obeying~\eqref{eq:psiu}, such that the corresponding functions $\psi$~\eqref{eq:psiuver} are polynomials in $u^{(i)}_{k}$'s. 
Using the equality (based on~\eqref{eq:key higher coordinate})
\beq
  {\bJ}_{b}^{a} {\Omega}_{\vec\nu} \, = \, - \left(  {\delta}^{a}_{b} \sum_{i \geq a} {\nu}_{i} \, + 
  \sum_{i=b}^{a-1} {\nu}_{i} {\bU}^{i|a-i}_{b} \right)  \cdot\, {\Omega}_{\vec\nu}  \, ,
\eeq
the generators ${\bJ}_{b}^{a}$ can be expressed as the first-order differential operators in $u_{k}^{(i)}$:
\beq
  {\bJ}_{b}^{a} = 
  - \sum_{1 \leq k \leq i \leq N-1} \left( {\delta}^{a}_{k} + {\bU}^{i|a-i}_{k} \right) 
  \left( {\delta}^{i+1}_{b} - u^{(i)}_{b} \right) \frac{\partial}{\partial u^{(i)}_{k}} - 
  {\delta}^{a}_{b} \sum_{i \geq a} {\nu}_{i}  - \sum_{i=b}^{a-1} {\nu}_{i} {\bU}^{i|a-i}_{b} \, ,
\label{eq:slngen}
\eeq
with polynomial in $u^{(i)}_{k}$'s coefficients. In particular, the Cartan generators of $\gl_N$ act by:
\beq
  {\bJ}_{a}^{a} =  - \sum_{k<a} \left(  u^{(a-1)}_{k} \frac{\partial}{\partial u_{k}^{(a-1)}} \right) + 
  \sum_{k \geq a} \left(  u^{(k)}_{a} \frac{\partial}{\partial u_{a}^{(k)}} - {\nu}_{k} \right) ,
\eeq
hence, the Cartan generators of $\ssl_N$ act by:
\begin{multline}
  {\bh}_{i} = \\ 
  - {\nu}_{i} + 2 u^{(i)}_{i} \frac{\partial}{\partial u_{i}^{(i)}}  -  
  \sum_{k < i} \left( u^{(i-1)}_{k} \frac{\partial}{\partial u_{k}^{(i-1)}} - 
    u^{(i)}_{k} \frac{\partial}{\partial u_{k}^{(i)}} \right) + 
  \sum_{k > i} \left( u^{(k)}_{i} \frac{\partial}{\partial u_{i}^{(k)}} - 
    u^{(k)}_{i+1} \frac{\partial}{\partial u_{i+1}^{(k)}} \right) = \\
  - {\nu}_{i} - {\rm deg}_{u^{(i-1)}_{*}} + {\rm deg}_{u^{(i)}_{*}} + {\rm deg}_{u_{i}^{(*)}} - {\rm deg}_{u^{(*)}_{i+1}}\, . 
\end{multline}
With the natural definition of the order on the weights, it is not difficult to show that the positive degree 
polynomials in $u_{k}^{(i)}$'s have higher weights than the vacuum, the state ${\psi} = 1$. 
According to~\eqref{eq:slngen}, the generators ${\fr}_{i} = {\bJ}_{i+1}^{i}$ act by:
\beq
  {\fr}_{i} = - \frac{\partial}{\partial u^{(i)}_{i}} + \sum_{k > i}  u^{(k)}_{i+1} \frac{\partial}{\partial u^{(k)}_{i}} \, ,
\eeq
thus annihilating the vacuum, the state ${\psi}=1$, as they should. 
Likewise, according to~\eqref{eq:slngen}, the generators ${\er}_{i} = {\bJ}_{i}^{i+1}$ act by:
\beq
  {\er}_{i} = -\sum_{k < i}  u^{(i)}_{k} \frac{\partial}{\partial u^{(i-1)}_{k}} + 
  \sum_{k > i}   u^{(k)}_{i} \frac{\partial}{\partial u^{(k)}_{i+1}} - 
  u^{(i)}_{i} \left( \sum_{k < i}  u_{k}^{(i-1)} \frac{\partial}{\partial u^{(i-1)}_{k}} - 
    \sum_{k \leq i}  u^{(i)}_{k} \frac{\partial}{\partial u^{(i)}_{k}} + {\nu}_{i} \right) \, ,
\eeq
which generate the whole module, as we can see using $[{\er}_{i}, {\er}_{i+1}] = {\bJ}^{i+2}_{i}$, etc. 


\subsubsection{Highest weight  module}\label{ssec:lwverma}

For a generic $\vec{\tilde\nu}\in \BC^{N-1}$, the highest weight Verma $\ssl_N$-module ${\tilde\CalV}_{\vec{\tilde\nu}}$ 
is defined similarly, so we'd be brief. Algebraically, ${\tilde\CalV}_{\vec{\tilde\nu}}$ is generated by a vector
${\tilde\Omega}_{\vec{\tilde\nu}}$, obeying:
\beq
  {\bJ}^{a}_{b} {\tilde\Omega}_{\vec{\tilde\nu}} = 0\, , \qquad a > b \, ,
\label{eq:lwc}
\eeq
and:
\beq
  {\bh}_{i}\,  {\tilde\Omega}_{\vec{\tilde\nu}} \, = \, {\tilde\nu}_{i} \, {\tilde\Omega}_{\vec{\tilde\nu}}\, , \qquad 
  i = 1, \ldots , N-1 \, .
\label{eq:lwvac}
\eeq
Geometrically, ${\tilde\CalV}_{\vec{\tilde\nu}}$ can be realized in the space of analytic functions 
${\tilde\Psi}$ of $\left( {\tilde U}_{i} \right)_{i=1}^{N-1}$, obeying:
\beq
  {\tilde\Psi} \left[ {\tilde g}_{i} {\tilde U}_{i} {\tilde g}_{i+1}^{-1} \right]_{i=1}^{N-1} 
  \prod_{i=1}^{N-1} {\rm det}({\tilde g}_{i})^{-{\tilde\nu}_{i}} = {\tilde\Psi} \left[ {\tilde U}_{i} \right]_{i=1}^{N-1} \, , \qquad 
  \left( {\tilde g}_{i} \right)_{i=1}^{N-1} \in {\tilde\CalG}^{\rm formal} \, ,
\label{eq:psiut}
\eeq
where ${\tilde g}_N$ is vacuous and ${\tilde\CalG}^{\rm formal}$ denotes the group of formal exponents 
${\tilde g}_{i} = {\exp}\, h {\tilde\xi}_{i}$ with ${\tilde\xi}_{i} \in {\rm End}({\tilde F}_{i})$ and $h$ being 
a nilpotent parameter. Again, we take:
\beq
\label{eq:vacua-highest}
  {\tilde\Omega}_{\vec{\tilde\nu}} := \prod_{i=1}^{N-1} \left( {\tilde\pi}^{i}({\pi}_{i}^{0}) \right)^{{\tilde\nu}_{i}} \, ,
\eeq
which clearly satisfies~(\ref{eq:lwc},~\ref{eq:lwvac}). 
Then, ${\tilde\CalV}_{\vec{\tilde\nu}}$ is realized in the space of analytic functions  ${\tilde\Psi}$, obeying~\eqref{eq:psiut}, 
of the form ${\tilde\Psi}[{\tilde U}_i]={\tilde\psi}[{\tilde u}_{(i)}^{k}]\cdot {\tilde\Omega}_{\vec{\tilde\nu}}$ with 
${\tilde\psi}$ polynomial in the $\tilde{\mathcal{G}}$-invariant coordinates 
\beq
\label{eq:coordinates-highest}
  {\tilde u}_{(i)}^{k} = 
  \frac{{\tilde e}^{k} \wedge \iota_{e_{i+1}} {\tilde\pi}^{i} \left(  {\pi}_{i}^{0} \right)}{{\tilde\pi}^{i} ({\pi}_{i}^{0})}  \, ,
  \qquad 1 \leq k \leq i \leq N-1 \, ,
\eeq
on the open domain $F(W^{*})^{\circ}$, where ${\tilde\pi}^{i}({\pi}^{0}_{i}) \neq 0$ for $i=1,\ldots, N-1$. 

\begin{rem}\label{Remark on Verma modules}
The identification of the vector space of representation ${\CalV}_{\vec\nu}$ with the space of polynomials in $u^{(i)}_{k}$'s, 
and similarly for ${\tilde\CalV}_{\vec{\tilde\nu}}$, is known mathematically under the name of the Poincare-Birkhoff-Witt 
theorem~\cite{PBW} (apparently proven in the case of our interest by A.~Capelli).
\end{rem}

\begin{rem}
The genericity assumption on $\vec\nu \in {\BC}^{N-1}$ (resp.\ ${\vec{\tilde\nu}} \in {\BC}^{N-1}$)  guarantees that 
the Verma $\ssl_N$-module ${\CalV}_{\vec\nu}$ (resp.\ ${\tilde\CalV}_{\tilde{\vec\nu}}$) is  irreducible, and thus is the unique 
lowest (resp.\ highest) weight module of the given lowest (resp.\ highest) weight, up to an isomorphism. 
\end{rem}


\subsection{Twisted HW-modules}\label{sssec:hwmod}

For generic ${\bn} = (n_{1}, \ldots , n_{N}) \in {\BC}^{N}$ and 
${\tilde\bn} = ({\tilde n}_{1}, \ldots , {\tilde n}_{N}) \in {\BC}^{N}$, let us define the \emph{HW-modules} 
$H_{\bn}$ and ${\tilde H}_{\tilde\bn}$ of $\gl_N$ (for W.~Heisenberg and H.~Weyl) by making  ${\bJ}_{a}^{b}$
act via the first-order differential operators in $N$ complex variables. In other words, the generators of $GL(N)$ 
in its defining $N$-dimensional representation $W$ or its dual $W^{*}$ act on the space of appropriately twisted 
functions on ${\rm Hom}({\sF}, W)$ or ${\rm Hom}(W, {\tilde\sF})$, where ${\sF} \approx {\BC}$, ${\tilde\sF} \approx {\BC}$ 
denote complex lines. 

Explicitly, let $\left( {\rz}^{a} \right)_{a=1}^{N}$ and $\left( {\tz}_{a} \right)_{a=1}^{N}$ 
denote the coordinates on ${\rm Hom}({\sF}, W)$ and ${\rm Hom}(W, {\tilde\sF})$, respectively, in the dual bases 
$(e_{a})_{a=1}^{N}$, $({\tilde e}^{a})_{a=1}^{N}$ of $W,W^*$ we used in the previous section and in the dual bases 
${\se} \in {\sF}, {\tilde\se} \in {\sF}^*$.  Then, the underlying vector spaces $H_{\bn}$, ${\tilde H}_{\tilde\bn}$ 
of the HW-modules are the spaces of homogeneous (i.e., degree zero) Laurent polynomials in 
$\{{\rz}^{a}\}, \{{\tz}_{a}\}$, respectively:
\beq
  H_{\bn} = {\BC}[{\rz}^{a}, ({\rz}^{a})^{-1}]^{{\BC}^{\times}} \, , \qquad
  {\tilde H}_{\tilde\bn} = {\BC}[{\tz}_{a}, {\tz}_{a}^{-1}]^{{\BC}^{\times}} \, , 
\eeq
while the generators of $\gl_N$ are represented by the following differential operators:
\beq
  T_{H_{\bn}} \left( {\bJ}_{b}^{a} \right) = 
  - {\omega}_{\bn}^{-1}\, \left( {\rz}^{a} {\pa}_{{\rz}^{b}} \right) \, {\omega}_{\bn} 
\label{eq:hwmod1}
\eeq
and
\beq
 T_{{\tilde H}_{\tilde\bn}} \left( {\bJ}_{b}^{a} \right) =  
 {\tilde\omega}_{\tilde\bn}^{-1}\, \left( {\tz}_{b} {\pa}_{{\tz}_{a}} \right) \, {\tilde\omega}_{\tilde\bn} 
\label{eq:hwmod2}
\eeq
with
\beq
\label{eq:omegas}
  {\omega}_{\bn} = \prod_{a=1}^{N} \left( {\rz}^{a} \right)^{n_{a}}\, , \qquad 
  {\tilde\omega}_{\tilde \bn} = \prod_{a=1}^{N} {\tz}_{a}^{{\tilde n}_{a}} \, .
\eeq

\begin{rem}\label{Remark on HW modules}
\noindent
For ${\tilde\bn}=(s,\ldots,s)$, the module ${\tilde H}_{\tilde\bn}$ coincides with $V_s$ of~\cite[\S1]{E1}, 
as $\ssl_N$-modules.
\end{rem}

In general, ${\tilde H}_{\tilde\bn}$ is a \emph{twisted} version of $V_{({\tilde n}_1+\ldots+{\tilde n}_N)/N}$, with 
underlying vector spaces being isomorphic. We thus shall use the following notation:

\begin{nott}
For $\mf \in \BC$ and $\vec\mu \in {\BC}^{N-1}$, define:
\beq
\label{eq:omega included}
  {\CalH}_{\mf}^{\vec\mu} := {\omega}_{\bn} \cdot H_{\bn}
\eeq
with 
\beq
  {\mf} = \sum_{a=1}^{N} n_{a} \, , \qquad {\mu}_{i} = n_{i} - n_{i+1}\, , \qquad i = 1, \ldots , N-1 \, .
  \label{eq:mmun}
\eeq
The action of $\gl_N$ on ${\CalH}_{\mf}^{\vec\mu}$ is represented by the ordinary
vector fields:
\beq
  T_{{\CalH}_{\mf}^{\vec\mu}} ({\bJ}_{a}^{b} ) = - {\rz}^{b} \frac{\pa}{\pa {\rz}^a} \, .
\eeq
\end{nott}

\begin{nott}
For ${\tilde\mf} \in \BC$ and $\vec{\tilde\mu} \in {\BC}^{N-1}$, define:
\beq
\label{eq:tomega included}
  {\tilde\CalH}_{\tilde\mf}^{\vec{\tilde\mu}} := {\tilde\omega}_{\tilde\bn} \cdot {\tilde H}_{\tilde\bn}
\eeq
with 
\beq
  {\tilde\mf} = \sum_{a=1}^{N} {\tilde n}_{a} \, , \qquad 
  {\tilde\mu}_{i} = {\tilde n}_{i} - {\tilde n}_{i+1}\, , \qquad  i = 1, \ldots , N-1 \, .
  \label{eq:mmunt}
\eeq
The action of $\gl_N$ on ${\tilde\CalH}_{\tilde\mf}^{\vec{\tilde\mu}}$ is represented by the ordinary vector fields:
\beq
  T_{{\tilde\CalH}_{\tilde\mf}^{\vec{\tilde\mu}}} ({\bJ}_{a}^{b} ) =  {\tz}_{a} \frac{\pa}{\pa {\tz}_b} \, .
\eeq
\end{nott}

\begin{rem}\label{rem:HW-weights}
(a) It is clear that the Casimirs $c_k \left( {\CalH}_{\mf}^{\vec\mu} \right)$ and 
$c_k \left( {\tilde\CalH}_{\tilde\mf}^{\vec{\tilde\mu}} \right)$, defined by \eqref{eq:casimir-values}, 
depend only on $\mf$ and $\tilde\mf$, respectively. 

\noindent
(b) The $\gl_N$-weight subspaces, i.e., the joint eigenspaces of a commuting family $\{{\bJ}_{a}^{a}\}_{a=1}^N$, 
of ${\CalH}_{\mf}^{\vec\mu}$ and ${\tilde\CalH}_{\tilde\mf}^{\vec{\tilde\mu}}$ are all one-dimensional, the corresponding sets 
of weights being $-{\bn}+{\Lambda}_{0}\subset \BC^N$ and ${\tilde\bn} + {\Lambda}_{0} \subset {\BC}^{N}$, respectively, where
${\Lambda}_{0}$ denotes the lattice ${\Lambda}_{0}=\left\{(r_1, \ldots, r_N)\in \BZ^N \Big| \sum_{i=1}^N r_i = 0\right\}$. 

\noindent
(c) The vectors 
  ${\Omega}_{{\CalH}_{\mf}^{\vec\mu}} := {\omega}_{\bn} \in {\CalH}_{\mf}^{\vec\mu}$, 
  ${\tilde\Omega}_{{\tilde\CalH}_{\tilde\mf}^{\vec{\tilde\mu}}} := 
   {\tilde\omega}_{\tilde\bn} \in {\tilde\CalH}_{\tilde\mf}^{\vec{\tilde\mu}}$ have the 
following $\ssl_N$-weights:
\beq
  {\bh}_{i} \cdot {\Omega}_{{\CalH}_{\mf}^{\vec\mu}} =  - {\mu}_{i} \cdot {\Omega}_{{\CalH}_{\mf}^{\vec\mu}}\, , \qquad 
  {\bh}_{i} \cdot {\tilde\Omega}_{{\tilde\CalH}_{\tilde\mf}^{\vec{\tilde\mu}}} = 
  {\tilde\mu}_{i} \cdot {\tilde\Omega}_{{\tilde\CalH}_{\tilde\mf}^{\vec{\tilde\mu}}} \, , \qquad 
  i=1,\dots,N-1 \, .
\label{eq:hw1w}
\eeq
\end{rem} 


\subsection{Vermas and HW-modules in the $N=2$ case}\label{ssec sl2-digression}

The generators ${\er} \equiv {\er}_{1}, {\fr} \equiv {\fr}_{1}, {\bh} \equiv {\bh}_{1}$ of~$\ssl_2$, 
see~(\ref{eq:h-gen},~\ref{eq:efh-gen}), obey the standard relations:
\beq
  [{\er}, {\fr} ] = {\bh}\, , \quad 
  [{\bh}, {\er} ] = 2{\er}\, , \quad 
  [{\bh}, {\fr}] = -2 {\fr} \, .
\eeq
For $a,s \in {\BC}$ and $i\in \{-1,0,1\}$, consider the differential operators:
\beq
  L_{i} = - z^{i+1} {\pa}_{z} + ( a + (i+1)s) z^{i}\ , 
\label{eq:sl2rep}
\eeq
obeying the commutation relations: 
\beq
  [L_i, L_j] = (i-j) L_{i+j} \, . 
\label{eq:sl2comm}
\eeq
The assignments 
\beq
  {\er} \mapsto -L_{-1}\, ,\quad {\fr} \mapsto L_{1}\, ,\quad  {\bh} \mapsto 2 L_0\, , 
\label{eq:ztrep} 
\eeq 
or 
\beq
  {\er} \mapsto - L_{1}\, ,\quad {\fr} \mapsto L_{-1}\, ,\quad  {\bh} \mapsto - 2 L_0\, , 
\label{eq:zrep}
\eeq
represent $\ssl_2$ by the first-order differential operators on a line. 

The modules we defined in the general $N$ case can be described quite explicitly. 
Specifically, the highest/lowest weight Verma and the twisted HW  $\ssl_2$-modules
are all realized in the spaces of the twisted tensors:
\beq
\label{eq:twisted tensors}
  f(z) z^{-a}dz^{-s}\, ,  
\eeq
with $f(z)$ being a single-valued function of $z \in {\BC}^{\times}$, so that the operators \eqref{eq:sl2rep} 
are the infinitesimal fractional linear transformations:
\beq
  z \mapsto \frac{A z+B}{C z+D} \, , \qquad 
  \left( \begin{matrix} A & B \\ C & D \\ \end{matrix} \right) \in SL(2, {\BC}) \, .
\label{eq:SL2g}
\eeq
To make this relation precise, let us start with the geometric descriptions of the Verma modules. 

\medskip

In the geometric realization of the lowest weight Verma modules, we have a two-component vector 
\beq
  U_1=\left(U^{1}_{1|1}, U^{2}_{1|1}\right)=:\left(u^1,u^2\right) , 
\label{eq:vec1}
\eeq
which is acted upon by the gauge ${\BC}^{\times}$-symmetry via  
$(u^1, u^2) \mapsto (t^{-1} u^{1}, t^{-1} u^{2})$. We look at the space of the locally defined functions $\Psi=\Psi(u^{1}, u^{2})$ 
which transform with weight $-\nu$ under the Lie algebra of the gauge ${\BC}^{\times}$-symmetry. More precisely,
following~\eqref{eq:psiuver} and the succeeding discussion, we look at $\Psi$ of the form:
\beq
\label{eq:Psi-N=2}
  {\Psi} (u^{1}, u^{2}) = \psi(z) \cdot \left( u^{1} \right)^{\nu} ,
\eeq
where $\psi$ is a polynomial and $z = u^2/u^1$ is the only coordinate $u^{(1)}_{1}$~\eqref{eq:key coordinate 1} in the present setting.
One can perceive the right-hand side of~\eqref{eq:Psi-N=2} as the local section of a complex power of a line bundle 
$\mathcal{O}(1)$ over a neighborhood of $z =0$ in ${\BC\BP}^{1}$, defined near the slice $u^1 = 1$.
The generators of $\ssl_2$ act via:
\beq
\begin{aligned}
  & {\er} = - u^{2} \frac{\partial}{\partial u^{1}} = z^{2} {\partial}_{z} - {\nu}z \, , \\
  & {\fr} = - u^{1} \frac{\partial}{\partial u^{2}} = - {\partial}_{z} \, , \\
  & {\bh} = u^{2} \frac{\partial}{\partial u^{2}} - u^{1} \frac{\partial}{\partial u^{1}} = 2 z {\partial}_{z} - {\nu} \, , \\
\end{aligned}
\label{eq:vermlw2}
\eeq
where the differential operators in the middle act on $\Psi$ while the rightmost ones act on $\psi=\psi(z)$. 
The vacuum is:
\beq
\label{eq:omega vs u}
  {\Omega}_{\nu} = ( u^{1} )^{\nu} ,
\eeq
corresponding to $\psi=1$, and the lowest weight Verma module is: 
\beq
  {\CalV}_{\nu} = {\BC}[{\er}] {\Omega}_{\nu} \, .
\eeq
The weight (eigenvalue of $\bh$) of the state $z^{n}$ is $2n-\nu$. Note that the fractional linear transformation~\eqref{eq:SL2g}
transforms $(u^1, u^2) \mapsto (Cu^2 + Du^1, A u^2 +  B u^1)$, hence it maps the vacuum to (again, we are working infinitesimally):
\beq
\label{eq:match1}
  ( C u^{2} + D u^{1} )^{\nu} = ( C z + D)^{\nu} {\Omega}_{\nu} \, .
\eeq
The formula~\eqref{eq:match1} allows us to match:
\beq
\label{eq:omega vs dz}
  {\Omega}_{\nu} \sim dz^{-\frac{\nu}{2}} \, .
\eeq
Thus, the lowest weight Verma module ${\CalV}_{\nu}$ corresponds to the 
realization~(\ref{eq:zrep},~\ref{eq:twisted tensors}) with:
\beq
\label{eq:lVerma via tensors}
  a = 0,\qquad s = \frac{\nu}{2} \, , 
\eeq
and with polynomial $f$ in~\eqref{eq:twisted tensors}. 

\medskip

In the geometric realization of the highest weight Verma modules, we have a two-component covector 
\beq
  {\tilde{U}}_1 = \left({\tilde U}^{1}_{1|1}, {\tilde U}^{1}_{1|2}\right) =: \left(v_{1}, v_{2}\right) , 
\label{eq:cov1}
\eeq
which is acted upon by the gauge $\BC^\times$-symmetry via $(v_1, v_2) \mapsto (t v_1, tv_2 )$. 
We are looking at the space of locally defined functions $\tilde\Psi = \tilde \Psi(v_{1}, v_{2})$, which transform 
with weight $\tilde\nu$ under the Lie algebra of the gauge ${\BC}^{\times}$-symmetry. More precisely,
following~(\ref{eq:vacua-highest},~\ref{eq:coordinates-highest}), we look at $\tilde\Psi$ of the form:
\beq
  {\tilde\Psi} (v_1, v_{2}) = 
  {\tilde\psi} ({\tilde z}) \cdot \left( v_{1} \right)^{\tilde\nu} ,
\eeq 
where ${\tilde\psi}$ is a polynomial and ${\tilde z} = v_{2}/v_{1}$ is the only coordinate 
${\tilde u}^{1}_{(1)}$~(\ref{eq:coordinates-highest}) in the present setting. 
The generators of $\ssl_2$ act via:
\beq
\begin{aligned}
  & {\er} = v_{1} \frac{\partial}{\partial v_{2}} =  {\partial}_{\tilde z}\, , \\
  & {\fr} = v_{2} \frac{\partial}{\partial v_{1}} =  
    - {\tilde z}^2 {\partial}_{\tilde z} + {\tilde\nu} {\tilde z}\, , \\
  & {\bh} = v_{1} \frac{\partial}{\partial v_{1}} - v_{2} \frac{\partial}{\partial v_{2}} = 
    - 2 {\tilde z}{\partial}_{\tilde z} + {\tilde\nu} \, , \\
\end{aligned}
\label{eq:vermlw2t}
\eeq
where the differential operators in the middle act on $\tilde\Psi$ while the rightmost ones act on $\tilde\psi=\tilde\psi({\tilde z})$.
The vacuum is:
\beq
\label{eq:tomega vs tu}
  {\tilde\Omega}_{\tilde\nu} = \left( v_{1} \right)^{\tilde\nu} , 
\eeq 
corresponding to $\tilde\psi=1$, and the highest weight Verma module is: 
\beq
  {\tilde\CalV}_{\tilde\nu} = {\BC} [ {\fr} ] {\tilde\Omega}_{\tilde\nu} \, . 
\eeq
The weight of the state ${\tilde z}^{n}$ is $- 2n + {\tilde\nu}$. Note that under the $SL(2,{\BC})$ fractional linear 
transformation~\eqref{eq:SL2g} the covector $(v_1, v_2)$ transforms via $(v_1, v_2) \mapsto (-B v_2  + A v_1, D v_2 -  C v_1)$ 
with $AD-BC=1$, so that the pairing ${\tilde U}_1\cdot U_1 = v \cdot u \equiv u^1 v_1 + u^2 v_2$ is invariant, leading to:
\beq
  {\tilde z} \mapsto \frac{D {\tilde z} - C}{- B {\tilde z}+A} \, . 
\label{eq:SL2gt}
\eeq
Thus, the vacuum $\tilde\Omega_{\tilde\nu}$ is transformed via:
\beq
  {\tilde\Omega}_{\tilde\nu} \mapsto \left( A v_{1} - B v_{2} \right)^{\tilde\nu} = 
  \left( A - B {\tilde z} \right)^{\tilde\nu} \, {\tilde\Omega}_{\tilde\nu}\, , 
\eeq 
which allows us to match:
\beq
\label{eq:tomega vs dtz}
  {\tilde\Omega}_{\tilde\nu} \sim d{\tilde z}^{-\frac{\tilde\nu}{2}} \, .
\eeq
Hence, the highest weight Verma module ${\tilde\CalV}_{\tilde\nu}$ corresponds to the 
realization~(\ref{eq:ztrep},~\ref{eq:twisted tensors}) with: 
\beq
  a = 0, \qquad s = \frac{\tilde\nu}{2} \, , 
\eeq
and with polynomial $f$ in~\eqref{eq:twisted tensors}. 

We note that the transformations \eqref{eq:SL2g} and \eqref{eq:SL2gt} are related via ${\tilde z}z = -1$, so that we get 
an equivalent representation~(\ref{eq:zrep},~\ref{eq:twisted tensors}) with: 
\beq
\label{eq:hVerma via tensors}
  a = {\tilde\nu}, \qquad s = \frac{\tilde\nu}{2} \, . 
\eeq 

\medskip
 
Finally, to describe the twisted HW-modules $H_{\bn}$, $\tilde{H}_{\tilde\bn}$ with ${\bn}= (n_{1}, n_{2})$, 
${\tilde\bn} = ({\tilde n}_{1}, {\tilde n}_{2})$, we recall the notation of~\eqref{eq:omegas}:
\beq
  {\omega}_{\bn} = \left( {\rz}^1 \right)^{n_1} \left( {\rz}^2 \right)^{n_2}, \qquad
  {\tilde\omega}_{\tilde\bn} = {\tz}_{1}^{{\tilde n}_{1}}{\tz}_{2}^{{\tilde n}_{2}} \, .
\eeq
The vector space underlying $H_{\bn}$ is the space of Laurent polynomials $\psi$ in $z={\rz}^{2}/{\rz}^{1}$. 
Analogously, the vector space underlying ${\tilde H}_{\tilde\bn}$ is the space of Laurent polynomials  $\tilde\psi$ 
in ${\tilde z}={\tz}_{2}/{\tz}_{1}$. 
 
In the first case, the generators of $\ssl_2$ act via: 
\beq
\begin{aligned}
  & {\er} = 
    - {\omega}_{\bn}^{-1} \left( {\rz}^2 \frac{\partial}{\partial {\rz}^1} \right) {\omega}_{\bn} =  z^2 {\partial}_{z}- n_{1}  z  
    \, , \\
  & {\fr} = 
    - {\omega}_{\bn}^{-1} \left( {\rz}^{1} \frac{\partial}{\partial {\rz}^2} \right) {\omega}_{\bn}  =  - {\partial}_{z} - n_{2} z^{-1}
    \, , \\
  & {\bh} = 
    {\omega}_{\bn}^{-1} \left( {\rz}^2 \frac{\partial}{\partial {\rz}^2} - {\rz}^{1} \frac{\partial}{\partial {\rz}^1} \right)
    {\omega}_{\bn} = 2 z {\partial}_{z} + n_{2} - n_{1} \, .
\end{aligned}
\eeq
Thus, the twisted HW-module $H_{\bn} \sim {\CalH}_{2s}^{2(s+a)}$ corresponds to 
the realization~(\ref{eq:zrep},~\ref{eq:twisted tensors}) with:
\beq 
\label{eq:HW1 via tensors}
  a = - n_{2}, \qquad s = \frac{n_{1} + n_{2}}{2} \, . 
\eeq
In the second case, analogously, the generators of $\ssl_2$ act via:
\beq
\begin{aligned}
  & {\er} =  {\tilde\omega}_{\tilde\bn}^{-1} \left( {\tz}_1 \frac{\partial}{\partial {\tz}_2} \right) {\tilde\omega}_{\tilde\bn} = 
    {\partial}_{\tilde z} +  {\tilde n}_{2}  {\tilde z}^{-1}  \, , \\
  & {\fr} = {\tilde\omega}_{\tilde\bn}^{-1} \left( {\tz}_{2} \frac{\partial}{\partial {\tz}_{1}} \right) {\tilde\omega}_{\tilde\bn} =  
    - {\tilde z}^{2}{\partial}_{\tilde z} + {\tilde n}_{1} {\tilde z} \, , \\
  & {\bh} = {\tilde\omega}_{\tilde\bn}^{-1} 
    \left( {\tz}_{1} \frac{\partial}{\partial {\tz}_{1}} - {\tz}_{2} \frac{\partial}{\partial {\tz}_{2}} \right)
    {\tilde\omega}_{\tilde\bn} = - 2 {\tilde z} {\partial}_{\tilde z} + {\tilde n}_{1} - {\tilde n}_{2} \, .
\end{aligned}
\eeq
Thus, the twisted HW-module $\tilde{H}_{\tilde\bn} \sim {\tilde\CalH}_{2s}^{2(s+a)}$ corresponds to the 
realization~(\ref{eq:ztrep},~\ref{eq:twisted tensors})  with: 
\beq
\label{eq:HW2 via tensors}
  a = - {\tilde n}_{2}, \qquad s = \frac{{\tilde n}_{1} + {\tilde n}_{2}}{2}\, . 
\eeq


\subsection{Tensor products and invariants}\label{ssec Invariants}

Let us recall the following $SL(2,{\BC})$-invariants (under the fractional linear action) on the configurations of 
$2$, $3$, and $4$ points on ${\BC\BP}^1$:
\beq
  {\upsilon}(z_1, z_2) = \frac{dz_{1} \otimes dz_{2}}{(z_{1}-z_{2})^2}
\eeq
is an invariant $(1,0) \otimes (1,0)$--form on ${\BC\BP}^1 \times {\BC\BP}^1$,
\beq
  \frac{z_{2}-z_{1}}{(z_{3}-z_{1})(z_{3}-z_{2})} dz_{3} = 
  \left( \frac{{\upsilon}(z_1, z_3) \otimes {\upsilon}(z_2, z_3)}{{{\upsilon}}(z_1, z_2)} \right)^{\frac 12}
\eeq
is an invariant $0 \otimes 0 \otimes (1,0)$--form on ${\BC\BP}^1 \times {\BC\BP}^1 \times {\BC\BP}^1$, and finally, 
the cross-ratio
\beq
  [z_1, z_2; z_3, z_4] := \frac{z_{2}-z_{1}}{z_{3}-z_{1}} \cdot \frac{z_{4}-z_{3}}{z_{4}-z_{2}} = 
  \left( \frac{{\upsilon}(z_1, z_3) \otimes {\upsilon}(z_2, z_4)}{{\upsilon}(z_1, z_2) \otimes {\upsilon}(z_3, z_4)} \right)^{\frac 12}
\label{eq:crr}
\eeq
is an invariant meromorphic function on ${\BC\BP}^1 \times {\BC\BP}^1 \times {\BC\BP}^1 \times {\BC\BP}^1$.

\medskip 

Thus,
\beq
  I_{\nu}^{(2)} = {\upsilon}(z_1, z_2)^{-\frac{\nu}{2}} = 
  (1+ z_{1}{\tilde z}_{2})^{\nu}\, (dz_1)^{-\frac{\nu}{2}}  \otimes (d{\tilde z}_2)^{-\frac{\nu}{2}}
\label{eq:2inv}
\eeq
is an $\ssl_2$-invariant element in the completed tensor product ${\CalV}_{\nu} {\hat\otimes} {\tilde\CalV}_{\nu}$. 
More precisely, we need to view \eqref{eq:2inv} as a power series in $z_1, {\tilde z}_2 = - z_{2}^{-1}$ 
in the domain $z_1\to 0, z_2\to \infty$:
\beq
  I_{\nu}^{(2)} \vert_{|z_{1}|\ll |z_{2}|} \,  \in \, \left( {\CalV}_{\nu} {\hat\otimes} {\tilde\CalV}_{\nu} \right)^{\ssl_2} .
\eeq
For another domain of convergence, e.g., $z_1 \to \infty, z_2 \to 0$, the expression \eqref{eq:2inv} would define an invariant 
in the completed tensor product ${\tilde\CalV}_{\nu} {\hat\otimes} {\CalV}_{\nu}$ instead:
\beq
  I_{\nu}^{(2)} \vert_{|z_{2} | \ll |z_{1}|} \,  \in \, \left( {\tilde\CalV}_{\nu} {\hat\otimes} {\CalV}_{\nu} \right)^{\ssl_2} .
\eeq
Finally, invoking~(\ref{eq:omega vs u},~\ref{eq:omega vs dz},~\ref{eq:tomega vs tu},~\ref{eq:tomega vs dtz}), 
we can express $I_{\nu}^{(2)}$~\eqref{eq:2inv} in terms of $U_1, {\tilde U}_1$~(\ref{eq:vec1},~\ref{eq:cov1}):
\begin{multline}
  I_{\nu}^{(2)} = 
  \left( {\tilde U}_{1} \cdot U_{1} \equiv U^{1}_{1|1} {\tilde U}^{1}_{1|1} + U^{2}_{1|1} {\tilde U}^{1}_{1|2} 
         \equiv u^1 v_1 + u^2 v_2 \right)^{\nu} \, \\
         =
  {\Omega}_{\nu} {\tilde\Omega}_{\nu} \, \times \, 
  \left( {\rm power \ series\ in}\ z=u^{2}/u^{1}, {\tilde z}=v_2/v_1 \right) \, .
\label{eq:2inv2}
\end{multline}
The benefit of formula~(\ref{eq:2inv2}) is that it admits a natural generalization to the general $N$:
\beq
  I_{\vec\nu}^{(2)} = \prod_{i=1}^{N-1}  {\tilde\pi}^{i}({\pi}_{i})^{\nu_i} =
  {\Omega}_{\vec\nu} {\tilde\Omega}_{\vec\nu} \, \times \, 
  \left( {\rm power \ series\ in}\ u^{(i)}_{k}, {\tilde u}_{(i)}^{k} \right)
  \in 
  \left( {\CalV}_{\vec\nu} {\hat\otimes} {\tilde\CalV}_{\vec\nu} \right)^{\gl_N} .
\label{eq:2invN}
\eeq

\begin{rem}
In coordinates, we have: 
\beq
  {\tilde\pi}^{i}({\pi}_{i}) = 
  {\rm Det} \left( {\tilde U}_{i} {\tilde U}_{i+1} \ldots {\tilde U}_{N-1} U_{N-1} \ldots U_{i+1} U_{i} \right).
\eeq
\end{rem}

\begin{rem}\label{vermapair}
The formula \eqref{eq:2invN} determines the unique $\gl_N$-invariant bilinear pairing:
\beq
\label{eq:inv-pairing}
  (\cdot,\cdot)_{\vec\nu}\colon {\CalV}_{\vec\nu} \times {\tilde\CalV}_{\vec\nu} \, \longrightarrow \, {\BC}
\eeq
such that
\beq
  \left( {\Omega}_{\vec\nu}, {\tilde\Omega}_{\vec\nu} \right)_{\vec\nu} = 1 \, .
\eeq
One can present $(\cdot,\cdot)_{\vec\nu}$ as an integral over $F(W)$, 
but the quicker way is the following: the matrix $G_{\vec n, {\vec{\tilde n}}}$ inverse to
\beq
  \left( \prod_{k \leq i} \left( u^{(i)}_{k} \right)^{n^{(i)}_{k}}{\Omega}_{\vec\nu} \, , \,  
         \prod_{k \leq i} \left({\tilde u}_{(i)}^{k} \right)^{{\tilde n}_{(i)}^{k}} {\tilde\Omega}_{\vec\nu} \right)_{\vec\nu} 
\eeq
is given by the coefficients of the expansion 
\beq
  {\CalI}_{\vec\nu} = 
  \prod_{i=1}^{N-1} \left( \frac{{\tilde\pi}^{i}({\pi}_{i})}{{\tilde\pi}^{i}({\pi}_{i}^{0}) \cdot 
  {\tilde\pi}^{i}_{0} ({\pi}_{i})} \right)^{{\nu}_{i}} = 
  \sum_{{\vec n}, {\vec{\tilde n}}} G_{\vec n , {\vec{\tilde n}}} 
  \prod_{1 \leq k \leq i \leq N-1} \left( u^{(i)}_{k} \right)^{n^{(i)}_{k}}\left( {\tilde u}_{(i)}^{k} \right)^{{\tilde n}_{(i)}^{k}} 
  = 1 + \ldots
\eeq
\end{rem}

\medskip

Let us now similarly produce an $\ssl_2$-invariant in the completed tensor product of three $\ssl_2$-representations: 
the lowest weight and the highest weight Vermas, as well as the twisted HW-module. To this end, we consider:
\beq
  I^{(3)}_{\nu_1, \nu_2, \nu_3} = 
  {\upsilon}(z_1, z_2)^{-\frac{\nu_1 + \nu_2 - \nu_3}{4}} {\upsilon}(z_1, z_3)^{-\frac{\nu_1+\nu_3-\nu_2}{4}} 
  {\upsilon}(z_2, z_3)^{-\frac{\nu_2+\nu_3-\nu_1}{4}} \, .
\label{eq:3inv}
\eeq
By invoking~(\ref{eq:lVerma via tensors},~\ref{eq:hVerma via tensors},~\ref{eq:HW1 via tensors}) and expanding~\eqref{eq:3inv} 
in the region $|z_1|\ll |z_2| \ll |z_3|$, we arrive at the following interpretation:
\beq
  I^{(3)}_{{\nu}_1, {\nu}_2, {\nu}_3} \vert_{|z_{1} | \ll |z_{2}| \ll |z_{3}|}\, \in \, 
  \left( {\CalV}_{\nu_1} {\hat\otimes} {\CalH}_{\nu_2}^{{\nu}_{3}-{\nu}_{1}} {\hat\otimes} {\tilde\CalV}_{\nu_3} \right)^{\ssl_2} .
\eeq
Finally, in the $\left(u^1,u^2\right), \left({\rz}^1,{\rz}^2\right), \left(v_1,v_2\right)$-realizations, 
this invariant takes the following form:
\begin{equation}
\begin{split}
  & I^{(3)}_{{\nu}_1, {\nu}_2, {\nu}_3} = 
    \left( u^{1} {\rz}^{2} -  u^{2} {\rz}^{1} \right)^{\frac{\nu_1 + \nu_2 - \nu_3}{2}}
    \left( v_{1} {\rz}^{1} + v_{2} {\rz}^{2} \right)^{\frac{\nu_2+\nu_3-\nu_1}{2}} 
    \left( u^{1} v_{1} + u^{2} v_{2} \right)^{\frac{\nu_1+\nu_3-\nu_2}{2}} = \\
  & \qquad \qquad \quad
    {\Omega}_{\nu_1} {\tilde\Omega}_{\nu_3} ({\rz}^{1})^{n_{1}} ({\rz}^{2})^{n_{2}} \, \times \, 
    \left( {\rm power \ series\ in}\ z=u^{2}/u^{1}, {\tilde z}=v_2/v_1, ({\rz}^2/{\rz}^{1})^{\pm 1} \right)
\end{split}
\label{eq:3in}
\end{equation}
with 
\beq
  n_{1} = \frac{\nu_2+\nu_3-\nu_1}{2}\, , \qquad n_{2} = \frac{\nu_1 + \nu_2 - \nu_3}{2}\, ,  
\eeq
where we matched $z_1 \sim z,\, z_2 \sim {\rz}^2/{\rz}^{1},\, z_3 \sim -1/{\tilde z}$.
We note that the last two factors in~\eqref{eq:3in} are $\gl_2$-invariant, while the first one is only $\ssl_2$-invariant. 

The formula~\eqref{eq:3in} admits a natural generalization to the general $N$, with the triple  ${\nu}_{1}, {\nu}_{2}, {\nu}_{3}$ 
being replaced with ${\vec\nu}_{1}, {\vec\nu}_{3} \in {\BC}^{N-1},\, {\nu}_2 \in {\BC}$. In this case, we have a unique invariant
(cf.~\eqref{eq:topf}):
\begin{equation}
\label{eq:3inntt}
\begin{split}
  & I^{(3)}_{{\vec\nu}_{1}, {\nu}_2, {\vec\nu}_{3}} =  
    \prod_{a=1}^{N} \, {\tilde\pi}^{a} \left( {\pi}_{a-1} \wedge {\bz}\right)^{n_{a}}\, \cdot 
    \prod_{i=1}^{N-1} \, {\tilde\pi}^{i} \left( {\pi}_{i} \right)^{{\nu}_{3,i}-n_{i}} = \\
  & \qquad \qquad \ \ 
    {\Omega}_{{\vec\nu}_1} \left( \prod_{a=1}^{N} ({\rz}^{a})^{n_{a}}  \right) {\tilde\Omega}_{{\vec\nu}_3} \,  \times \, 
    \left( {\rm power\ series \ in} \ u^{(i)}_{k}, \, {\tilde u}^{k}_{(i)}, \, {\rz}^a/{\rz}^{b} \right) \\
  &  \qquad \qquad  \qquad \qquad  \qquad \qquad  \qquad \qquad  \qquad \qquad   
    \in \left( {\CalV}_{\vec\nu_1} {\hat\otimes} {\CalH}_{\nu_2}^{{\vec\nu}_{3}-{\vec\nu}_{1}} {\hat\otimes} 
    {\tilde\CalV}_{\vec\nu_3}    \right)^{\ssl_N} , 
\end{split}
\end{equation}
where the vector ${\bn} = (n_{1}, \ldots , n_{N}) \in {\BC}^{N}$ is determined from 
\beq
  \sum_{a=1}^{N} n_{a} = {\nu}_{2} 
\eeq
and 
\beq
  n_{i+1} - n_i = {\nu}_{1,i}  - {\nu}_{3,i}\, , \qquad i = 1, \ldots , N-1 \, . 
\label{eq:muvec13}
\eeq
Similarly to the $N=2$ case, the factor ${\tilde\pi}^{N} \left( {\pi}_{N-1} \wedge {\bz}\right)^{n_N}$ is 
only $\ssl_N$-invariant, while all other factors in~\eqref{eq:3inntt} are naturally $\gl_N$-invariant.

Another generalization of \eqref{eq:3in} is the invariant
\begin{equation}
\label{eq:3innt}
\begin{split}
  & {\tilde I}^{(3)}_{{\vec\nu}_{1}, {\nu}_2, {\vec\nu}_{3}} =  
    \prod_{a=1}^{N} \, \left( {\tilde\pi}^{a-1} \wedge {\tilde\bz} \left(  {\pi}_{a} \right) \right) ^{{\tilde n}_{a}}\, \cdot
    \prod_{i=1}^{N-1} \, {\tilde\pi}^{i} \left( {\pi}_{i} \right)^{{\nu}_{1,i}-{\tilde n}_{i}} = \\
  & \qquad \qquad \ \ 
    {\Omega}_{{\vec\nu}_1} \left( \prod_{a=1}^{N} {\tz}_{a}^{{\tilde n}_{a}}  \right) {\tilde\Omega}_{{\vec\nu}_3} \, \times \, 
    \left( {\rm power \ series \ in} \ u^{(i)}_{k}, \, {\tilde u}^{k}_{(i)}, \, {\tz}_b/{\tz}_{a} \right) \\
  &  \qquad \qquad  \qquad \qquad  \qquad \qquad  \qquad \qquad  \qquad \qquad     
      \in \left( {\CalV}_{\vec\nu_1} {\hat\otimes} {\tilde\CalH}_{\nu_2}^{{\vec\nu}_{1}-{\vec\nu}_{3}} {\hat\otimes}
      {\tilde\CalV}_{\vec\nu_3} \right)^{\ssl_N} ,
\end{split}
\end{equation}
where the vector ${\tilde\bn} = ({\tilde n}_{1}, \ldots , {\tilde n}_{N}) \in {\BC}^{N}$ is determined from 
\beq
  \sum_{a=1}^{N} {\tilde n}_{a} = {\nu}_{2} 
\eeq
and
\beq
  {\tilde n}_{i+1} - {\tilde n}_i = {\nu}_{3,i}  - {\nu}_{1,i}\, , \qquad i = 1, \ldots , N-1 \, . 
\label{eq:muvec31}
\eeq

\begin{rem} 
The examples~(\ref{eq:3inntt},~\ref{eq:3innt}) demonstrate the need for twists in the definition of the HW-modules 
in Section~\ref{sssec:hwmod}. 
\end{rem}

To prove that $I^{(2)}$ of~\eqref{eq:2invN}, $I^{(3)}$ of~\eqref{eq:3inntt}, and ${\tilde I}^{(3)}$ of~\eqref{eq:3innt} 
are the only invariants in the corresponding (completed) tensor products of $2$ and $3$ modules of $\ssl_N$, see 
Corollary~\ref{cor:inv-uniqueness}, let us recall the realization of the corresponding spaces of invariants as 
the weight subspaces. 

\begin{nott}
For an $\ssl_N$-module $W$ and ${\vec\lambda} \in {\BC}^{N-1}$, we denote by $W [ {\vec\lambda} ]$ the weight $\vec\lambda$ subspace:
\beq
  w \in W[{\vec\lambda}] \, \Leftrightarrow\, {\bh}_{i}(w)  = {\lambda}_{i} \cdot w\, , \qquad i = 1, \ldots , N-1 \, .
\label{eq:wss}
\eeq
\end{nott}

\begin{rem}
We have (cf.~Remark~\ref{rem:HW-weights}):
\beq
  {\CalH}_{\mf}^{\vec\mu} [ - {\vec\mu}  ] = {\BC} \cdot {\omega}_{\bn} \, , \qquad 
  {\tilde\CalH}_{\tilde\mf}^{\vec{\tilde\mu}} [ {\vec{\tilde\mu}}  ] = {\BC} \cdot {\tilde \omega}_{\tilde \bn} \, . 
\eeq
\end{rem}

To Verma modules ${\CalV}_{\vec\nu}, {\tilde\CalV}_{\vec{\tilde\nu}}$ defined in Sections \ref{ssec:hwverma}, \ref{ssec:lwverma}, 
we associate the \emph{restricted dual} modules ${\CalV}_{\vec\nu}^{*}, {\tilde\CalV}_{\vec{\tilde\nu}}^{*}$. These are defined 
as the submodules of ${\rm Hom}_{\BC}({\CalV}_{\vec\nu}\, , \, {\BC})$, 
${\rm Hom}_{\BC}({\tilde\CalV}_{\vec{\tilde\nu}}\, , \, {\BC})$, respectively, whose underlying vector spaces are 
direct sums of the spaces, dual to the $\ssl_N$-weight subspaces of ${\CalV}_{\vec\nu}, {\tilde\CalV}_{\tilde{\vec\nu}}$.
The following is well-known:

\begin{lem}\label{duality of Verma}
If ${\CalV}_{\vec\nu}$ (resp.\ ${\tilde{\CalV}}_{\vec{\tilde\nu}}$) is an irreducible $\ssl_N$-module, 
then ${\CalV}_{\vec\nu}^*\simeq {\tilde\CalV}_{\vec\nu}$ 
(resp.\ ${\tilde\CalV}_{\vec{\tilde\nu}}^{*} \simeq {\CalV}_{\vec{\tilde\nu}}$).
\end{lem}

\noindent
For any $\ssl_N$-module $W$, we define the completed tensor products ${\CalV}_{\vec\nu} \hat{\otimes} W$ and  
${\tilde\CalV}_{\vec{\tilde\nu}} \hat{\otimes} W$ via:
\beq\label{completed tensor product}
  {\CalV}_{\vec\nu} \hat{\otimes} W := \Hom_{\BC} \left( {\CalV}_{\vec\nu}^{*}\, ,\, W\right)\, , \qquad
  {\tilde\CalV}_{\vec{\tilde\nu}} \hat{\otimes} W := \Hom_{\BC} \left( {\tilde\CalV}_{\vec{\tilde\nu}}^{*}\, ,\, W\right) ,
\eeq
both of which have natural structure of $\ssl_N$-modules.

Now we are ready to invoke the standard interpretation of the space of $\ssl_N$-invariants in the tensor product, 
completed in the sense of~(\ref{completed tensor product}), of $\ssl_N$-modules involving both the highest weight 
and the lowest weight Verma modules (cf.~the proof of~\cite[Proposition 1.1]{E1}):

\begin{lem}\label{key property}
If the lowest weight Verma ${\CalV}_{\vec\nu}$ and the highest weight Verma ${\tilde\CalV}_{\vec{\tilde\nu}}$ modules of $\ssl_N$
are irreducible, then the space of $\ssl_N$-invariants in 
${\CalV}_{\vec\nu} \hat{\otimes} W \hat{\otimes} {\tilde\CalV}_{\vec{\tilde\nu}}$ can be described as follows:
\beq
  \left( {\CalV}_{\vec\nu} \hat{\otimes} W \hat{\otimes} {\tilde\CalV}_{\vec{\tilde\nu}} \right)^{\ssl_N} \simeq \, 
  W\left[ \, {\vec\nu} - {\vec{\tilde\nu}} \, \right] \, .
\eeq
\end{lem}

\begin{proof}
This follows from the following sequence of canonical identifications:
\begin{equation}\label{canonical identifications}
\begin{split}
  & \left( {\CalV}_{\vec\nu} \hat{\otimes} W  \hat{\otimes} {\tilde\CalV}_{\vec{\tilde\nu}} \right)^{\ssl_N} \, \simeq \ 
    \Hom_{\ssl_N} \left( {\CalV}_{\vec\nu}^*, W\hat{\otimes}{\tilde{\CalV}}_{\vec{\tilde\nu}} \right) \, \simeq \\
  & \qquad \qquad \qquad 
    \Hom_{\ssl_N} \left( {\tilde{\CalV}}_{\vec\nu}, W\hat{\otimes} {\tilde{\CalV}}_{\vec{\tilde\nu}} \right) \, \simeq\,
    \left( W\hat{\otimes} {\tilde{\CalV}}_{\vec{\tilde\nu}} \right)^{\n_+} [{\vec\nu}]  \, \simeq \\
  & \qquad \qquad \qquad 
    \Hom_{\n_+} \left( {\tilde{\CalV}}_{\vec{\tilde\nu}}^* , W \right) [{\vec\nu}] \, \simeq\,
    \Hom_{\n_+} \left( {\CalV}_{\vec{\tilde\nu}} , W \right) [{\vec\nu}] \, \simeq \,
    W\left[{\vec\nu}-{\vec{\tilde\nu}}\right] 
\end{split}
\end{equation}
by using the conventions~(\ref{completed tensor product}), Lemma~\ref{duality of Verma}, and Frobenius reciprocity.
\end{proof}

\begin{rem}
Putting together the identifications~(\ref{canonical identifications}), we see that the resulting vector space isomorphism
\begin{equation}\label{invariants via weight space}
  \X \colon 
  \left( {\CalV}_{\vec\nu} \hat{\otimes} W \hat{\otimes} {\tilde\CalV}_{\vec{\tilde\nu}} \right)^{\ssl_N}
    \,\iso\,  W\left[{\vec\nu}-{\vec{\tilde\nu}}\right]
\end{equation}
is obtained by pairing an element of 
  $\left( {\CalV}_{\vec\nu} \hat{\otimes} W \hat{\otimes} {\tilde\CalV}_{\vec{\tilde\nu}} \right)^{\gl_N}$
with 
  $\tilde{\Omega}_{{\vec\nu}} \otimes {\Omega}_{{\vec{\tilde\nu}}} \in 
   \tilde{\CalV}_{{\vec\nu}} \otimes {\CalV}_{{\vec{\tilde\nu}}}$
with respect to $(\cdot,\cdot)_{\vec{\nu}}$ and $(\cdot,\cdot)_{{\vec{\tilde\nu}}}$ in the first and third tensor factors,
cf.~Remark~\ref{vermapair} and Lemma~\ref{duality of Verma}. 
\end{rem}

Applying Lemma~\ref{key property} to the trivial and the twisted HW-modules of $\ssl_N$, we obtain:

\begin{cor}\label{cor:inv-uniqueness}
(a) For the trivial $\ssl_N$-module $W=\BC$, the space of invariants 
$\left( {\CalV}_{\vec{\nu}_1}  \hat{\otimes} {\tilde\CalV}_{\vec{\nu}_2} \right)^{\ssl_N}$ vanishes if 
$\vec{\nu}_1\ne \vec{\nu}_2$, and is one-dimensional (hence, is spanned by $I^{(2)}_{\vec{\nu}_1}$ of~\eqref{eq:2invN}) 
if $\vec{\nu}_1=\vec{\nu}_2$. 

\noindent
(b) For the twisted HW-modules $W={\CalH}_{\nu_2}^{{\vec\mu}}, {\tilde{\CalH}}_{\nu_2}^{{\vec{\tilde\mu}}}$, the spaces of invariants
  $\left( {\CalV}_{\vec\nu_1} {\hat\otimes} {\CalH}_{\nu_2}^{\vec\mu} {\hat\otimes} {\tilde\CalV}_{\vec\nu_3} \right)^{\ssl_N}$
and 
 $\left( {\CalV}_{\vec\nu_1} {\hat\otimes} {\tilde{\CalH}}_{\nu_2}^{\vec{\tilde\mu}} 
  {\hat\otimes} {\tilde\CalV}_{\vec\nu_3} \right)^{\ssl_N}$  
are at most one-dimensional, and they vanish if $\vec\mu+\vec{\nu}_1-\vec{\nu}_3\notin \BZ^{N-1}$, 
$\vec{\tilde\mu}+\vec{\nu}_3-\vec{\nu}_1\notin \BZ^{N-1}$, respectively. In particular, the invariants
  $I^{(3)}_{{\vec\nu}_{1}, {\nu}_2, {\vec\nu}_{3}}\in 
   \left( {\CalV}_{\vec\nu_1} {\hat\otimes} {\CalH}_{\nu_2}^{{\vec\nu}_{3}-{\vec\nu}_{1}} {\hat\otimes} 
   {\tilde\CalV}_{\vec\nu_3} \right)^{\ssl_N}$ 
and
 ${\tilde I}^{(3)}_{{\vec\nu}_{1}, {\nu}_2, {\vec\nu}_{3}}\in 
   \left( {\CalV}_{\vec\nu_1} {\hat\otimes} {\tilde{\CalH}}_{\nu_2}^{{\vec\nu}_{1}-{\vec\nu}_{3}} {\hat\otimes} 
   {\tilde\CalV}_{\vec\nu_3} \right)^{\ssl_N}$ 
of~\eqref{eq:3inntt} and~\eqref{eq:3innt} are unique, up to scalar multipliers.
\end{cor}


\subsection{Our quartet}\label{ssec:quartet}

We are now finally ready to relate (\ref{eq:ups},~\ref{eq:phiups}) to the invariants in the completed tensor products 
of four $\ssl_N$-modules: the two Vermas and the two twisted HW-modules. 

Let us fix $\vec\nu, {\vec{\tilde\nu}}, {\vec\gamma} \in {\BC}^{N-1}$, and ${\mf}, {\tilde\mf} \in {\BC}$. 
Let us specify four $\ssl_N$-representations as follows:
\beq
  V_{1} = {\CalV}_{\vec\nu}\, , \ V_{2} = {\CalH}_{\mf}^{\vec\gamma - \vec\nu}\, , \ 
  V_{3}= {\tilde\CalH}_{{\tilde\mf}}^{\vec\gamma - \vec{\tilde\nu}}\, , \ V_{4} = {\tilde{\CalV}}_{\vec{\tilde\nu}}\, . 
\label{eq:4reps}
\eeq
We shall work with the completion
\begin{equation*}
  V_1 \hat{\otimes} V_2 \hat{\otimes} V_3 \hat{\otimes} V_4\, ,
\end{equation*}
so defined (cf.~\eqref{completed tensor product}) that it contains the power series expansion 
in $u^{(i)}_{k}, {\tilde u}_{(i)}^{k}, {\rz}^a, {\tz}_a$ of $\Upsilon$ given by \eqref{eq:ups}.

{}Let us now apply Lemma~\ref{key property} to the case $W=V_2\otimes V_3$. Noticing that
\begin{equation}
\label{eq:W-description}
  W\simeq \Big\{\, f \ \Big| \  
    f\in {\BC} \left[ \, ({\rz}^{1})^{\pm 1}, \ldots, ({\rz}^{N})^{\pm 1} , {\tz}_1^{\pm 1} , \ldots , {\tz}_{N}^{\pm 1} \, \Big]\, ,\,
    \deg_{\rz}(f)=\deg_{\tz}(f)=0\right\} ,
\end{equation}
with the $\ssl_N$-action~(\ref{eq:hwmod1},~\ref{eq:hwmod2}) twisted by the factors~\eqref{eq:omegas}, 
we get the following identification:
\begin{equation}
\label{weight space realization}
  \left( V_1 \hat{\otimes} V_2 \hat{\otimes} V_3 \hat{\otimes} V_4 \right)^{\ssl_N} \simeq
  W \left[ {\vec\nu}-{\vec{\tilde\nu}} \right]\simeq 
    \BC \left[ \eta_1^{\pm 1},\ldots, \eta_{N-1}^{\pm 1} \right] , 
\end{equation}
where the variables $\eta_i$'s are defined via:
\beq
\label{eq:eta-variables}
  \eta_i:=\frac{{\rz}^{i+1}{\tz}_{i+1}}{{\rz}^{i}{\tz}_{i}}\ \, , \qquad 1\leq i\leq N-1\, .
\eeq
The above vector space isomorphism 
 ${\BC}\left[ \eta_1^{\pm 1},\ldots, \eta_{N-1}^{\pm 1} \right] \iso 
  \left( V_1 \hat{\otimes} V_2 \hat{\otimes} V_3 \hat{\otimes} V_4 \right)^{\ssl_N}$
is constructive. Explicitly, given $\vec{r}=(r_{1},\ldots,r_{N-1})\in \BZ^{N-1}$, 
define the $\ssl_N$-weight $\vec{\delta}=(\delta_1,\ldots,\delta_{N-1})\in \BZ^{N-1}$ via 
$\delta_i=r_{i-1}-2r_i+r_{i+1}$ with $r_0=r_N=0$.
According to Lemma~\ref{key property}, the spaces of invariants 
 $\left( {\CalV}_{\vec\nu} {\hat\otimes} {\CalH}_{\mf}^{\vec\gamma-\vec\nu} 
  {\hat\otimes} {\tilde\CalV}_{\vec\gamma+\vec\delta} \right)^{\ssl_N}$
and 
 $\left( {\CalV}_{\vec\gamma+\vec\delta}\, {\hat\otimes} {\tilde{\CalH}}_{\tilde\mf}^{\vec\gamma-\vec{\tilde\nu}} 
  {\hat\otimes} {\tilde\CalV}_{\vec{\tilde\nu}} \right)^{\ssl_N}$  
are one-dimensional (for $\vec{r}=\vec{0}$, they are spanned by $I^{(3)}_{\vec{\nu},{\mf},\vec{\gamma}}$ and 
${\tilde I}^{(3)}_{\vec{\gamma},{\tilde \mf},\vec{\tilde\nu}}$). 
Equivalently, there are unique $\ssl_N$-module homomorphisms:
\beq
\begin{split}
   & \varphi_1\colon {\CalV}_{\vec\gamma+\vec\delta}\longrightarrow {\CalV}_{\vec\nu} {\hat\otimes} {\CalH}_{\mf}^{\vec\gamma-\vec\nu} \, ,\\
   & \varphi_2\colon {\tilde{\CalV}}_{\vec\gamma+\vec\delta} \longrightarrow 
   {\tilde{\CalH}}_{\tilde\mf}^{\vec\gamma-\vec{\tilde\nu}} {\hat\otimes} {\tilde\CalV}_{\vec{\tilde\nu}} \, ,
\end{split}
\eeq
such that 
\beq
  \left(\varphi_1(\Omega_{\vec{\gamma}+\vec\delta}), {\tilde \Omega}_{\vec\nu}\right)_{\vec\nu}=
  \prod_{a=1}^N {\left({\rz}^a\right)}^{r_{a-1}-r_a} \cdot \omega_{\bn}\, , \qquad
  \left( \Omega_{\vec{\tilde\nu}} , \varphi_2({\tilde\Omega}_{\vec{\gamma}+\vec\delta})\right)_{\vec{\tilde\nu}}=
  \prod_{a=1}^N {{\tz}_a}^{r_{a-1}-r_a} \cdot {\tilde\omega}_{\tilde\bn}\, ,
\eeq
cf.~(\ref{eq:omegas},~\ref{eq:omega included},~\ref{eq:mmun},~\ref{eq:tomega included},~\ref{eq:mmunt}), 
where we used Lemma~\ref{duality of Verma} and the pairing $(\cdot,\cdot)_{\vec\nu}, (\cdot,\cdot)_{\vec{\tilde\nu}}$ 
of Remark~\ref{vermapair} on the first and second components, respectively.
Hence, we get an $\ssl_N$-module homomorphism:
\beq
\label{eq:4-inv via 23-inv}
   \varphi:=\varphi_1\otimes \varphi_2\colon 
   {\CalV}_{\vec\gamma+\vec\delta}\, \hat{\otimes}{\tilde{\CalV}}_{\vec\gamma+\vec\delta} \longrightarrow 
   V_1 \hat{\otimes} V_2 \hat{\otimes} V_3 \hat{\otimes} V_4 \, .
\eeq
Invoking the $\ssl_N$-invariant 
  $I^{(2)}_{\vec\gamma+\vec\delta}\in \left( {\CalV}_{\vec\gamma+\vec\delta}\, \hat{\otimes} 
   {\tilde{\CalV}}_{\vec\gamma+\vec\delta} \right)^{\ssl_N}$,
we obtain the sought-after $\ssl_N$-invariant
\beq
  \varphi\left( I^{(2)}_{\vec\gamma+\vec\delta} \right)\in 
  \left( V_1 \hat{\otimes} V_2 \hat{\otimes} V_3 \hat{\otimes} V_4 \right)^{\ssl_N}\, ,
\eeq
which exactly corresponds to $\eta_1^{r_1}\eta_2^{r_2}\cdots \eta_{N-1}^{r_{N-1}}$ under 
the identification~\eqref{weight space realization}.

\begin{rem}
The realization \eqref{weight space realization} corresponds to the family (over $\qe$) of maps 
\begin{multline}
  {\Phi} = {\Upsilon} (U, {\tilde U}, {\rz}, {\tz}) \cdot 
  {\psi} \left( \, v_{1}  (U, {\tilde U}, {\rz}, {\tz}), \ldots , v_{N-1} (U, {\tilde U}, {\rz}, {\tz}); {\qe}\, \right)
  \mapsto \\
  {\bf\Psi}^{\rm inst} \left( {\eta}_1, {\eta}_2, \ldots , {\eta}_{N-1}, \frac{\qe}{{\eta}_{1}{\eta}_{2} \ldots {\eta}_{N-1}} \right)
\end{multline}
which consists, in detail, of restricting to ${\pi}_{i} \to  {\pi}_{i}^{0}$~\eqref{eq:pi0}, 
${\tilde\pi}^{i} \to {\tilde\pi}^{i}_{0}$~\eqref{eq:fixed form}, and dropping the factor 
\beq
  {\Upsilon} (U_0, {\tilde U}_{0}, {\rz}, {\tz} ) = 
  \prod_{a=1}^{N} {\tz}_{a}^{{\tilde n}_{a}} ({\rz}^{a})^{n_{a}} \sim 
  {\bf\Psi}^{\rm tree} \cdot \prod_{a=1}^{N}  \left( \frac{{\rz}^{a}}{{\tz}_{a}} \right)^{\delta\mu_{a-1}} \, .
\eeq
\end{rem}


\section{Knizhnik-Zamolodchikov equations}\label{sec KZ}


\subsection{KZ equations}\label{ssec KZ equations}

{}Let us recall the notion of Knizhnik-Zamolodchikov (KZ) equations \cite{KZ} associated with the following data:

\begin{enumerate}

\item[(a)]
$\g$ -- a semisimple Lie algebra, 

\item[(b)]
$t$ -- a non-degenerate $ad$-invariant bilinear form on $\g$, that is:
\begin{equation*}
  t([a,b],c) = t(a, [b,c]) \quad \mathrm{for\ any} \ a,b,c \in \g \, ,
\end{equation*}

\item[(c)] 
$V_1,\ldots,V_n$ -- representations of $\g$,

\item[(d)] 
$\kappa\in \BC^\times$ -- a nonzero constant.

\end{enumerate}

Define the Casimir tensor ${\hat C}\in \g\otimes \g$ and the Casimir element $\Cas\in U(\g)$ via: 
\begin{equation}
\label{eq:Casimir-tensor}
  {\hat C}\, := \sum_{A,B \in I} t^{AB} X_A\otimes X_B 
\end{equation}
and 
\begin{equation}
\label{eq:Casimir-elt}
  \Cas\, := \sum_{A,B\in I} t^{AB} X_A X_B  \, ,
\end{equation}
where $\{X_A\}_{A\in I}$ is a basis of $\g$, $\Vert t^{AB} \Vert$ is the matrix inverse to $\Vert t(X_{A}, X_{B}) \Vert$.

Define the configuration space $\Sigma_n\subset \BC^n$ via: 
\beq
  \Sigma_n: = \Big\{ (p_1,\ldots,p_n)\in \BC^n \, \Big|\, p_i\ne p_j\ \mathrm{for}\ i\ne j \Big\} \, .
\eeq

A function 
  $F\colon \Sigma_n\to V_1\otimes \dots \otimes V_n$ 
is said to satisfy the \emph{KZ equations}~\cite{KZ} if:
\begin{equation}\label{KZ original}
  \kappa\frac{dF}{dp_i}\, +\, \sum_{j\ne i}\frac{{\hat C}_{ij}\cdot F}{p_i-p_j} \, =\, 0 \, , \qquad i=1,\ldots,n \, ,
\end{equation}
where ${\hat C}_{ij}$ denotes\footnote{A more pedantic notation would be:
\begin{equation*}
  {\hat C}_{ij} \, = 
  \sum_{A,B \in I} \, t^{AB} \, {\bf 1}_{V_{1}} \otimes \dots \otimes T_{V_{i}}(X_{A}) \otimes 
  \dots \otimes T_{V_{j}} (X_{B}) \otimes \dots \otimes {\bf 1}_{V_{n}}\, .
\end{equation*}
} 
the action of $\hat C$~\eqref{eq:Casimir-tensor} on the $i$-th and $j$-th factors of $V_1\otimes \dots \otimes V_n$.

\begin{rem}
Note that the KZ equations essentially depend only on the $ad$-invariant form $\frac{t}{\kappa}$.
\end{rem}


\subsection{$\g$-invariance and $n=4$ case}\label{ssec n=4 setup}

{}A function $F\colon \Sigma_n\to V_1\otimes \dots \otimes V_n$ is called $\g$-\emph{invariant} if: 
\beq
  F({\rm p})\in (V_1\otimes\dots\otimes V_n)^\g
  \, , \qquad \forall\, {\rm p}=(p_1,\ldots,p_n)\in \Sigma_n\, .
\eeq
Let $n=4$. Recall the cross-ratio \eqref{eq:crr} of $4$ points,  which  can be thought of as a map:
\begin{equation*}
  \pi\colon \Sigma_4\longrightarrow \BC^\times\, , \qquad 
  {\rm p}=(p_1,p_2,p_3,p_4)\mapsto [p_1,p_2;p_3,p_4]:=\frac{(p_1-p_2)(p_3-p_4)}{(p_1-p_3)(p_2-p_4)}\, .
\end{equation*}
This map can be naturally extended to a map 
  $\bar{\pi}\colon \bar{\Sigma}_4 \to {\mathbb{CP}}^{1}$, 
where $\bar{\Sigma}_4\subset(\mathbb{CP}^1)^4$ is the locus of points with pairwise 
distinct coordinates. The map $\bar{\pi}$ is the quotient map for the natural free action 
of $H = SL(2,{\BC})$ on $\bar{\Sigma}_4$ (the diagonal action by the fractional linear transformations). 
In particular, for any ${\rm p}\in \Sigma_4$ the points ${\rm p}=(p_1,p_2,p_3,p_4)$ 
and $(0, {\qe} = [p_1,p_2;p_3,p_4], 1, \infty)$ of $\bar{\Sigma}_4$ lie in the same $H$-orbit.
Naturally the four KZ equations~\eqref{KZ original} on a $\g$-invariant function $F$ reduce to a single equation 
on a $(V_1\otimes V_2 \otimes V_3 \otimes V_4)^\g$-valued function of $\qe$: 

\begin{prop}\label{reduction}
Assume that the Casimir element $\Cas$~(\ref{eq:Casimir-elt}) acts on $V_i$ as a multiplication by ${\Delta}_{i}\in {\BC}$ 
for any $1\leq i\leq 4$. Choose constants $\{ {\sd}_{ij}\, \vert\, 1\leq i\ne j\leq 4\}$ 
so that $\sd_{ij}=\sd_{ji}$ and $\sum_{j\neq i} \sd_{ij} = \Delta_i$ for any $1\leq i\leq 4$.\footnote{Such $\{\sd_{ij}\}$ 
exist and are unique for an arbitrary choice of $\sd_{12}$ and $\sd_{13}$.}
Then, $F\colon \Sigma_4\to (V_1\otimes V_2\otimes V_3\otimes V_4)^\g$ satisfies all 
four KZ equations~(\ref{KZ original}) if and only if~\footnote{On any simply connected region in 
$\left( {\mathbb{CP}}^{1} \right)^{4} \backslash\, \{{\rm diagonals}\}$.}
\begin{equation}
\label{eq:reduction-to-q}
  F\left( p_1,p_2,p_3,p_4 \right) \, = \,
  \prod_{i<j}(p_i-p_j)^{\frac{\sd_{ij}}{\kappa}}\cdot \, \Phi\Big([p_1,p_2;p_3,p_4]\Big)
\end{equation}
with 
  $\Phi\colon {\BC}^\times\backslash\{1\}\to (V_1\otimes V_2\otimes V_3\otimes V_4)^\g$
satisfying the following equation:
\begin{equation}\label{single KZ}
  \kappa\frac{d\Phi}{d{\qe}} + \left(\frac{\sd_{23}}{{\qe}-1}+\frac{\sd_{12}}{\qe}\right){\Phi}\, +\,
  \left(\frac{{\hat C}_{23}}{{\qe}-1}+\frac{{\hat C}_{12}}{\qe}\right){\Phi}\, =\, 0\, .
\end{equation}
\end{prop}

The proof of this result is elementary.


\subsection{Our KZ setup}\label{ssec KZ-setup}

{} Let us now apply the above discussion to $\g=\ssl_N$ endowed with an $ad$-invariant bilinear form 
$t(a,b)=\mathrm{tr}_{{\BC}^{N}}(ab)$, and the $n=4$ modules $V_i\ (1\leq i\leq 4)$ as in~\eqref{eq:4reps}:
$$
  V_{1} = {\CalV}_{\vec\nu}\, , \ V_{2} = {\CalH}_{\mf}^{\vec\gamma - \vec\nu}\, , \ 
  V_{3}= {\tilde\CalH}_{{\tilde\mf}}^{\vec\gamma - \vec{\tilde\nu}}\, , \ V_{4} = {\tilde{\CalV}}_{\vec{\tilde\nu}}\, . 
$$
According to Lemma~\ref{key property} and the identification~\eqref{weight space realization}, we have:
$$
  \left( V_1 \hat\otimes V_2 \hat\otimes V_3 \hat\otimes V_4 \right)^{\ssl_N} \simeq \,
  \BC\left[ \eta_1^{\pm 1},\ldots,\eta_{N-1}^{\pm 1} \right] \, ,
$$
with $\eta_i$'s defined in~\eqref{eq:eta-variables}. Hence, functions $F$ and $\Phi$ of Proposition~\ref{reduction} 
can be thought of as:
\begin{equation}
\label{eq:FPhi via etas}
  F\colon \Sigma_4\longrightarrow {\BC} \left[ \eta_1^{\pm 1},\ldots,\eta_{N-1}^{\pm 1} \right]
    \, \quad \mathrm{and}\, \quad
  \Phi\colon {\BC}^\times\backslash\{1\}\longrightarrow {\BC} \left[ \eta_1^{\pm 1},\ldots,\eta_{N-1}^{\pm 1} \right]\, .
\end{equation}
Our next goal is to rewrite the equation~\eqref{single KZ} on $\Phi$ as a differential equation in ${\qe},\eta_1,\ldots,\eta_{N-1}$.

   
\subsection{The differential operator $\widehat{H}^{\KZ}$}\label{ssec KZ operator}

{}Choose the basis $\{X_A\}$ of $\g=\ssl_N$ as follows:
\begin{equation*}
  \{X_A\} = \{ {\bJ}_{a}^{b} \, | \, 1 \leq a \ne b \leq N \, \} \,  \sqcup \, \{ \, {\bh}_i\, | \, i = 1, \ldots , N-1 \, \} \, .
\end{equation*}
Then, the Casimir tensor~(\ref{eq:Casimir-tensor}) has the following form: 
\beq
  {\hat C}=\sum_{a\ne b} {\bJ}_{a}^{b}\otimes {\bJ}_{b}^{a} + 
  \sum_{i,j=1}^{N-1} C^{ij} {\bh}_i\otimes {\bh}_j \in \ssl_N\otimes \ssl_N \, ,
\eeq
where $\Vert C^{ij} \Vert$ is the matrix inverse to the Cartan matrix 
$\Vert (2\delta_{i}^{j}-\delta_{i}^{j+1}-\delta_i^{j-1} ) \Vert$ of $\ssl_N$. 
To simplify the calculations, it is convenient to consider a natural embedding 
$\iota\colon \ssl_N\hookrightarrow \gl_N$, so that: 
\beq
\label{eq:sl-to-gl simplification}
  (\iota\otimes\iota)\left( \sum_{i, j=1}^{N-1} C^{ij} {\bh}_i\otimes {\bh}_j \right) \, = \
  \sum_{a=1}^N {\bJ}_{a}^{a}\otimes {\bJ}_{a}^{a} - \frac{1}{N} {\CalC}_{1} \otimes {\CalC}_{1} \, ,
\eeq
where ${\CalC}_1=\sum_{a=1}^N {\bJ}_{a}^a\in \gl_N$ is the first Casimir operator~\eqref{eq:casimirs}. 
Similarly, the image of the Casimir element $\Cas$~\eqref{eq:Casimir-elt} under the induced embedding 
$\iota\colon U(\ssl_N)\hookrightarrow U(\gl_N)$ is given by:
\beq
  (\iota\otimes \iota)(\Cas)= {\CalC}_2 - \frac{{\CalC}_1^2}{N} \, .
\eeq
\medskip
Define 
\beq
  \widehat{H}^{\KZ} = \frac{{\hat C}_{12}}{\qe} + \frac{{\hat C}_{23}}{{\qe}-1} \, .
\label{eq:dkz}
\eeq
The operators 
\beq
\begin{aligned}
  & {\hat C}_{12} = 
    \sum_{a,b=1}^{N} T_{{\CalV}_{\vec\nu}} ({\bJ}_{a}^{b} ) \otimes T_{{\CalH}_{\mf}^{\vec\gamma - \vec\nu}} ({\bJ}_{b}^{a}) + 
    \frac{{\mf} c_{1} ({\CalV}_{\vec\nu})}{N}  \, , \\
  & {\hat C}_{23} = 
    \sum_{a,b = 1}^{N} T_{{\CalH}_{\mf}^{\vec\gamma - \vec\nu}} ({\bJ}_{a}^{b}) \otimes 
    T_{\tilde{\CalH}_{\tilde\mf}^{\vec\gamma - \vec{\tilde\nu}}} ({\bJ}_{b}^{a}) + 
    \frac{{\mf}{\tilde{\mf}}}{N} \end{aligned}
\label{eq:kzres}
\eeq
coincide with ${\hat h}^{\rm cft}_{0}, {\hat h}^{\rm cft}_{1}$ of~\eqref{eq:cft-operators}, respectively, which 
in turn coincide with ${\hat h}^{\rm bps}_{0}, {\hat h}^{\rm bps}_{1}$ of~(\ref{eqn:bps-connection}), according to 
Theorem~\ref{BPSCFT}. This concludes the proof of our main result: 
the vacuum expectation value $\langle\, {\CalS} \, \rangle$ of the surface defect obeys 
the Knizhnik-Zamolodchikov equation~\cite{KZ}, specifically the equation obeyed by the 
$\left( \widehat{\ssl}_{N} \right)_{k}$ current algebra conformal block
\beq
  {\Phi} =  \Big\langle \, 
    {\bf V}_{1} (0) {\bf V}_{2} ({\qe})
    {\bf V}_{3} (1) {\bf V}_{4} ({\infty}) \, 
  \Big\rangle^{\bf a}
\eeq
with the vertex operators at $0$ and $\infty$ corresponding to the generic lowest weight 
${\CalV}_{\vec\nu}$ and highest weight ${\tilde{\CalV}}_{\vec{\tilde\nu}}$ Verma modules, 
while the vertex operators at $\qe$ and~$1$ correspond to the \emph{twisted} HW-modules 
${\CalH}_{\mf}^{\vec \mu}$ and ${\tilde{\CalH}}_{\tilde{\mf}}^{\vec{\tilde\mu}}$.


\section{Conclusions and further directions}

In this paper, we established that the vacuum expectation value of the regular surface defect in $SU(N)$ gauge theory 
in four dimensions with ${\CalN}=2$ supersymmetry, with $2N$ fundamental hypermultiplets, obeys the analytical continuation 
of Knizhnik-Zamolodchikov equation for the four-point conformal block $\left\langle V_1 V_2 V_3 V_4 \right\rangle$ of 
the two-dimensional $\ssl_N$ current algebra at the level
\beq
  k = \frac{{\ve}_{2}}{{\ve}_{1}} - N \, .
\eeq
The surprising feature we discovered is the need to \emph{twist} the irreducible representations
corresponding to the \emph{middle} vertex operators $V_2$ and $V_3$. 

Our result has been anticipated for many years, see~\cite{NN2004}. In particular, in the specific limit 
$m_{i} \to \infty$, ${\qe} \to 0$, with 
\beq
  {\Lambda}^{2N} = {\qe} \prod_{f=1}^{2N} m_{f}
\eeq
the equation~\eqref{eqn:bps-connection} becomes the non-stationary version of the periodic Toda equation:
\beq
  {\kappa} {\Lambda} \frac{\partial}{\partial\Lambda} {\Psi} = 
  \left( \frac 12 \sum_{i=1}^{N} \frac{\partial^2}{\partial x_i^2}+ {\Lambda}^2 \sum_{i=1}^{N} e^{x_{i} -x_{i+1}} \right) {\Psi}\, ,
  \qquad x_{N+1} = x_{1} \, ,
\label{eq:pertoda}
\eeq
where
\beq
  {\qe}_{\omega} m_{\omega}^{+} m_{\omega}^{-} = {\Lambda}^2 e^{x_{\omega+1} - x_{\omega+2}} \, .
\eeq
It was shown in~\cite{Braverman:2004} that the equation~\eqref{eq:pertoda} is obeyed by the $J$-function of the affine flag variety,
which in~\cite{NN2004} was interpreted as the vev of the surface defect in the pure ${\CalN}=2$ super-Yang-Mills theory with 
$SU(N)$ gauge group. However, the method of \cite{Braverman:2004} does not generalize to the theories with matter. 
In \cite{NekBPSCFT} the equations, obeyed by the surface defects of certain quiver gauge theories, were derived. 

In the limit $\ve_1 \to 0$ and/or $\ve_2 \to 0$, the differential operator \eqref{eqn:bps-connection} becomes the equation 
describing certain Lagrangian submanifolds in the complex symplectic manifolds, which are related to the moduli spaces
\cite{Seiberg:1996nz} of vacua of 
the four-dimensional gauge theory we started with, compactified on a circle. These moduli spaces can be also identified
with the moduli space of solutions of some partial differential equations, describing monopoles and instantons
in some auxiliary gauge theory \cite{Cherkis:2000cj, Nekrasov:2003rj, Nekrasov:2012xe, NPS, NRS}. 

In this paper, we studied the simplest case of the asymptotically conformal ${\CalN}=2$ gauge theory, corresponding
to the $A_1$-type quiver. There exist various quiver generalizations, whose Seiberg-Witten geometry can be exactly computed \cite{Nekrasov:2012xe}. The orbifold surface defects of the $A_r$-generalizations conjecturally obey the KZ equations corresponding to the $r+3$-point conformal 
blocks of the $\widehat{\mathfrak{su}(N)}_{k}$ current algebra, with two Verma modules and $r+1$ twisted HW-modules. 
One can also study the intersecting surface defects. For example, in the companion paper \cite{JLN} a $5$-point conformal block corresponding to the infinite-dimensional modules ${\CalV}_{\vec\nu}, {\CalH}^{\vec\mu}_{\mf}, {\tilde\CalH}^{\vec{\tilde\mu}}_{\tilde\mf}, {\tilde\CalV}_{\vec{\tilde\nu}}$, and the $N$-dimensional standard representation is associated
with the intersecting surface defect of the orbifold type studied in this paper, and the orthogonal surface defect
corresponding to the $Q$-observable of gauge theory \cite{NPS, NekBPSCFT}

Perhaps the most interesting continuation of our work would be a translation of the connection between the conformal blocks of two-dimensional current algebra $\left({\widehat{\ssl_N}}\right)_{k}$ to the surface defect partition function of four-dimensional gauge theory that we firmly established, to the $A_{N-1}$ $(0,2)$-theory in six dimensions. 

For integral level $k$ and the weights $\vec\nu, \vec{\tilde\nu}, {\mf}, {\tilde\mf}$ the current algebra conformal blocks have a familiar Chern-Simons interpretation. It can be represented as the path integral in the $SU(N)$ gauge theory on a three-ball $B^3$ with the action 
\beq
\frac{k}{4\pi} \int_{B^3} {\rm Tr} \left( A \wedge d A + \frac{2}{3} A  \wedge A \wedge A \right)
\eeq
with the gauge fields having a curvature singularity along an embedded graph $\Gamma$, as in Fig.~\ref{fig:pic1}.
The edges of the graph are labelled by the conjugacy classes of the monodromy of connection around the small loop linking the edge. We need an extension, or an analytic continuation, to the
case of complex levels and weights. The paper \cite{Witten6} offers such a continuation for the Chern-Simons level. The analytic continuation of Chern-Simons theory in the representation parameters of Wilson and 't Hooft lines is not yet available, but our results strongly suggest it should be possible. 
\begin{figure}
    \centering
    \includegraphics[scale=0.4]{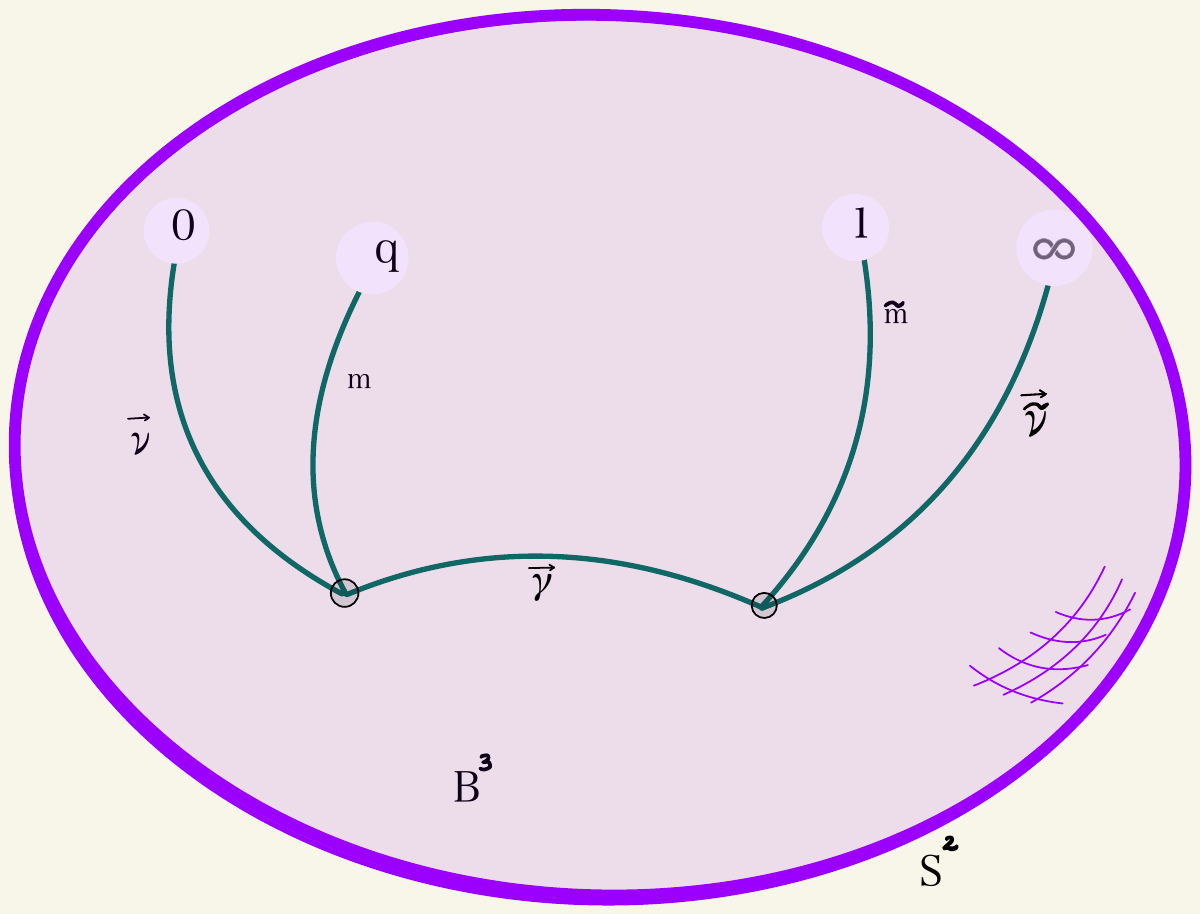}
    \caption{Wilson graph corresponding to the $4$-point conformal block}
    \label{fig:pic1}
\end{figure}
We are familiar with the Wilson line operators $W_{R}(C)$, associated with the representation of the gauge group $G$
and its representation $R$, 
\beq
W_{R}(C) = {\rm Tr}_{R} \, T_{R} \left( P {\rm exp} \oint_{C} A \right) .
\eeq
More generally, a tri-valent orientation graph $\Gamma$, with oriented edges $e$ labelled by representations $R_{e}$, 
with the understanding that the change of the orientation flips the representation $R_{\bar e} = R_{e}^{*}$, 
and vertices labelled by the invariants
\beq
I_{v} \in (R_{e_{1}} \otimes R_{e_{2}} \otimes R_{e_{3}})^{G} 
\eeq
with the edges $e_1, e_2, e_3$ coming out of the vertex $v$, corresponds to the Wilson graph observable
\beq
W_{R_{e}, I_{v}} ({\Gamma} ) = \prod_{l} {\rm Tr}_{R_{l}} \prod_{v} I_{v} \left( \bigotimes\limits_{e} T_{R_{e}} \left( P {\rm exp} \int_{e} A \right) \right)
\label{eq:wgraph}
\eeq
where $l$ labels the loops, i.e., the edges with coinciding ends.

In the case the graph has tails, i.e., $1$-valent vertices, which are placed at the boundary ${\partial} B$, the path integral takes values in the Hilbert space obtained by quantizing the moduli space of flat $G$-connections
on ${\Sigma}^2 = {\partial}B^3$ with singularities at the end-points, with fixed conjugacy classes of monodromies around 
those. In the case of $B^3, {\Sigma}^2 \approx S^2$ this Hilbert space is isomorphic to the space of invariants
in the tensor product of representations attached to the edges ending at the tails. For the graph $\Gamma$ on Fig.~\ref{fig:pic1} this would be
\beq
\left( R_{1} \otimes R_{2} \otimes R_{3} \otimes R_{4} \right)^{G} \, .
\label{eq:r1234}
\eeq
Having the invariants $I_{1} \in \left( R_{1} \otimes R_{2} \otimes R \right)^{G}$, $I_{2} \in 
\left( R^{*} \otimes R_{3} \otimes R_{4} \right)^{G}$ at the two internal vertices of $\Gamma$ identifies the 
conformal block with the channel of the tensor product decomposition \eqref{eq:r1234} corresponding the intermediate
representation $R \in R_{3} \otimes R_{4}$, $R^{*} \in R_{1} \otimes R_{2}$. 

All this, to a limited extent, generalizes to the infinite-dimensional $\mathfrak{g}$-representations, although the 
expression \eqref{eq:wgraph} does not literally make sense. Nevertheless,  the form
\begin{multline}
  {\Upsilon} = I^{(3)}_{{\vec\nu}, {\mf}, {\vec\gamma}} ({\tilde U}', {\rz}, U ) \cdot 
  I^{(2)}_{-\vec\gamma} (U'', {\tilde U}'')
  \cdot {\tilde I}^{(3)}_{{\vec\gamma}, {\tilde\mf}, {\vec{\tilde\nu}}} ({\tilde U}, {\tz}, U')  \, \Biggr\vert_{\rm diag}  \equiv \\
  \prod\limits_{a=1}^{N}\, 
  \left( \frac{\left( {\tilde Z} \wedge {\tilde\Pi}^{a-1} \right)  ({\Pi}_{a}') }{{\tilde\Pi}^{a}({\Pi}_{a}')} \right)^{{\tilde n }_{a}} 
  \left( \frac{{\tilde\Pi}^{'a} \left( {Z} \wedge {\Pi}_{a-1} \right)}{{\tilde\Pi}^{{'}a}({\Pi}_{a})} \right)^{n_{a}} \cdot \, 
  \prod\limits_{i=1}^{N-1} \left( \frac{{\tilde\Pi}^{i} \left( {\Pi}_{i}'\right) \cdot {\tilde\Pi}^{'i} \left( {\Pi}_{i}\right)}{{\tilde\Pi}^{''i} \left( {\Pi}_{i}''\right)} \right)^{\gamma_i} \, \Biggr\vert_{\substack{{\tilde U} = {\tilde U}'={\tilde U}''\\ U = U'=U''}}  
  \label{eq:ups2}
\end{multline}
of our basic invariant $\Upsilon$
\eqref{eq:ups}, and moreover,  
 the ${\qe} \to 0$ asymptotics of the surface defect partition function \eqref{renormalized partition function},  which can be analyzed \cite{Lee:2020hfu} rather explicitly, are suggestive of some sort of three-dimensional interpretation with the graph $\Gamma$, with some
 intermediate $\ssl_N$-module with the highest/lowest/middle weight ${\vec\gamma}$.

It does not seem to be possible to analytically continue \eqref{eq:wgraph} as a line operator in the analytically continued Chern-Simons theory, as in \cite{Witten6}. However, it might be possible to analytically continue the S-dual
't Hooft operator, as a surface defect in the topologically twisted ${\CalN}=4$ theory on a four-dimensional manifold
with corners, which locally looks like $B^3 \times I$. 

{}On the other hand, the surface defect in four dimensions
can be related \cite{NW} to boundary conditions in the two-dimensional sigma model valued in the moduli space
of vacua of the theory, compactified on a circle, which in the present case is believed to be the moduli space ${\CalM}_{N} \left( S^{2} \backslash 4\, {\rm pts} \, ; {\vec\nu}, {\mf}, {\tilde\mf}, {\vec{\tilde\nu}} \right)$
of $SU(N)$ Higgs pairs on a $4$-punctured sphere with the regular punctures at $0$ and $\infty$, and the
minimal punctures at $\qe$ and $1$, see Fig.~\ref{fig:pic2}. 
\begin{figure}
    \centering
    \includegraphics[scale=0.4]{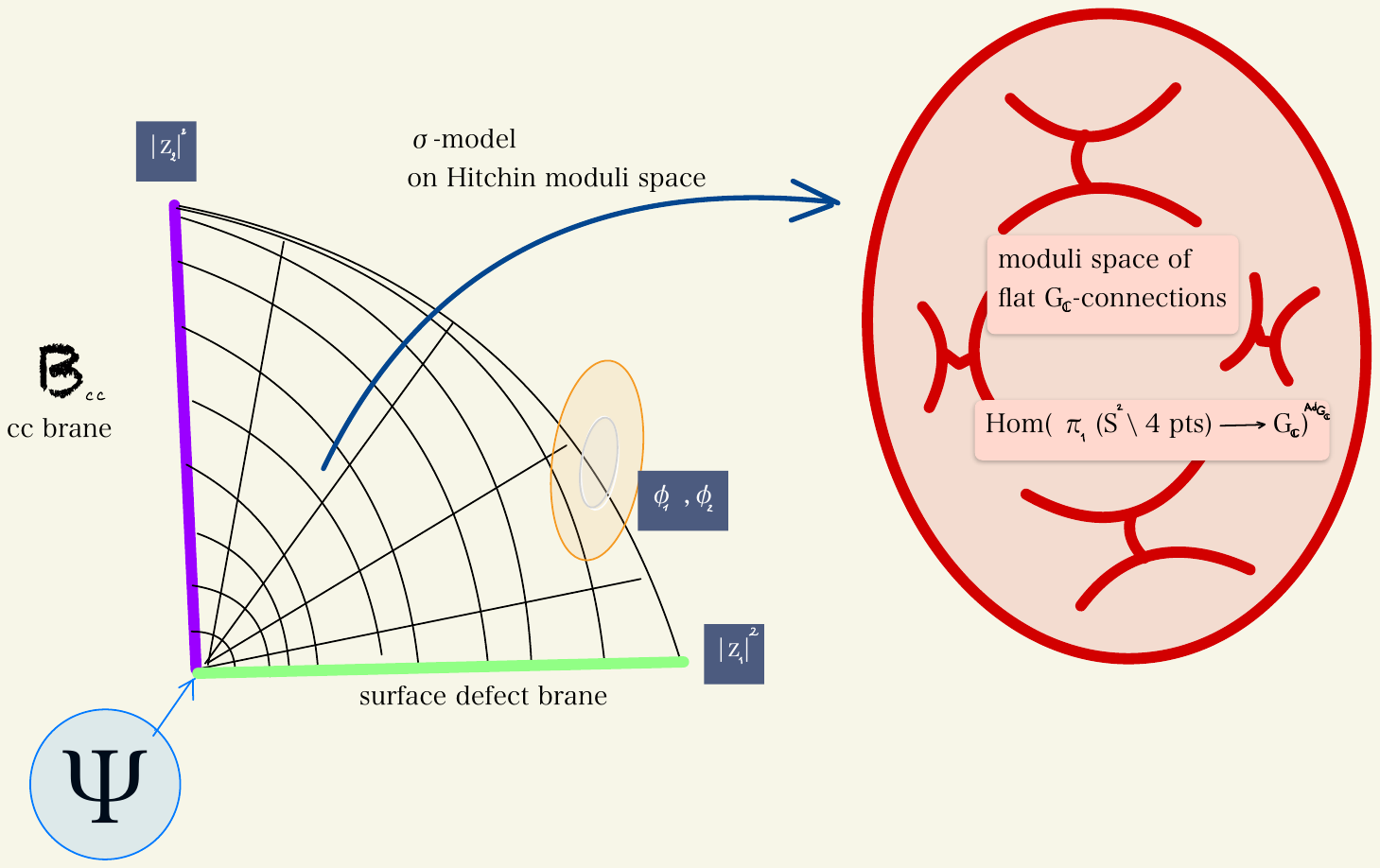}
    \caption{Four-dimensional gauge theory in two-dimensional presentation}
    \label{fig:pic2}
\end{figure}
The homotopy between these two representatives of a cohomology class of an intrinsic operator in the six-dimensional theory proceeds by viewing the two-dimensional sigma model, with the worldsheet $C$ as a long distance limit of the four-dimensional ${\CalN}=2$ $\Omega$-deformed theory
compactified on a two-torus $T^2$ as in \cite{NW}, which, in turn, is a limit of the $A_{N-1}$ $(0,2)$-theory 
compactified on $\left( S^{2} \backslash 4\, {\rm pts}  \right) \times T^2$, which, finally, can be reinterpreted, as
the ${\CalN}=4$ theory on $C \times \left( S^{2} \backslash 4\, {\rm pts}  \right)$. As in \cite{NW}, the canonical parameter \cite{KW} $\Psi$ (not to be confused with the vev of our surface defect) is identified with the ratio $\kappa$ of the $\Omega$-deformation parameters. With $C$ having the topology of the corner ${\BR}_{+}^{2}$, as in Fig.~\ref{fig:pic2}, the ${\CalN}=4$ theory on $C \times \left( S^{2} \backslash 4\, {\rm pts}  \right)$
looks very much like a gradient flow theory of the analytically continued Chern-Simons theory on ${\BR}_{+} \times \left( S^{2} \backslash 4\, {\rm pts}  \right)$, with certain boundary conditions. We plan to discuss this duality
in detail elsewhere.


\appendix
 

\section{Analyticity properties}\label{Appendix A0}

In this Appendix, we provide proofs of the regularity properties from Section~\ref{sec Setup}.
 
\begin{proof}[Proof of Proposition~\ref{regularity 1}]
By inspecting the right-hand side of~\eqref{eq:yobs}, we see that, for generic 
$\ba, \ve_1,\ve_2$, and for any $\overline{\lambda}\in \PP^N$, the rational functions
$\left( Y(x)\,\vert_{\overline{\lambda}} \right)^{\pm 1}$ have only simple poles in $x$.
Moreover, all the poles of $Y(x+{\ve})\,\vert_{\overline{\lambda}}$ and
$\left( Y(x)\,\vert_{\overline{\lambda}} \right)^{-1}$ belong to the set
\begin{equation}
\label{eq:pole set 1}
  \bigsqcup_{1\leq b\leq N} 
  \Big\{ a_{b} + i\cdot\ve_1 + j\cdot\ve_2\, \Big|\, i,j\in\BZ_{\geq 0}\, \Big\}\, .
\end{equation}

Hence, to prove the regularity of $\left\langle\, \CalX(x)\, \right\rangle_{\mu}$, 
it suffices to verify that it has no poles at the above locus~\eqref{eq:pole set 1}. 
Fix $1\leq b\leq N,\, i\geq 0,\, j\geq 0$, and set 
\beq
\label{eq:pole candidate}
  x_0:=a_{b}+i\cdot \ve_1+j\cdot \ve_2.
\eeq
The function $Y(x+\ve)\,\vert_{\overline{\lambda}}$ has a pole at $x=x_0$ iff 
  $\square=(i+1,j+1)\in {\partial}_{-}{\lambda}^{({b})}$,
while the function $\left(Y(x)\,\vert_{\overline{\lambda}}\right)^{-1}$ has a pole at $x=x_0$ iff
  ${\square}=(i+1,j+1)\in {\partial}_{+}{\lambda}^{({b})}$. 
Note that
\begin{equation}
\label{eq:lambda-bijection}
  \overline{\lambda}\mapsto 
  \overline{\lambda}':=\overline{\lambda}\backslash \square^{b}_{(i+1,j+1)}
\end{equation}
(where $\square^{b}_{(i+1,j+1)}$ denotes the $(i,j)$-th box in the $b$-th Young diagram)  
establishes a bijection between the loci of $\overline{\lambda}$ satisfying 
the first condition and the loci of $\overline{\lambda}'$ satisfying the second condition.
Finally, for any $\overline{\lambda}$ from the first locus, 
a straightforward computation shows that:
\begin{equation}
\label{equality of residues}
  \mu\,\vert_{\overline{\lambda}}\cdot \Res_{x=x_0}\,  Y(x+{\ve})\,\vert_{\overline{\lambda}}\, = \, 
  -\, {\qe}\, \cdot \mu\,\vert_{\overline{\lambda}'} \cdot 
  \Res_{x=x_0}\, \left(\frac{P(x)}{Y(x)\,\vert_{\overline{\lambda}'}}\right).
\end{equation}
This completes our proof of the proposition.
\end{proof}

This result admits the following multi-parameter generalization~\cite{NekBPSCFT}:

\begin{prop}\label{regularity 2}
For arbitrary parameters $\bnu=(\nu_1,\ldots,\nu_m)\in \BC^m$, define 
the $\BC(x)$-valued observable $\CalX(x;\bnu)\colon \PP^N\to \BC(x)$ via:
\begin{equation}
\label{multiparameter qq-character}
  \CalX(x;\bnu)\,\vert_{\overline{\lambda}}\ :=
  \sum_{I\sqcup J=\{1,\ldots,m\}}\ {\qe}^{|J|}\, \cdot 
  \prod_{i\in I}^{j\in J}\, R(\nu_i-\nu_j)\cdot \, 
  \prod_{i\in I} Y(x-\nu_i+\ve)\,\vert_{\overline{\lambda}}\cdot \, 
  \prod_{j\in J}\frac{P(x-\nu_j)}{Y(x-\nu_j)\,\vert_{\overline{\lambda}}}\, ,
\end{equation}
where $R(z)=\frac{(z-\ve_1)(z-\ve_2)}{z(z-\ve_1-\ve_2)}$.
Then, the average $\left\langle\, \CalX(x;\bnu)\, \right\rangle_\mu$ is a regular function of $x$.
\end{prop}

As for $m=1$ and $\nu_1=0$, we have $\CalX(x;0)=\CalX(x)$, 
this result generalizes Proposition~\ref{regularity 1}.

\begin{proof}[Proof of Proposition~\ref{regularity 2}]
The proof is similar to the previous one. For generic 
$(\bnu, \ba, \ve_1,\ve_2)$, each summand of~\eqref{multiparameter qq-character} is a 
rational function in $x$ with simple poles, all belonging to the set
\begin{equation}
\label{eq:pole set 2}
  \bigsqcup_{1\leq b\leq N} \Big\{ a_{b}+\nu_r+i\cdot \ve_1+j\cdot \ve_2\, \Big| \,
  1\leq r\leq m\, ,\, i,j\in {\BZ}_{\geq 0}\Big\}\, .
\end{equation}
Moreover, for a fixed quadruple $(b,r,i,j)\in \{1,\ldots,N\}\times \{1,\ldots,m\}\times \BZ_{\geq 0}\times \BZ_{\geq 0}$ 
as in~\eqref{eq:pole set 2}, the $(I,J)$-th summand of 
$\CalX(x;\bnu)\,\vert_{\overline{\lambda}}$~\eqref{multiparameter qq-character} has a pole at
\beq
  x_0:=a_{b}+{\nu}_r+i\cdot \ve_1+j\cdot \ve_2
\eeq
iff either of the following two conditions hold:

\begin{itemize}
\item[(I)] 
$r\in I$ and $\square=(i+1,j+1)\in {\partial}_{-}\lambda^{(b)}$,

\item[(II)]
$r\in J$ and $\square=(i+1,j+1)\in {\partial}_{+} \lambda^{(b)}$.
\end{itemize}
Clearly, the map
\beq
  \Big\{(I,J),\overline{\lambda}\Big\}\mapsto
  \left\{ \Big(I':=I\backslash\{r\}, J':=J\sqcup\{r\}\Big),
      \overline{\lambda}':=\overline{\lambda}\backslash \square^{b}_{(i+1,j+1)}\right\}
\eeq
establishes a bijection between the loci of $\overline{\lambda}$ satisfying the first condition (I) and the loci 
of those satisfying the second condition (II), while a straightforward computation shows that: 
\begin{equation}
  \mu\,\vert_{\overline{\lambda}}\cdot \Res_{x=x_0}\, \CalX(x;\bnu)\,\vert_{\overline{\lambda}}\, = 
  \, -\, \mu\,\vert_{\overline{\lambda}'} \cdot 
  \Res_{x=x_0}\, \CalX(x;\bnu)\,\vert_{\overline{\lambda}'}\, .
\end{equation}
The regularity of $\left\langle\, \CalX(x;\bnu)\, \right\rangle_\mu$ follows.
\end{proof}

Finally, let us prove the analyticity in the orbifold/colored setup.
 
\begin{proof}[Proof of Proposition~\ref{regularity 3}]
It follows immediately from the proof of Proposition~\ref{regularity 1} presented above. 
The key observation is that, while each \emph{non-colored} residue of
$Y(x+\ve)\,\vert_{\overline{\lambda}}$ and $\frac{P(x)}{Y(x)\,\vert_{\overline{\lambda}'}}$ 
at $x=x_0$~\eqref{eq:pole candidate} is a product of elements from the lattice $\Lambda$~\eqref{lattice} 
and their inverses, the corresponding \emph{colored} residues of $Y_{{\omega}+1}(x+\ve)\,\vert_{\overline{\lambda}}$ 
and $\frac{P_{\omega}(x)}{Y_{\omega}(x)\,\vert_{\overline{\lambda}'}}$ at $x=x_0$ are zero 
unless $\mathfrak{S}_{x_0}={\omega}$, while in the latter case they are obtained from 
their \emph{non-colored} counterparts by disregarding all factors from $\Lambda$ 
with a nonzero $\BZ_N$-grading. Likewise, all elements of the lattice $\Lambda$ that appear in 
$\mu^{\rm orb}\,\vert_{\overline{\lambda}}$~\eqref{orbifold measure} are obtained from those that appear in 
$\mu\,\vert_{\overline{\lambda}}$~\eqref{eq:bulkmu} by disregarding all factors from $\Lambda$ 
with a nonzero $\BZ_N$-grading.

Therefore, for each pair $(\overline{\lambda},\overline{\lambda}')$ from the proof of Proposition~\ref{regularity 1},
see~\eqref{eq:lambda-bijection}, we get (cf.~\eqref{equality of residues}):
\begin{equation}
  \mu^{\rm orb}\,\vert_{\overline{\lambda}}\cdot 
  \Res_{x=x_0}\,  \CalX_{\omega}(x)\,\vert_{\overline{\lambda}} \, = \,
  -\, \mu^{\rm orb}\,\vert_{\overline{\lambda}'}\cdot 
  \Res_{x=x_0}\, \CalX_{\omega}(x)\,\vert_{\overline{\lambda}'} \, .
\end{equation}
The regularity of $\left\langle\, \CalX_{\omega}(x) \,\right\rangle_{\mu^{\rm orb}}$ follows.
\end{proof}


\section{Some technical computations}\label{sec:cftcalc}

The following equations are used in the proof of Theorem \ref{BPSCFT}:

\beq
\label{eq:tc1}
\begin{aligned}
 & {\tz}({\rz}) 
 \sum_{a=1}^{N} \left( \, \frac{\partial^2 v_i}{\partial {\rz}^a \partial {\tilde \rz}_a} \, , \,
     \frac{\partial v_i}{\partial {\rz}^a} \frac{\partial v_j}{\partial {\tilde \rz}_{a}} \, , \,
     \frac{\partial^2 {\rm log}{\Upsilon}}{\partial {\rz}^a \partial {\tilde \rz}_a} \, \right) \, = \, 
  \left( \, 1 - N v_i \, , \, v_{i} {\delta}_{i}^{j} - v_{i}v_{j} \, , \, 0 \, \right) \\
 & {\tz}({\rz}) \sum_{a=1}^{N} 
   \left( \, \frac{\partial {\rm log}{\Upsilon}}{\partial {\rz}^a} \frac{\partial v_i}{\partial {\tilde \rz}_{a}} \, , \,
          \frac{\partial {\rm log}{\Upsilon}}{\partial {\tilde \rz}_a} \frac{\partial v_i}{\partial {\rz}^{a}} \, 
          \right) \, = \, 
  \left( \, n_i - \left(\sum_{a=1}^N n_a\right) v_i \, , \, 
         {\tilde n}_i - \left(\sum_{a=1}^N {\tilde n}_a\right) v_i \, \right) \\  
 & {\tz}({\rz}) \sum_{a=1}^{N} 
   \frac{\partial {\rm log}{\Upsilon}}{\partial {\rz}^a} \frac{\partial {\rm log}{\Upsilon}}{\partial {\tz}_a} = 
   \sum_{a=1}^N \frac{n_a {\tilde n}_{a}}{v_a}
\end{aligned}
\eeq
and
\beq
\label{eq:tc2}
\begin{aligned}
  & \sum_{a,b = 1}^{N}  {\rz}^{b} J_{b}^{a} \left( \frac{\partial v_i}{\partial {\rz}^a} \right) = \, 
    v_{i} (  v_{i} + i - 2) + u_{i} (2 v_i - 1) \\
  & \sum_{a,b = 1}^{N}  {\rz}^{b} J_{b}^{a} \left( \frac{\partial {\rm log}{\Upsilon}}{\partial {\rz}^a} \right) = \, 
    \sum_{a=1}^{N} ( a- 1) { n }_{a} \\
  & \sum_{a,b = 1}^{N} {\rz}^{b} J_{b}^{a} ( {\rm log}{\Upsilon}) \frac{\partial v_i}{\partial {\rz}^a} =  
    v_{i} \left( \sum_{j=1}^{i-1} \left( {\gamma}_{j} - n_{j} \right) + 
                 \sum_{j=1}^{N-1} (n_{j} + {\tilde n}_{j} - {\gamma}_{j}) u_{j} \right) -\, {\tilde n}_{i} u_{i} \\
  & \sum_{a,b = 1}^{N} {\rz}^{b} J_{b}^{a} ( v_i ) \frac{\partial ({\rm log}{\Upsilon}) }{\partial {\rz}^a} =    
    \left( \sum_{a =1}^{i-1}  { n }_{a} \right) v_{i} - { n }_{i} u_{i}  \\
  & \sum_{a,b = 1}^{N} {\rz}^{b} J_{b}^{a} ( v_i ) \frac{\partial v_{j}}{\partial {\rz}^a} =    
    v_{i}v_{j} \left(  {\delta}_{j < i} + 2 u_{i} -  1 + v_{i} \right)  - u_{i} v_{i}  {\delta}_{i}^{j}  \\
  & \sum_{a,b = 1}^{N} {\rz}^{b} J_{b}^{a} ({\rm log}{\Upsilon}) \frac{\partial ({\rm log}{\Upsilon})}{\partial {\rz}^a} \ =
    \sum_{1\leq a\leq b\leq N} n_{a}n_{b} \ - \sum_{1\leq a\leq b\leq N-1} n_{a} \gamma_{b} \ - \
    \sum_{a=1}^N n_a {\tilde n}_{a} \frac{u_a}{v_a} 
\end{aligned}
\eeq
with $u_i$'s defined in~\eqref{eq:ucor} and satisfying the equality $v_{i+1}=u_i-u_{i+1}$ of~\emph{loc.cit}.



\begin{thebibliography} {XXX}

\bibitem{Alday:2009aq} 
  L.~Alday, D.~Gaiotto and Y.~Tachikawa,
  \emph{Liouville correlation functions from four-dimensional gauge theories},
  Lett.\ Math.\ Phys.\  {\bf 91} (2010), no.~2, 167--197,
  doi:10.1007/s11005-010-0369-5 
  [hep-th/0906.3219].
  
\bibitem{Alday:2010vg} 
  L.~Alday and Y.~Tachikawa,
  \emph{Affine $SL(2)$ conformal blocks from 4d gauge theories},
  Lett.\ Math.\ Phys.\  {\bf 94} (2010), no.~1, 87--114, 
  doi:10.1007/s11005-010-0422-4
  [hep-th/1005.4469].

\bibitem{Babujian:1993tm} 
  H.~M.~Babujian,
  \emph{Off-shell Bethe ansatz equation and $N$-point correlators in the $SU(2)$ WZNW theory},
  J.\ Phys.\ A {\bf 26} (1993), no.~21, 6981--6990, 
  doi:10.1088/0305-4470/26/23/037
  [hep-th/9307062].
  
\bibitem{BD} 
  A.~Beilinson and V.~Drinfeld, 
  \emph{Opers}, 
  preprint 1993 [math/0501398].

\bibitem{BPZ} 
  A.~Belavin, A.~Polyakov and A.~Zamolodchikov,
  \emph{Infinite conformal symmetry of critical fluctuations in two-dimensions},
  J.\ Statist.\ Phys.\  {\bf 34} (1984), no.~5-6, 763--774, 
  doi:10.1007/BF01009438,\\
  $\sim\sim\sim$, 
  \emph{Infinite conformal symmetry in two-dimensional quantum field theory},
  Nucl.\ Phys.\ B {\bf 241} (1984), no.~2, 333--380, 
  doi:10.1016/0550-3213(84)90052-X.
  
\bibitem{Belavin:1975fg} 
  A.~A.~Belavin, A.~M.~Polyakov, A.~S.~Schwartz and Y.~S.~Tyupkin,
 \emph{Pseudoparticle solutions of the Yang-Mills equations},
  Phys.\ Lett.\ {\bf 59B} (1975), no.~1, 85--87, 
  doi:10.1016/0370-2693(75)90163-X.  
  
\bibitem{Braverman:2004} 
  A.~Braverman, 
  \emph{Instanton counting via affine Lie algebras I: Equivariant $J$-functions of  
        (affine) flag manifolds and Whittaker vectors}, 
  Algebraic structures and moduli spaces, 113--132, CRM Proc.\ Lecture Notes, 38, 
  Amer.\ Math.\ Soc., Providence, RI, 2004 
  [math/0401409].
  
\bibitem{Bullimore:2014awa}
  M.~Bullimore, H.~Kim and P.~Koroteev,
  \emph{Defects and quantum Seiberg-Witten geometry},
  JHEP \textbf{05} (2015), 095,  
  doi:10.1007/JHEP05(2015)095
  [hep-th/1412.6081].
  
\bibitem{Cherkis:2000cj}
  S.~A.~Cherkis and A.~Kapustin,
  \emph{Nahm transform for periodic monopoles and ${\CalN}=2$ super Yang-Mills theory},
  Commun.\ Math.\ Phys.\ \textbf{218} (2001), 333--371, 
  doi:10.1007/PL00005558
  [hep-th/0006050].

\bibitem{Drukker:2009id} 
  N.~Drukker, J.~Gomis, T.~Okuda and J.~Teschner,
  \emph{Gauge theory loop operators and Liouville theory},
  JHEP {\bf 1002} (2010), Paper No.~057,
  doi:10.1007/JHEP02(2010)057
  [hep-th/0909.1105].

\bibitem{EN} 
  B.~Enriquez and N.~Nekrasov, 
  unpublished, 1992.
  
\bibitem{E1}
  P.~Etingof,
  {\em Quantum integrable systems and representations of Lie algebras},
  J.\ Math.\ Phys.\ {\bf 36} (1995), no.~6, 2636--2651  
  [hep-th/9311132].
  
\bibitem{Feigin:1994in} 
  B.~Feigin, E.~Frenkel and N.~Reshetikhin,
  \emph{Gaudin model, Bethe ansatz and critical level},
  Commun.\ Math.\ Phys.\  {\bf 166} (1994), no.~1, 27--62,  
  doi:10.1007/BF02099300
  [hep-th/9402022].
  
\bibitem{Felder}
  G.~Felder and M.~M{\"u}ller-Lennert, 
  \emph{Analyticity of Nekrasov partition functions}, 
  Commun.\ Math.\ Phys.\ {\bf 364} (2018), no.~2, 683--718,
  doi:10.1007/s00220-018-3270-1   
  [math/1709.05232].
  
\bibitem{ChainSaw} 
  M.~Finkelberg and L.~Rybnikov, 
  \emph{Quantization of Drinfeld Zastava in type A}, 
  J.\ Eur.\ Math.\ Soc.\ (JEMS) {\bf 16} (2014), no.~2, 235--271 
  [math/1009.0676]. 
  
\bibitem{FZ} 
  V.~A.~Fateev and A.~B.~Zamolodchikov, 
  \emph{Operator algebra and correlation functions in the two-dimensional Wess-Zumino}
  $SU(2) \times SU(2)$ \emph{Chiral Model}, 
  Sov.\ J.\ Nucl.\ Phys.\ {\bf 43} (1986), 657--664.
  
\bibitem{Frenkel:2015rda} 
  E.~Frenkel, S.~Gukov and J.~Teschner,
  \emph{Surface operators and separation of variables},
  JHEP {\bf 1601} (2016), no.~1, 179, 
  doi:10.1007/JHEP01(2016)179
  [hep-th/1506.07508].
  
\bibitem{Gaiotto:2009we} 
  D.~Gaiotto,
  \emph{N=2 dualities},
  JHEP {\bf 1208} (2012), Paper No.~034, 
  doi:10.1007/JHEP08(2012)034
  [hep-th/0904.2715].
  
\bibitem{AG}
  A.~Gerasimov, 
  \emph{In the case of the WZWN theory with the simple real Lie group $G$ it can be
        understood using the free field representations}, 
  unpublished remarks, 1991.
  
\bibitem{Gorsky:2017hro} 
  A.~Gorsky, B.~Le Floch, A.~Milekhin and N.~Sopenko,
  \emph{Surface defects and instanton-vortex interaction},
  Nucl.\ Phys.\ B {\bf 920} (2017), 122--156, 
  doi:10.1016/j.nuclphysb.2017.04.010
  [hep-th/1702.03330].
  
\bibitem{Gorsky:1995zq} 
  A.~Gorsky, I.~Krichever, A.~Marshakov, A.~Mironov and A.~Morozov,
  \emph{Integrability and Seiberg-Witten exact solution},
  Phys.\ Lett.\ B {\bf 355} (1995), 466--474, 
  doi:10.1016/0370-2693(95)00723-X
  [hep-th/9505035].

\bibitem{GukovWitten} 
  S.~Gukov and E.~Witten,
  \emph{Gauge Theory, ramification, and the geometric Langlands program},
  Current developments in mathematics, 2006, 35--180, Int.\ Press, Somerville, MA, 2008  
  [hep-th/0612073], \\
  $\sim\sim\sim$,   
  \emph{Rigid surface operators},
  Adv.\ Theor.\ Math.\ Phys.\  {\bf 14} (2010), no.~1, 87--177, 
  doi:10.4310/ATMP.2010.v14.n1.a3 
  [hep-th/0804.1561].  
  
\bibitem{Gukov:2008ve}
  S.~Gukov and E.~Witten,
  \emph{Branes and quantization},
  Adv.\ Theor.\ Math.\ Phys.\ {\bf 13} (2009), no.~5, 1445--1518, 
  doi:10.4310/ATMP.2009.v13.n5.a5
  [hep-th/0809.0305].
  
\bibitem{Haouzi:2019jzk}
  N.~Haouzi and C.~Kozçaz,
  \emph{Supersymmetric Wilson Loops, Instantons, and Deformed $W$-Algebras},
  preprint [hep-th/1907.03838].
  
\bibitem{Haouzi:2020yxy}
  N.~Haouzi and J.~Oh,
  \emph{On the quantization of Seiberg-Witten Geometry},
  JHEP {\bf 1} (2021), Paper No.~184, doi:10.1007/JHEP01(2021)184 
  [hep-th/2004.00654].
  
\bibitem{Jeong:2020uxz}
  S.~Jeong and N.~Nekrasov,
  \emph{Riemann-Hilbert correspondence and blown up surface defects}, 
  JHEP {\bf 12} (2020), Paper No.~6, doi:10.1007/JHEP12(2020)006 
  [hep-th/2007.03660].
  
\bibitem{JLN}
  S.~Jeong, N.~Lee and N.~Nekrasov,
  \emph{Intersecting defects in gauge theory, quantum spin chains, and Knizhnik-Zamolodchikov equation}, 
  JHEP {\bf 10} (2021), Paper No.~120, doi:10.1007/JHEP10(2021)120 
  [hep-th/2103.17186]. 
  
\bibitem{KT} 
  H.~Kanno and Y.~Tachikawa, 
  \emph{Instanton counting with a surface operator and the chain-saw quiver},
  JHEP (2011), Paper No.~119, doi:10.1007/JHEP06(2011)119
  [hep-th/1105.0357].

\bibitem{KW}
  A.~Kapustin and E.~Witten,
  \emph{Electric-Magnetic Duality And The Geometric Langlands Program},
  Commun.\ Num.\ Theor.\ Phys.\ \textbf{1} (2007), 1--236, 
  doi:10.4310/CNTP.2007.v1.n1.a1
  [hep-th/0604151].

\bibitem{Kimura:2015rgi}
  T.~Kimura and V.~Pestun,
  \emph{Quiver W-algebras},
  Lett.\ Math.\ Phys.\ \textbf{108} (2018), no.~6, 1351--1381, 
  doi:10.1007/s11005-018-1072-1
  [hep-th/1512.08533].

\bibitem{Klemm:1996bj} 
  A.~Klemm, W.~Lerche, P.~Mayr, C.~Vafa and N.~P.~Warner,
  \emph{Selfdual strings and N=2 supersymmetric field theory},
  Nucl.\ Phys.\ B {\bf 477} (1996), no.~3, 746--764, 
  doi:10.1016/0550-3213(96)00353-7
  [hep-th/9604034].
  
\bibitem{KZ} 
  V.~Knizhnik and A.~Zamolodchikov,
  \emph{Current algebra and Wess-Zumino model in two-dimensions},
  Nucl.\ Phys.\ B {\bf 247} (1984), no.~1, 83--103,
  doi:10.1016/0550-3213(84)90374-2.
  
\bibitem{Lee:2020hfu}
  N.~Lee and N.~Nekrasov,
  \emph{Quantum spin systems and supersymmetric gauge theories. Part I},
  JHEP \textbf{03} (2021), Paper No.~093, 
  doi:10.1007/JHEP03(2021)093
  [hep-th/2009.11199].
  
\bibitem{Losev:2003py} 
  A.~S.~Losev, A.~Marshakov and N.~A.~Nekrasov,
  \emph{Small instantons, little strings and free fermions},
  In, M.~Shifman (ed.) et al.: \emph{From fields to strings, Ian Kogan Memorial volume},
  vol.~1, 581--621 
  [hep-th/0302191].
  
\bibitem{Moore:1988qv} 
  G.~W.~Moore and N.~Seiberg,
  \emph{Classical and Quantum Conformal Field Theory},
  Commun.\ Math.\ Phys.\  {\bf 123} (1989), no.~2, 177--254,
  doi:10.1007/BF01238857.
  
\bibitem{Moore:1997pc} 
  G.~W.~Moore and E.~Witten,
  \emph{Integration over the $u$-plane in Donaldson theory},
  Adv.\ Theor.\ Math.\ Phys.\  {\bf 1} (1997), 298--387, 
  doi:10.4310/ATMP.1997.v1.n2.a7
  [hep-th/9709193].

\bibitem{NakALE} 
  H.~Nakajima, 
  \emph{Moduli spaces of anti-self-dual connections on ALE gravitational instantons},
  Invent.\ Math.\ {\bf 102} (1990), no.~2, 267--303.

\bibitem{NakHilb} 
  H.~Nakajima, 
  \emph{Lectures on Hilbert schemes of points on surfaces}, 
  University Lecture Series, vol.~18, 
  American Mathematical Society, Providence, RI, 1999.
  
\bibitem{Nekrasov:1995nq} 
  N.~Nekrasov,
  \emph{Holomorphic bundles and many body systems},
  Commun.\ Math.\ Phys.\  {\bf 180} (1996), no.~3, 587--603, 
  doi:10.1007/BF02099624
  [hep-th/9503157].

\bibitem{Nekrasov:2002qd} 
  N.~A.~Nekrasov,
  \emph{Seiberg-Witten prepotential from instanton counting},
  Adv.\ Theor.\ Math.\ Phys.\  {\bf 7} (2003), no.~5, 831--864, 
  doi:10.4310/ATMP.2003.v7.n5.a4
  [hep-th/0206161].
  
\bibitem{Nekrasov:2003rj} 
  N.~Nekrasov and A.~Okounkov,
  \emph{Seiberg-Witten theory and random partitions},
  Prog.\ Math.\  {\bf 244} (2006), 525--596, 
  doi:10.1007/0-8176-4467-9$\_$15
  [hep-th/0306238].
  
\bibitem{NekLisbon}
  N.~Nekrasov, 
  \emph{Localizing gauge theories}, 
  14th International Congress on Mathematical Physics 2003, Lisbon, 
  J.C.~Zambrini (Ed.), World Scientific (2005), 645--654.
  
\bibitem{NN2004}
  N.~Nekrasov, 
  \emph{On the BPS/CFT correspondence}, 
  Lecture at the University of Amsterdam string theory group seminar, Feb. 3, 2004, \\
 $\sim\sim\sim$, 
  \emph{2d CFT-type equations from 4d gauge theory}, 
  Lecture at the ``Langlands Program and Physics" conference, 
  IAS, Princeton, March 8-10, 2004.

\bibitem{Zth}
  $\sim\sim\sim$, 
  \emph{{\`A} la recherche de la $M$-theorie perdue. Z theory: Chasing $m/f$-theory},
  Comptes Rendus Physique {\bf 6} (2005), no.~2, 261--269 
  [hep-th/0412021].
  
\bibitem{NekBPSCFT}
  $\sim\sim\sim$, 
  \emph{Non-Perturbative Schwinger-Dyson equations: from BPS/CFT correspondence 
        to the novel symmetries of quantum field theory}, 
  Phys.-Usp.\ {\bf 57} (2014), 133--149, doi:10.1142/9789814616850$\_$0008, \\
  $\sim\sim\sim$, 
  \emph{BPS/CFT correspondence: non-perturbative Dyson-Schwinger equations and $qq$-characters}, 
  JHEP {\bf 1603} (2016), Paper No.~181, doi:10.1007/JHEP03(2016)181
  [hep-th/1512.05388],\\
  $\sim\sim\sim$, 
  \emph{BPS/CFT correspondence II: instantons at crossroads, moduli and compactness theorem}, 
  Adv.\ Theor.\ Math.\ Phys.\ {\bf 21} (2017), no.~2, 503--583,
  [hep-th/1608.07272], \\
  $\sim\sim\sim$, 
  \emph{BPS/CFT Correspondence III: Gauge Origami partition function and $qq$-characters}, 
  Commun.\ Math.\ Phys.\ {\bf 358} (2018), no.~3, 863--894, 
  [hep-th/1701.00189],\\
  $\sim\sim\sim$, 
  \emph{BPS/CFT correspondence IV: sigma models and defects in gauge theory}, 
  Lett.\ Math.\ Phys.\ {\bf 109} (2019), no.~3, 579--622, 
  [hep-th/1711.11011],\\
  $\sim\sim\sim$, 
  \emph{BPS/CFT correspondence V: BPZ and KZ equations from $qq$-characters}, 
  preprint [hep-th/1711.11582]. 

\bibitem{NikBlowup} 
  N.~Nekrasov, 
  \emph{Blowups in BPS/CFT correspondence, and Painlev{\'e} VI}, 
  preprint [hep-th/2007.03646].
  
  
\bibitem{Nekrasov:2012xe}
  N.~Nekrasov and V.~Pestun,
  \emph{Seiberg-Witten geometry of four dimensional ${\CalN}=2$ quiver gauge theories},
  preprint [hep-th/1211.2240].

\bibitem{NPS}
  N.~Nekrasov, V.~Pestun and S.~Shatashvili,
  \emph{Quantum geometry and quiver gauge theories},
  Commun.\ Math.\ Phys.\ \textbf{357} (2018), no.~2, 519--567, 
  doi:10.1007/s00220-017-3071-y
  [hep-th/1312.6689].
  
\bibitem{NRS}   
  N.~Nekrasov, A.~Rosly and S.~Shatashvili,
  \emph{Darboux coordinates, Yang-Yang functional, and gauge theory},
  Nucl.\ Phys.\ Proc.\ Suppl.\  {\bf 216} (2011), 69--93, 
  doi:10.1016/j.nuclphysbps.2011.04.150
  [hep-th/1103.3919].
 
\bibitem{NW}
  N.~Nekrasov and E.~Witten, 
  \emph{The omega deformation, branes, integrability, and Liouville theory},
  JHEP {\bf 1009} (2010), Paper No.~092, 
  doi:10.1007/JHEP09(2010)092
  [hep-th/1002.0888].  
  
\bibitem{OP}
  M.~Olshanetsky, A.~Perelomov, 
  \emph{Explicit solution of the Calogero model in the classical case and geodesic flows on symmetric spaces of zero curvature},  
  Lett.\ Nuovo Cimento (2) {\bf 16} (1976), no.~11, 333--339.
  
\bibitem{Pestun:2007rz} 
  V.~Pestun,
  \emph{Localization of gauge theory on a four-sphere and supersymmetric Wilson loops},
  Commun.\ Math.\ Phys.\  {\bf 313} (2012), no.~1, 71--129, 
  doi:10.1007/s00220-012-1485-0
  [hep-th/0712.2824].
  
\bibitem{PBW}
  H.~Poincar{\'e}, 
  \emph{Sur les groupes continus}, 
  Trans.~Cambr.~Philos.\ Soc.\ {\bf 18} (1900), 220--225,\\
  E.~Witt, 
  \emph{Treue Darstellung Liescher Ringe}, 
  J.~Reine Angew.\ Math.\ {\bf 177} (1937), 152--160,\\
  G.~D.~Birkhoff, 
  \emph{Representability of Lie algebras and Lie groups by matrices}, 
  Ann.\ of Math.\ {\bf 38} (1937), no.~2, 526--532.
  
\bibitem{Seiberg:1996nz}
  N.~Seiberg and E.~Witten,
  \emph{Gauge dynamics and compactification to three-dimensions},
  preprint [hep-th/9607163].
  
\bibitem{Vafa:1994tf} 
  C.~Vafa and E.~Witten,
  \emph{A strong coupling test of S-duality},
  Nucl.\ Phys.\ B {\bf 431} (1994), no.~1-2, 3--77, 
  doi:10.1016/0550-3213(94)90097-3 
  [hep-th/9408074]. 
  
\bibitem{Witten5} 
  E.~Witten, 
  \emph{Solutions of four-dimensional field theories via $M$-theory},
  Nucl.\ Phys.\ B {\bf 500} (1997), no.~1-3, 3--42, 
  doi:10.1016/S0550-3213(97)00416-1
  [hep-th/9703166].
  
\bibitem{Witten6} 
  E.~Witten, 
  \emph{Analytic continuation Of Chern-Simons theory}, 
  Chern-Simons gauge theory: 20 years after, 347--446, AMS/IP Stud.\ Adv.\ Math., 
  50, Amer.\ Math.\ Soc., Providence, RI, 2011 
  [hep-th/1001.2933].
  
\bibitem{Witten:2011zz}
  E.~Witten,
  \emph{Fivebranes and knots},
  Quantum Topol.\ {\bf 3} (2012), no.~1, 1--137, doi:10.4171/QT/26 
  [hep-th/1101.3216].

\bibitem{Wyllard:2009hg}
  N.~Wyllard,
  \emph{$A_{N-1}$ conformal Toda field theory correlation functions from conformal 
        $N = 2$ SU(N) quiver gauge theories},
  JHEP \textbf{11}, 002 (2009), 
  doi:10.1088/1126-6708/2009/11/002
  [hep-th/0907.2189].

\bibitem{ZW}
  A.~B.~Zamolodchikov, 
  \emph{Infinite extra symmetries in two-dimensional conformal quantum field theory},
  Theor.\ Math.\ Physics (in Russian), {\bf 65} (1985), no.~3, 347--359, 
  ISSN 0564-6162, MR 0829902.
  
\bibitem{Zamolodchikov:1995aa} 
  A.~B.~Zamolodchikov and A.~B.~Zamolodchikov,
  \emph{Conformal bootstrap in Liouville field theory},
  Nucl.\ Phys.\ B {\bf 477}, (1996), no.~2, 577--605, 
  doi:10.1016/0550-3213(96)00351-3 
  [hep-th/9506136].  
  
\end{thebibliography}
\end{document}